\documentclass[11pt,letterpaper,reqno]{amsart}
\usepackage[leqno]{amsmath}
\usepackage{amsmath,amsfonts,amsthm,amssymb}
\usepackage{epsfig}
\usepackage{tikz}
\usepackage{graphicx}
\usepackage[labelfont=rm]{subcaption}
\usepackage{cite,url}
\usepackage{hyperref}
\usepackage[left=1in,right=1in,top=1in,bottom=1in,foot=0.5in]{geometry}
\usepackage{amsthm}
\usepackage[T1]{fontenc}
\usepackage[utf8]{inputenc}
\usepackage{tikz}
\usepackage{float}
\usepackage{fancyhdr}
\usepackage[linesnumbered,ruled,noend,nofillcomment]{algorithm2e}
\usepackage{enumerate}
\usepackage{xcolor}
\usepackage{xspace}
\usepackage{todonotes}

\definecolor{mygreen}{RGB}{0,180,0}

\SetCommentSty{commcolor}
\SetArgSty{textnormal}

\SetKwInput{KwData}{Input}
\SetKwInput{KwResult}{Output}
\let\Input\KwData
\let\Output\KwResult
\SetKw{Stop}{stop}

\newcommand{\distenc}[1][alg_distenc]{\textsc{Distance\_Encoding}}
\newcommand{\distdec}[1][alg_distdec]{\textsc{Distance}}
\newcommand{\encStar}{\texttt{Enc\_Star}}
\newcommand{\distDecStar}{\texttt{Dist\_Star}}
\newcommand{\routDecStar}{\texttt{Rout\_Star}}
\newcommand{\distEncTree}{\texttt{Dist\_Enc\_Tree}}
\newcommand{\distDecTree}{\texttt{Dist\_Tree}}
\newcommand{\routEncTree}{\texttt{Rout\_Enc\_Tree}}
\newcommand{\routDecTree}{\texttt{Rout\_Tree}}
\newcommand{\routenc}[1][alg_routenc]{\textsc{Routing\_Encoding}}
\newcommand{\routdec}[1][alg_routdec]{\textsc{Routing}}
\DeclareMathOperator{\depth}{depth}

\SetKwProg{Fn}{function}{:}{}
\newenvironment{myFStyle}{}{}
\newenvironment{myFunction}{
\smallskip
\begin{myFStyle}
    \RestyleAlgo{plain}
    \LinesNotNumbered
    \SetAlgoNoLine
    \begin{algorithm}[H]
    }{
    \end{algorithm}
\end{myFStyle}
}
\SetKwFunction{DistanceSeparated}{Distance\_Separated}
\SetKwFunction{RoutingSeparated}{Routing\_Separated}
\SetKwFunction{DistanceANeighboring}{Distance\_2-Neighboring}
\SetKwFunction{RoutingANeighboring}{Routing\_2-Neighboring}
\SetKwFunction{DistanceNeighboring}{Distance\_1-Neighboring}
\SetKwFunction{RoutingNeighboringPtoC}{Routing\_Panel\_to\_Cone}
\SetKwFunction{RoutingNeighboringCtoP}{Routing\_Cone\_to\_Panel}
\SetKwFunction{RoutingConetoCone}{Routing\_Cone\_to\_Cone}

\theoremstyle{plain}
\newtheorem{theorem}{Theorem}[section]

\newtheorem{lemma}{Lemma}
\newtheorem{proposition}{Proposition}
\newtheorem{question}{Question}
\newtheorem{corollary}{Corollary}
\newtheorem{remark}[theorem]{Remark}

\newtheorem{claim}[theorem]{Claim}

\theoremstyle{definition}
\newtheorem{example}[theorem]{Example}

\newcommand{\separated}{separated\xspace}
\newcommand{\neighboring}{1-neighboring\xspace}
\newcommand{\aNeighboring}{2-neighboring\xspace}
\newcommand{\roommates}{roommates\xspace}

\newcommand{\cone}{cone\xspace}						
\newcommand{\fiber}{fiber\xspace}	 				
\newcommand{\panel}{panel\xspace}					 
\newcommand{\cones}{\cone{}s\xspace}
\newcommand{\panels}{\panel{}s\xspace}
\newcommand{\fibers}{\fiber{}s\xspace}

\newcommand{\Cone}{\expandafter\MakeUppercase\cone}
\newcommand{\Style}{\expandafter\MakeUppercase\panel}
\newcommand{\Fiber}{\expandafter\MakeUppercase\fiber}

\DeclareMathOperator{\fib}{fib}

\DeclareMathOperator{\tbound}{tbd}
\DeclareMathOperator{\imp}{imp}
\DeclareMathOperator{\gate}{gate}
\DeclareMathOperator{\neighbor}{neighbor}

\DeclareMathOperator{\closest}{proj}
\DeclareMathOperator{\pred}{pred}

\newcommand{\ceil}[1]{\lceil #1 \rceil}

\newcommand{\NN}{\mathbb N}

\newcommand{\mcg}{\mathcal G}

\newcommand{\dist}{d}
\DeclareMathOperator{\St}{St}

\newcommand{\eproj}[2]{{\Upsilon}{(#1,#2)}}

\DeclareMathOperator{\twin}{twin}
\DeclareMathOperator{\degree}{deg}
\DeclareMathOperator{\port}{port}

\newcommand{\LD}{\text{LD}} 
\newcommand{\LR}{\text{LR}} 
\newcommand{\conc}{\circ}

\newcommand{\id}{\text{id}}

\newcommand{\dls}{\text{DLS}\xspace}
\newcommand{\rls}{\text{RLS}\xspace}

\newcommand{\Lb}[2]{\text{L}_{#2}(#1)}
\newcommand{\LDt}[2]{\text{LD}_{#2}(#1)}
\newcommand{\LRt}[2]{\text{LR}_{#2}(#1)}
\newcommand{\lleft}{\text{1st}\xspace}
\newcommand{\rright}{\text{2nd}\xspace}
\newcommand{\st}{\text{St}\xspace}
\newcommand{\med}{\text{Cent}\xspace}
\newcommand{\distance}{\text{Dist}\xspace}
\newcommand{\rootB}{\text{gate}\xspace}
\newcommand{\gateD}{\text{gate\_LDT}\xspace}
\newcommand{\imprintD}{\text{imp\_LDT}\xspace}
\newcommand{\gateR}{\text{gate\_{LRT}}\xspace}
\newcommand{\imprintR}{\text{imp\_{LRT}}\xspace}
\newcommand{\toGate}{\text{toGate}\xspace}
\newcommand{\toImprint}{\text{toImp}\xspace}

    \newcommand{\fromGate}{\text{fromGate}\xspace}
\newcommand{\toMed}{\text{toCent}\xspace}
\newcommand{\fromMed}{\text{fromCent}\xspace}

\title{Distance and routing labeling schemes for cube-free median 
graphs\footnote{The short version of this paper appeared at MFCS 2019}}
\date{}

\begin{document}

\centerline{\Large\bf Distance and routing labeling schemes for cube-free
median graphs\footnote{The short version of this paper appeared at MFCS 2019}}

\vspace{10mm}
\centerline{Victor Chepoi, Arnaud Labourel, and S\'ebastien Ratel}

\medskip
\begin{small}
    \medskip
    \centerline{Aix Marseille Univ, Université de Toulon, CNRS, LIS, Marseille, France}

    \centerline{\texttt{\{victor.chepoi, arnaud.labourel,
    sebastien.ratel\}@lis-lab.fr}}
\end{small}

\bigskip\bigskip\noindent
{\footnotesize {\bf Abstract.}
    Distance labeling schemes are schemes that label the vertices of a graph
    with short labels in such a way that the distance between any two vertices
    $u$ and $v$ can be determined efficiently by merely inspecting the labels
    of $u$ and $v$, without using any other information. Similarly, routing
    labeling schemes label the vertices of a graph in a such a way that given
    the labels of a source node and a destination node, it is possible to
    compute efficiently the port number of the edge from the source that heads
    in the direction of the destination. One of important problems is finding
    natural classes of graphs admitting distance and/or routing labeling
    schemes with labels of polylogarithmic size.  In this paper, we show that
    the class of cube-free median graphs on $n$ nodes enjoys distance and
    routing labeling schemes with labels of $O(\log^3 n)$ bits.
}
\sloppy

\section{Introduction}
\label{sect_intro}

Classical network representations are usually global in nature. In order to 
derive a useful piece of information, one must access to a global data 
structure representing the entire network even if the needed information only
concerns few nodes. Nowadays, with networks getting bigger and bigger, the need
for locality is more important than ever. Indeed, in several cases, global
representations are impractical and network representation must be distributed.
The notion of (distributed) labeling scheme has been introduced
\cite{breuer1967unexpected,kannan1992implicat,Peleg00,Peleg05,GavoillePPR04}
in order to meet this need. A \emph{(distributed) labeling scheme} is a scheme
maintaining global information on a network using local data structures (or
labels) assigned to nodes of the network. Their goal is to locally store some
useful information about the network in order to answer a specific query
concerning a pair of nodes by only inspecting the labels of the two nodes.
Motivation for such localized data structure in distributed computing is
surveyed and widely discussed in \cite{Peleg00}. The predefined queries can
be of various types such as distance, adjacency, or routing.  The quality of a
labeling scheme is measured by the size of the labels of nodes and the time
required to answer queries. Trees with $n$ vertices admit adjacency and routing
labeling schemes with size of labels and query time $O(\log n)$ and distance
labeling schemes with size of labels and query time $O(\log^2n)$, and this is
asymptotically optimal. Finding natural classes of graphs admitting distance
labeling schemes with labels of polylogarithmic size is an
important and challenging problem.

A connected graph $G$ is {\it median} if any triplet of vertices $x,y,z$
contains a unique vertex
simultaneously lying on shortest $(x,y)$-, $(y,z)$-, and $(z,x)$-paths.
Median graphs constitute the most important class in metric graph theory
\cite{BaCh_survey}. This importance is explained by the bijections between
median graphs and discrete structures arising and playing important roles in
completely different areas of research in mathematics and theoretical computer
science:  in fact, median graphs, 1-skeletons of CAT(0) cube complexes from
geometric group theory \cite{Gr,Sa_survey}, domains of event structures from
concurrency \cite{WiNi}, median algebras from universal algebra \cite{BaHe},
and solution sets of 2-SAT formulae from complexity theory
\cite{MulderSch,Scha} are all the same. In this paper, we design a distance
labeling scheme for median graphs
containing no cubes. In our scheme, the labels have $O(\log^3 n)$ bits and
$O(1)$ query time. Our constant query time assumes the standard word-RAM model
with word size $\Omega(\log n)$.

We continue with the idea of the labeling scheme. It generalizes the distance 
labeling scheme for trees proposed by Peleg in \cite{Peleg00} and our work can 
be viewed in a sense as an answer to the question  ``How far can we take the 
scheme of Peleg?''. 
Let $G=(V,E)$ be a cube-free median graph with $n$ vertices. First, the 
algorithm computes a centroid (median) vertex $c$ of $G$ and the star $\St(c)$ 
of $c$ (the union of all edges and squares of $G$ incident to $c$). The star 
$\St(c)$ is gated, i.e., each vertex of $G$ has a unique projection (nearest 
vertex) in $\St(c)$. Therefore, with respect to the projection function, the 
vertex-set of $G$ is partitioned into fibers: the fiber $F(x)$ of $x\in \St(c)$ 
consists of all vertices $v\in V$ having $x$ as the projection in $\St(c)$. 
Since $c$ is a centroid of $G$, each fiber contains at most $\frac{n}{2}$ 
vertices. The fibers are also gated and are classified into panels and cones 
depending on the distance between their projections and $c$ (one for panels and 
two for cones).
Each cone has at most two neighboring panels however a panel may have an
unbounded number of neighboring cones.
Given two arbitrary vertices $u$ and $v$ of $G$, we show that
$d_G(u,v)=d_G(u,c)+d_G(c,v)$ for all locations of $u$ and $v$ in the fibers of
$\St(c)$ except the cases when $u$ and $v$ belong to neighboring cones and
panels, or $u$ and $v$ belong to two cones neighboring the same panel, or $u$
and $v$ belong to the same fiber. If $d_G(u,v)=d_G(u,c)+d_G(c,v)$, then
$d_G(u,v)$ can be retrieved by keeping $d_G(u,c)$ in the label of $u$ and
$d_G(v,c)$ in the label of $v$. If $u$ and $v$ belong to the same fiber $F(x)$,
the computation of $d_G(u,v)$ is done by recursively partitioning the cube-free
median graph $F(x)$ at a later stage of the recursion. In the two other cases, 
we show that $d_G(u,v)$ can be retrieved by keeping in the labels of vertices 
in all cones the distances to their projections on the two neighboring panels. 
It turns out (and this is the main technical contribution of the paper), that 
for each panel $F(x)$, the union of all projections of vertices from 
neighboring cones on $F(x)$ is included in an isometric tree of $G$ and that 
the vertices of the panel $F(x)$ contain one or two projections in this tree. 
All such outward and inward projections are kept in the labels of respective 
vertices. Therefore, one can use distance labeling schemes for trees to deal 
with vertices $u$ and $v$ lying in neighboring fibers or in cones having a
common neighboring panel. Consequently, the size of the label of a vertex $u$ 
on each recursion level is $O(\log^2 n)$. Since the recursion depth is $O(\log 
n)$, the vertices of $G$ have labels of size $O(\log^3 n)$. The distance 
$d_G(u,v)$ can be retrieved by finding the first time in the recursion  when 
vertices $u$ and $v$ belong to different fibers of the partition. Consequently, 
the main result of the paper is the following theorem:

\begin{theorem} 
    \label{thm_dist_labeling}
    There exists a distance labeling scheme that constructs in $O(n\log n)$ 
    time labels of size $O(\log^3 n)$  of the vertices of a cube-free median 
    graph $G=(V,E)$.
    Given the labels of two vertices $u$ and $v$ of $G$, it computes in constant
    time the distance $\dist_G(u,v)$ between $u$ and $v$.
\end{theorem}

With the same ideas, it is possible to adapt our technique to design a routing 
labeling scheme.

\begin{theorem}
    \label{thm_rout_labeling}
    There exists a routing labeling scheme that constructs in $O(n\log n)$ time 
    labels of size $O(\log^3 n)$ of the vertices of a cube-free median graph 
    $G=(V,E)$. Given the labels of two vertices $u$ and $v$, it computes in 
    constant time a port of $u$ to a neighbor of $u$ on a shortest path to $v$.
\end{theorem}

\section{Preliminaries}
\label{sect_prelim}

\subsection{Basic notions}
\label{sect_basics}

In this subsection, we recall some basic notions from graph theory. All graphs
$G=(V,E)$ occurring in this note are undirected, simple, and connected. In our
algorithmic results we will also suppose that they are finite.
The \emph{closed neighborhood} of a vertex $v$ is denoted by $N[v]$ and 
consists of $v$ and the vertices adjacent to $v$. The \emph{(open) 
neighborhood} $N(v)$ of $v$ is the set $N[v]\setminus \{ v\}$.  The 
\emph{degree} $\degree(v)$ of $v$ is the number of vertices in its open
neighborhood. We will write $u\sim v$ if two vertices $u$ and $v$ are adjacent 
and $u\nsim v$ if $u$ and $v$ are not adjacent. We will denote by $G[S]$ the 
subgraph of $G$ induced by a subset of vertices $S$ of $V$. If it is clear from
the context, we will use the same notation $S$ for the set $S$ and the 
subgraph  $G[S]$ induced by $S$.

The \emph{distance} $d_G(u,v)$ between two vertices $u$ and $v$ is the length of
a shortest $(u,v)$-path, and the \emph{interval} $I(u,v)$ between $u$ and $v$
consists of all the vertices on shortest $(u,v)$--paths, that is, of all 
vertices (metrically) \emph{between} $u$ and $v$:
$I(u,v):=\{ x\in V: d_G(u,x)+d_G(x,v)=d_G(u,v)\}$.
A subgraph $H$ of a graph $G=(V,E)$ is called an \emph{isometric subgraph} of 
$G$ if $d_H(u,v)=d_G(u,v)$ for any two vertices $u,v$ of $H$, i.e., any pair of 
vertices of $H$ can be connected inside $H$ by a shortest path of $G$.
A subgraph $H=(S,E')$ of $G$ (or the corresponding vertex set $S$) is called
\emph{convex} if it includes the interval of $G$ between any pair of $H$'s
vertices, i.e., if for any pair $u,v$ of vertices of $H$ all shortest 
$(u,v)$-paths of $G$ are included in $H$.
A \emph{halfspace} of $G$ is a convex subset $S$ with convex complement $V\setminus S$.
The \emph{distance from a vertex $v$ to a subgraph} $H$ of $G$ is $\dist_G(v,H)=\min \{ \dist_G(v,x): x\in V(H)\}$.
A subgraph $H$ of $G$ is said to be \emph{gated} if for every vertex $v \notin V(H)$, there
exists
a vertex $v' \in V(H)$ such that for all $u \in V(H)$, $d_G(v,u) = d_G(v,v')
+ d_G(v',u)$ ($v'$ is called the \emph{gate} of $v$ in $H$). Therefore, $v'$ is the gate of $v$ in $H$ if for any vertex $u$ of $H$, there exists a shortest $(u,v)$-path passing via $v'$.
For a vertex $x$ of a gated subgraph $H$ of
$G$, the set (or the subgraph induced by this set) $F(x)=\{ v\in V: x \mbox{ is the gate of } v \mbox{ in } H \}$
is called the {\it fiber} of $x$ with respect to $H$. From the definition it follows that the
fibers $\{ F(x): x\in H\}$ define a partition
of the vertex set of $G$. Notice also that gated sets of a graph enjoy
the finite \emph{Helly property}, that is, every finite family of gated sets
that pairwise intersect has a nonempty intersection.

The \emph{$m$-dimensional hypercube $Q_m$} is the graph whose vertex-set
consists of all subsets of an $m$-set $X := \{1,\ldots,m\}$ and in which
two vertices $A$ and $B$ are linked by an edge if and only if $|A \triangle B|= 1$.

A graph $G$ is called \emph{median} if the intersection $I(x,y)\cap
I(y,z)\cap I(z,x)$ is a singleton for each triplet $x,y,z$ of vertices.
The unique vertex $m(x,y,z) \in I(x,y)\cap I(y,z)\cap I(z,x)$ is called the
\emph{median} of $x,y,z$.  Median graphs are bipartite.
Basic examples of median graphs are trees, hypercubes, rectangular grids, and
Hasse diagrams of distributive lattices and  of median semilattices 
\cite{BaCh_survey}. The {\it star} $\St(z)$ of a vertex $z$ of a median graph
$G$ is the union of all hypercubes of $G$ containing $z$. If $G$ is a tree
and $z$ has degree $r$, then $\St(z)$ is the closed neighborhood of $z$ and is 
isomorphic to $K_{1,r}$. The {\it dimension} $\dim(G)$ of a median graph $G$ is 
the largest dimension of a hypercube of $G$.

A {\it cube-free median graph} is a median graph $G$ of dimension $2$, i.e.,
a median graph not containing 3-cubes as isometric subgraphs. Two illustrations
of cube-free median graphs
are given in Figure \ref{fig_cube-free}. The left figure will be used as a 
running example to illustrate the main definitions. Even if cube-free median 
graphs are the skeletons of 2-dimensional CAT(0) cube complexes, their 
combinatorial structure is rather intricate. For example, cube-free median 
graphs are not necessarily planar: for this, take the Cartesian product 
$K_{1,n}\times K_{1,m}$ of the stars $K_{1,n}$ and $K_{1,m}$ for  $n,m\ge 5$. 
Moreover, they may contain any complete graph $K_n$ as a minor.
Cube-free median graphs and their square complexes have been previously studied 
in \cite{BaChEp,BrKlSk,ChHa,ChMa}.

\begin{figure}[ht]
    \begin{minipage}{0.4\linewidth}
        \centering
        \includegraphics[width=0.7\linewidth]{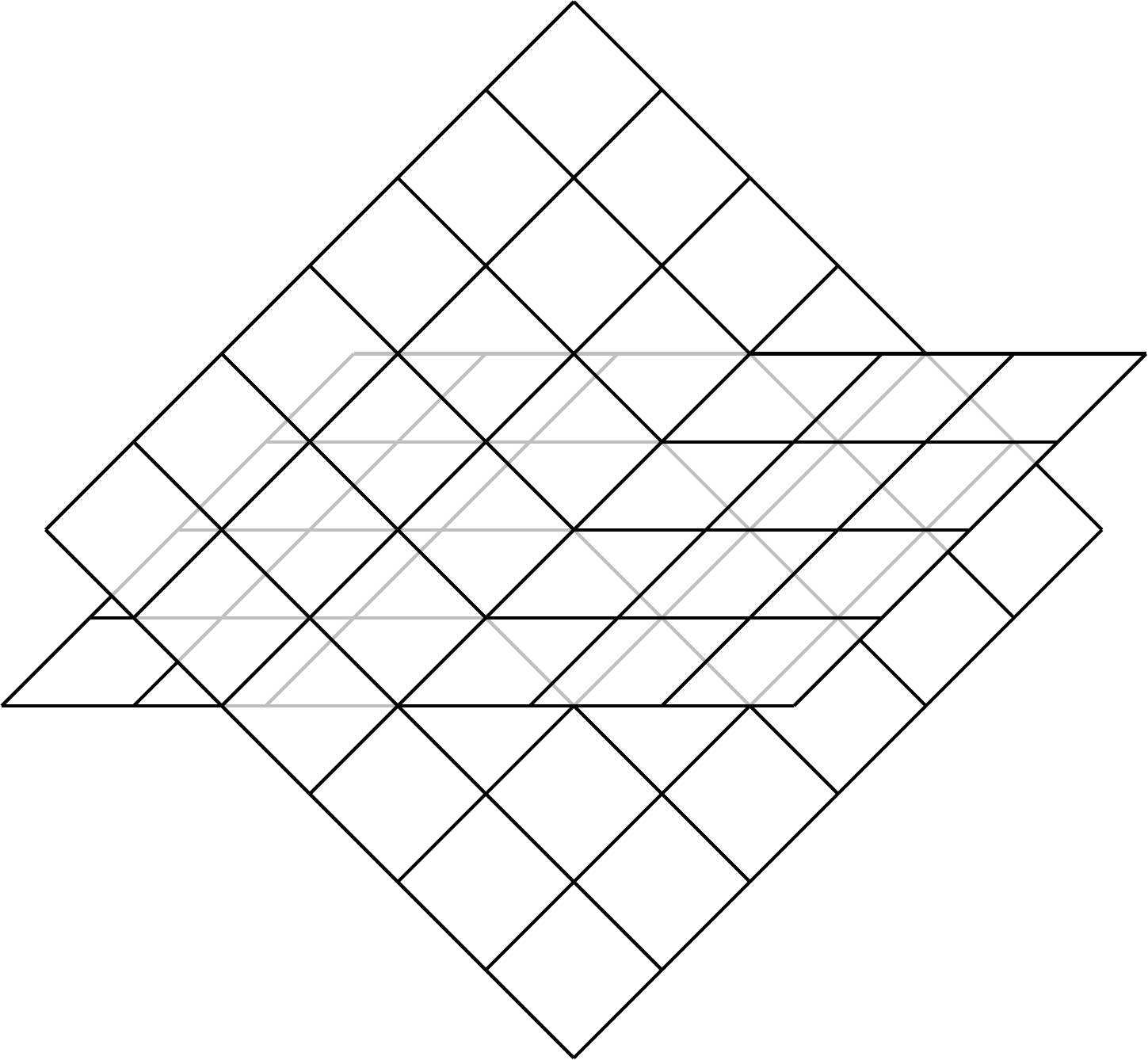}
    \end{minipage}
    \begin{minipage}{0.49\textwidth}
        \centering
        \includegraphics[width=0.8\linewidth]{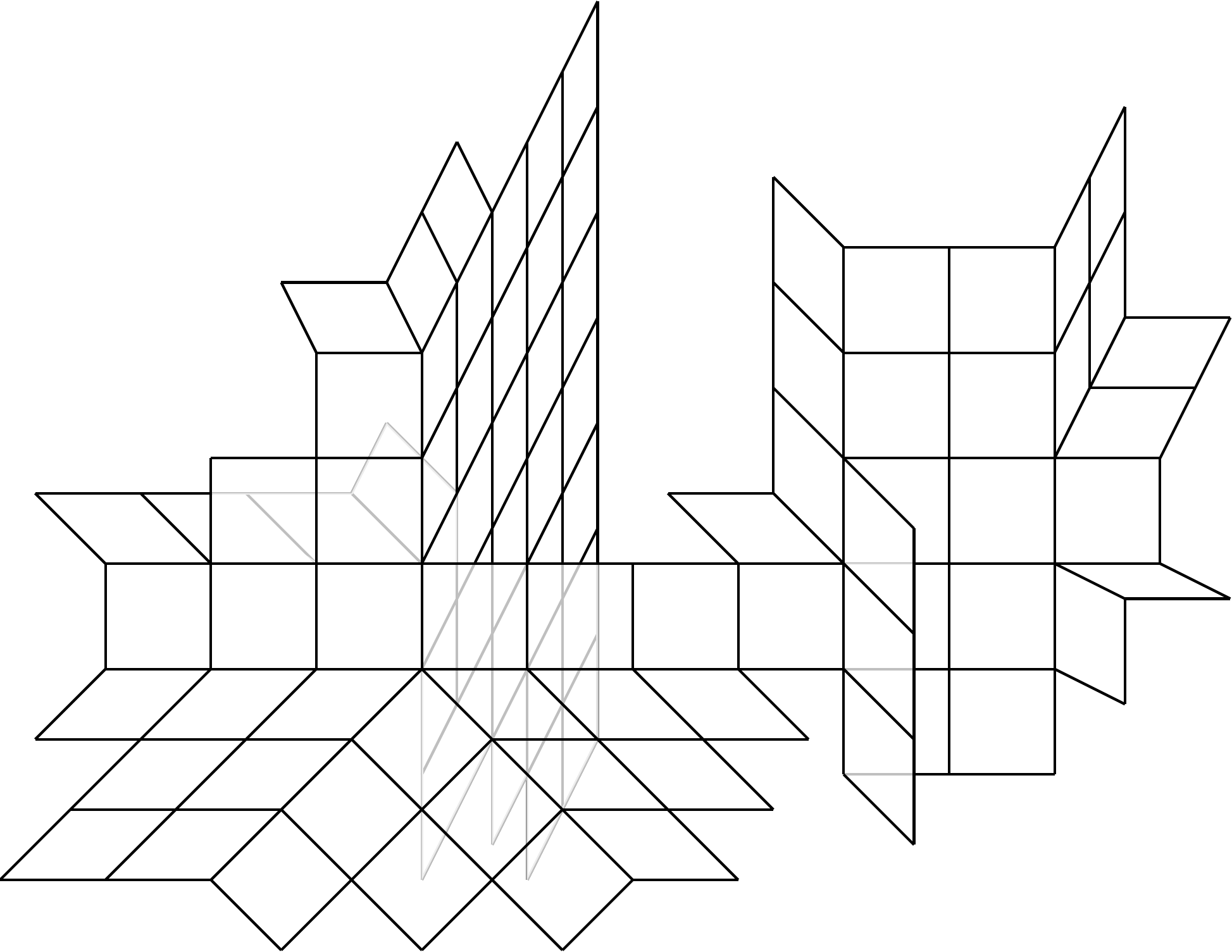}
    \end{minipage}
    \caption{\label{fig_cube-free}
        Two cube-free median graphs. The left graph is used
        as a running example.
    }
\end{figure}

For a vertex $u \in V$ of a graph $G=(V,E)$, let $M(u) := \sum_{v \in V} \dist_G(u,v)$.
A vertex $z \in V$ minimizing the function $M$ is called a
\emph{median} of $G$.  It is well known that any tree $T$ has either a
single median or two adjacent medians. Moreover, a vertex $v$ is a median of $T$ if and only if
any subtree of $T\setminus \{ v\}$ contains at most a half of vertices of $T$. For this reason,
a median vertex  of a tree is often called a centroid. Further, in order to distinguish medians of triplets and the median vertices of a graph $G$,
we will use the name \emph{centroid} also for median vertices  of $G$.

In order to make the proofs easier to follow, most of the main definitions and
notations are recalled in Section \ref{sect_glossary}.

\subsection{Distance and routing labeling schemes}
\label{sect_prelim_dist_rout}

Let $G=(V,E)$ be a finite graph. The \emph{ports} of a vertex $u \in V$ are the
distinct (with respect to $u$) integers, ranging from one to the degree of $u$,
given to the oriented edges around $u$, i.e., the edges $\overrightarrow{uv}$. 
If $uv \in E$, then
the \emph{port from $u$ to $v$}, denoted $\port(u,v)$, is the integer given to
$\overrightarrow{uv}$. More generally, for arbitrary vertices $u,v$ of $G$,
$\port(u,v)$ denotes any value $\port(u,v')$ such that $uv' \in E$ and $v' \in 
I(u,v)$.
\emph{A graph with ports} is a graph to which edges are given ports. All the
graphs in this paper are supposed to be graphs with ports.

A \emph{labeling scheme} on a graph family $\mcg$ consists of an encoding
function and a decoding function. The encoding function is given a total
knowledge of a graph $G \in \mcg$ and gives labels to its vertices in order to
allow the decoding function to answer a predefined question (query) with
knowledge of a restricted number of labels only.
The encoding and decoding functions highly depend on the family $\mcg$ and on
the type of queries: adjacency, distance, or routing queries. More formally,
a \emph{distance labeling scheme} on a graph family $\mcg$
consists of an \emph{encoding function} $C_G : V(G) \to \{0,1\}^*$ that gives
to every vertex of a graph $G$ of $\mcg$ a label, and of a \emph{decoding 
function} $D_G : \{0,1\}^* \times \{0,1\}^* \to \NN$ that, given the labels of
two vertices $u$ and $v$ of $G$, can compute efficiently the distance 
$\dist_G(u,v)$ between them.
In a \emph{routing labeling scheme}, the encoding function $C'_G : V(G) \to
\{0,1\}^*$ gives labels such that the decoding function $D'_G : \{0,1\}^*
\times \{0,1\}^* \to \NN$ is able, given the labels of a \emph{source} $u$ and
a \emph{target} $v$, to decide which port of $u$ to take to get closer to $v$.

We continue by recalling the distance labeling scheme for trees proposed by Peleg in \cite{Peleg00}.
First, as we noticed above, if $T$ is a tree with $n$ vertices and $c$ is a centroid of $T$, then
the removal of $c$  splits $T$ in subtrees with at most $\frac{n}{2}$ vertices 
each. 
The distance between any two vertices $u$ and $v$ from different subtrees  of 
$T\setminus \{ c\}$ is $d_T(u,c)+d_T(c,v)$. Therefore, each vertex of $T$ can 
keep in its label the distance to the centroid $c$. Hence, it remains to  
recover the information necessary to compute the distance between two vertices
in the same subtree of $T\setminus \{ c\}$. This can be done by recursively 
applying to each subtree $T'$ of $T\setminus \{ c\}$ the same procedure as for 
$T$. Consequently, the label of each vertex $v$ of $T$ consists of the 
distances from $v$ to the roots of all subtrees occurring in the recursive 
calls and containing $v$. Since from step to step the size of such subtrees is 
divided by at least 2, $v$ belongs to $\log_2 n$ subtrees, thus the label of 
each vertex $v$ of $T$ has size $\log_2^2n$ bits.

\section{Related work}
\label{related}

In this section we review some known results on distance/routing schemes and  median graphs.

\subsection{Distance and routing labeling schemes}

Distance Labeling Schemes (\dls) have been introduced in a series
of papers by Peleg et al. \cite{Peleg00,Peleg05,GavoillePPR04}. Before these
works, some closely related notions already existed such as embeddings in a
squashed cube \cite{Winkler1983} (equivalent to distance labeling schemes with
labels of size $\log_2 n$ times the dimension of the cube) or labeling schemes
for adjacency requests \cite{kannan1992implicat}. One of the main results for 
\dls is that general graphs support distance labeling schemes with labels of 
size $O(n)$ bits \cite{Winkler1983,GavoillePPR04,AlstrupGHP16a}.
This scheme is asymptotically optimal since $\Omega(n)$
bits labels are needed for general graphs. Another important result is that
there exists a distance labeling scheme for trees with $O(\log^2
n)$ bits labels \cite{Peleg00,AlstrupGHP16b,FrGaNiWe_trees}.
Several classes of graphs containing trees also enjoy a distance labeling
scheme with $O(\log^2 n)$ bit labels such as bounded tree-width graphs
\cite{GavoillePPR04}, distance-hereditary graphs \cite{gavoille2003distance},
bounded clique-width graphs \cite{courcelle2003query}, and non-positively
curved plane graphs \cite{chepoi2006distance}. 
A lower bound of $\Omega(\log^2 n)$ bits on the label length is known for
trees \cite{GavoillePPR04,AlstrupGHP16b}, implying that all the results
mentioned above are optimal as well. Other families of graphs have been
considered such as interval graphs, permutation graphs, and their
generalizations \cite{bazzaro2009localized,gavoille2008optimal} for which an
optimal bound of $\Theta(\log n)$ bits was given, and planar graphs for which
there is a lower bound of $\Omega(n^{\frac{1}{3}})$ bits \cite{GavoillePPR04}
and an upper bound of $O(\sqrt{n})$ bits \cite{GawrychowskiU16}. $(1+\epsilon)$-Approximate 
distance labeling schemes with optimal (polylogarithmic) label length are known for networks
of bounded doubling dimension \cite{AbChGaPe,Talwar}.

Routing is a  basic task that a distributed network must be able to
perform. The design of efficient Routing Labeling Scheme (\rls) is a well
studied subject;  for an overview, we refer to
the book \cite{Peleg00book}. One trivial way to produce a \rls, i.e., a
routing via shortest paths, is to store a complete routing table at each node of
the network. This table specifies, for any destination, the port leading to a
shortest path to that destination. This gives an exact \rls with labels of size
$O(n \log d)$ bits for graphs of maximum degree $d$ that is optimal for general
graphs \cite{gavoille1996memory}. For trees, there exists exact RLS with labels
of size $(1+o(1)) \log_2 n$  \cite{fraigniaud2001routing,thorup2001compact}.
Exact RLS with labels of polylogarithmic size also exist for graphs of bounded
tree-width, clique-width or chordality \cite{dragan2010collective} and for
non-positively curved plane graphs \cite{chepoi2006distance}.
For the families of graph excluding a fixed minor (including planar and bounded
genus graphs), there is an exact \rls with labels of  size $O(\sqrt{n}\log^2
n/\log\log n)$ \cite{dragan2010collective}.  For compact $(1+\epsilon)$-approximate routing 
schemes for networks of bounded doubling dimension, see \cite{AbGaGoMa,KoRiXi}.

\subsection{Median graphs} 
\label{sect_median_graphs}

Median graphs and related structures have an extensive literature; for surveys 
listing their numerous characterizations and properties, see
\cite{BaCh_survey,KlMu,knuth2008}.
These structures have been investigated in different contexts by quite a number
of authors for more than half a century. 
In this subsection we briefly review  the links between median graphs and
CAT(0) cube complexes. We also recall some results, related to the subject of
this paper,  about the distance and shortest path problems in  median graphs
and  CAT(0) cube complexes. For a survey of results on median graphs and their
bijections with median algebras, median semilattices, CAT(0) cube complexes,
and solution spaces of 2-SAT formulae, see \cite{BaCh_survey}. For a 
comprehensive presentation of median graphs and CAT(0) cube complexes as 
domains of event structures, see \cite{CC-JCSS}.

Median graphs are intimately related to hypercubes: median graphs can be 
obtained from hypercubes by amalgams and median graphs are themselves isometric 
subgraphs of hypercubes \cite{BaVdV,Mu}. Moreover, median graphs are exactly 
the retracts of hypercubes \cite{Bandelt_retracts}.
Due to the abundance of hypercubes, to each median graph $G$ one can associate
a cube complex $X(G)$ obtained by replacing every hypercube of $G$ by a solid 
unit cube. Then $G$ can be recovered as the 1-skeleton of $X(G)$.  The cube 
complex $X(G)$ can be endowed with several intrinsic metrics, in particular 
with the $\ell_2$-metric. An important class of cube complexes studied in
geometric group theory  is the class of CAT(0) cube complexes. CAT(0) geodesic 
metric spaces are usually defined via the nonpositive curvature comparison 
axiom of Cartan--Alexandrov--Toponogov \cite{BrHa}. For  cube complexes (and
more generally for cell complexes) the CAT(0) property  can be defined in a
very simple and intuitive way by the property that \emph{$\ell_2$-geodesics 
between any two points are unique}.  Gromov \cite{Gr} gave a nice combinatorial 
characterization of CAT(0) cube complexes as {\it simply connected cube 
complexes with flag links.}
It was also shown in \cite{Ch_CAT,Ro} that \emph{median graphs are exactly the 
1-skeletons  of CAT(0) cube complexes.}

Previous characterizations can be used to show that several cube complexes
arising in applications are CAT(0). Billera et al. \cite{BiHoVo} proved
that the space of trees (encoding all tree topologies with a given set of 
leaves) is a CAT(0) cube complex.
Abrams et al. \cite{AbGh,GhPe} considered the space of all possible positions 
of a reconfigurable system and showed that in many cases this state complex is 
CAT(0).
Billera et al. \cite{BiHoVo}  formulated the problem of  computing the geodesic
between two points in the space of trees. In the robotics literature, geodesics 
in state complexes correspond to the motion planning to get the robot from one 
position to another one with minimal power consumption. A polynomial-time 
algorithm for geodesic problem in the space of trees was provided in 
\cite{OwPr}. A linear-time algorithm for computing distances in CAT(0) square 
complexes (cube complexes of cube-free median graphs) was proposed in 
\cite{ChMa}.  Recently, Hayashi \cite{Hayashi} designed the first 
polynomial-time algorithm for geodesic problem in all CAT(0) cube complexes.

Returning to median graphs, the following is known about the labeling schemes
for them and about some related combinatorial problems.
First, the arboricity of any median graph $G$ on $n$ vertices is at most $\log
n$, leading to adjacency schemes of $O(\log^2n)$ bits per vertex.
As noted in \cite{ChLaRa}, one $\log_2 n$ factor can be replaced by the 
dimension of $G$. Compact distance labeling schemes can be obtained for some 
subclasses of cube-free median graphs. One particular class is that of 
\emph{squaregraphs}, i.e., plane graphs in which  all inner vertices have
degree $\ge 4$. For squaregraphs, distance schemes with labels of size 
$O(\log^2n)$ follow from a more general result of \cite{chepoi2006distance} for 
plane graphs of nonpositive curvature.
Another such class of graphs is that of partial double trees \cite{BaChEp}.
Those are the median graphs which isometrically embed into a Cartesian product 
of two trees and can be characterized as the cube-free median graphs in which 
all links are bipartite graphs \cite{BaChEp}.
The isometric embedding of partial double trees into a product of two trees
immediately leads to distance schemes with labels of $O(\log^2n)$ bits.
Finally, with a technically involved proof, it was shown in \cite{ChHa} that
there exists a constant $M$ such that any cube-free median graph $G$ with
maximum degree $\Delta$ can be isometrically embedded into a Cartesian product
of at most $\epsilon(\Delta):=M\Delta^{26}$ trees. This immediately shows that
cube-free median graph admit distance labeling schemes with labels of length
$O(\epsilon(\Delta)\log^2n)$. Compared with the $O(\log^3n)$-labeling scheme
obtained in the current paper, the disadvantage of the resulting 
$O(\epsilon(\Delta)\log^2n)$-labeling scheme is its dependence of the maximum 
degree $\Delta$ of $G$.
The situation is even worse for high dimensional median graphs:
\cite{ChHa} presents an example of a 5-dimensional median graph/CAT(0)
cube complex with constant degree  which cannot be embedded into a Cartesian 
product of a finite number of trees. Therefore, for general finite median 
graphs the function $\epsilon(\Delta)$ does not exist.

Analogously, it was shown in \cite{Ch_nice} that the nice labeling conjecture 
for event structures (a conjecture formulated in concurrency theory) is false 
for event domains which are median graphs of dimension at least 4 but it was 
proven in \cite{ChHa} that this conjecture is true for event structures with 
2-dimensional (cube-free) domains. On the other hand, it was shown in 
\cite{CC-JCSS} that the Thiagarajan conjecture (an important conjecture in 
concurrency relating 1-safe Petri nets and regular event structures) is false 
already for regular event structures with  2-dimensional domains but is true 
for event structures with hyperbolic domains (the second result heavily relies 
on the very deep result of Agol \cite{Agol} from geometric group theory).
All this in a sense explains the difficulty of designing polylogarithmic  
distance labeling schemes for general median graphs and motivates the 
investigation of such schemes for cube-free median graphs.
Nevertheless, we do not have a proof that such schemes do not exist for all 
median graphs.

\section{Fibers in  median graphs}
\label{sect_fibers}

In this section, we recall  the properties of median graphs and of the fibers
of their gated subgraphs. They will be used in our labeling schemes and some of
them could be potentially useful for designing distance labeling schemes for
general median graphs. Since all those results are dispersed in the literature
and time, we present them with (usually, short and unified) proofs in the
appendix.

\begin{lemma} 
    \label{thm_quadrangle} 
    (Quadrangle Condition) For any vertices $u,v,w,z$ of a median graph $G$ 
    such that $\dist_G(u,z) = k+1$, $v, w\sim z$, and $\dist_G(u,v) = 
    \dist_G(u,w) = k$, there is a unique vertex $x\sim v,w$  such that 
    $\dist_G(u,x) = k-1$.
\end{lemma}

The Quadrangle Condition is simultaneously a local and global metric condition 
on graphs, which has topological consequences. In bipartite graphs, this 
condition implies that any cycle $C$ can be ``paved'' in a special way (from 
top-to-bottom with respect to any basepoint of $C$) with quadrangles (4-cycles 
or squares). This implies that all cycles are null-homotopic, i.e., that the 
square complexes of such graphs  are simply connected. The Quadrangle and 
Triangle Conditions define the class of weakly modular graphs 
\cite{BaCh_survey,Ch_metric}, which comprises most of te classes of graphs 
investigated in Metric Graph Theory; for a full account, see the survey 
\cite{BaCh_survey} and the recent paper \cite{ChChHiOs}.

The following result is a particular case of the local-to-global 
characterization of convexity and gatedness in weakly modular graphs 
established in \cite{Ch_metric}:

\begin{lemma} 
    \label{thm_convex=gated}
    A subset of vertices $A$ of a median graph $G$ gated iff $A$ is convex and 
    iff $H:=G[A]$ is connected and $A$ is \emph{locally convex}, i.e., if 
    $x,y\in A$ and $d_G(x,y)=2$, then $I(x,y)\subseteq A$.
\end{lemma}

\begin{lemma} 
    \label{convex-interval} 
    Intervals of a median graph $G$ are convex. 
\end{lemma}

Suppose that the median graph $G$ is rooted at an arbitrary vertex $v_0$. For a 
vertex $v$, all neighbors $u$ of $v$ such that $u\in I(v,v_0)$ are called 
\emph{predecessors} of $v$.  A median graph $G$ satisfies the \emph{downward 
cube property} \cite{BeChChVa} if any vertex $v$ and all its predecessors 
belong to a single cube of $G$.

\begin{lemma} 
    \label{descendent_cube}
    \cite{Mu}
    Any median graph $G$ satisfies the downward cube property.
\end{lemma}

Lemma~\ref{descendent_cube} immediately implies the following upper
bound on the number of edges of $G$:

\begin{corollary}
    \label{upper_edges}
	If a median graph $G$ has $n$ vertices, $m$ edges, and dimension $d$, then 
	$m\le dn\le n\log_2 n$. In particular, $m\le 2n$ if $G$ is cube-free.
\end{corollary}

Lemma~\ref{descendent_cube} also implies the following useful property of cube-free median graphs:

\begin{corollary} 
    \label{neighbors_interval} 
    If $u$ and $v$ are any two vertices of a cube-free median graph $G$, then 
    $v$ has at most two neighbors in the interval $I(u,v)$.
\end{corollary}

We continue with properties of stars and fibers of stars.
Combinatorially, the stars of median graphs may have quite an arbitrary 
structure: by a result of \cite{BaVdV}, there is a bijection between the stars 
of median graphs and arbitrary graphs. Given an arbitrary graph $H$, the {\it 
simplex graph} \cite{BaVdV} $\sigma(H)$ of $H$ has a vertex $v_{\sigma}$ for 
each clique of $G$ (i.e., empty set, vertices, edges, triangles, etc.) and two 
vertices $v_{\sigma}$ and $v_{\sigma'}$ are adjacent in $\sigma(H)$ if and only 
if the cliques $\sigma$ and $\sigma'$ differ only in a vertex. It was shown
in \cite{BaVdV} that the simplex graph $\sigma(H)$ of any graph $H$ is a median 
graph. Moreover, one can easily show that the star in $\sigma(H)$ of the vertex 
$v_{\varnothing}$ (corresponding to the empty set)  coincides with the whole 
graph $\sigma(H)$. Vice-versa, any star $\St(z)$ of a median graph can be 
realized as the simplex graph $\sigma(H)$ of the graph $H$ having the neighbors 
of $z$ as the set of vertices and two such neighbors $u',u''$ of $z$ are 
adjacent in $H$ if and only if $z,u',u''$ belong to a common square of $G$.

Next, we consider stars $\St(z)$ of median graphs from the metric point of view.

\begin{lemma} 
    \label{prop_St(m)_convex}
    For any vertex $z$ of a median graph $G$, the star  $\St(z)$ is a gated 
    subgraph of $G$.
\end{lemma}

The following property of median graphs is also well-known in more general 
contexts. The graphs satisfying this property are called {\it 
fiber-complemented} \cite{Chast}.

\begin{lemma} 
    \label{fiber-gated1}
    For any gated subgraph $H$ of a median graph $G$, the fibers $F(x), x \in 
    V(H)$, are gated.
\end{lemma}

Lemma \ref{fiber-gated1} has two corollaries. First, from this lemma and Lemma \ref{prop_St(m)_convex} we obtain:

\begin{corollary} 
    \label{fiber-gated2} 
    For any vertex $z$ of a median graph $G$, the fibers of the star $\St(z)$ 
    are gated.
\end{corollary}

For an edge $uv$ of a median graph $G$, let $W(u,v)=\{ z\in V: d_G(z,u)<d_G(z,v)\}$  and $W(v,u)=\{ z\in V: d_G(z,v)<d_G(z,u)\}$.
Since $G$ is bipartite, $W(u,v)$ and $W(v,u)$ constitute a partition of the 
vertex set of $G$.  In view of the following result,  we will refer to the sets 
of the form $W(u,v), W(v,u)$ as to the \emph{halfspaces} of $G$ (recall that a 
halfspace is a convex set with convex complement):

\begin{corollary} 
    \label{halfspaces-gated1} 
    For any edge $uv$ of a median graph $G$, the sets $W(u,v)$ and $W(v,u)$ are 
    gated and, thus, are complementary halfspaces of $G$. Conversely, any pair 
    of complementary halfspaces of $G$ has the form $\{ W(u,v),W(v,u)\}$ for an 
    edge $uv$ of $G$.
\end{corollary}

\begin{proof}
    Since edges of a median graph $G$ are gated, the first assertion of the 
    corollary follows by applying Lemma \ref{fiber-gated1} to the edge $uv$.
    Conversely, if $\{H',H''\}$ is a pair of complementary halfspaces of $G$, 
    let $uv$ be an edge of $G$ with $u\in H'$ and $v\in H''$. Since both $H'$ 
    and $H''$ are convex, we conclude that $H'\subseteq W(u,v)$ and 
    $H''\subseteq W(v,u)$. Since both pairs $\{ H',H''\}$ and $\{ 
    W(u,v),W(v,u)\}$ define partitions, we have $H'=W(u,v)$ and $H''=W(v,u)$.
\end{proof}

That the halfspaces of a median graph are convex was established first by 
Mulder \cite{Mu}. He also proved that the boundaries of halfspaces are convex 
(the boundary of the halfspace $W(u,v)$ is the set $\partial W(u,v)=\{ z'\in 
W(u,v): \exists z''\in W(v,u), z''\sim z'\}$).
We will prove this property for boundaries of fibers of arbitrary gated 
subgraphs of a median graph.

Let $H$ be a gated subgraph of a median graph $G=(V,E)$ and let ${\mathcal F}(H)=\{ F(x): x\in V(H)\}$ be the partition of $V$ into the  fibers of $H$.
We will call two fibers $F(x)$ and $F(y)$ {\it neighboring} (notation $F(x)\sim F(y)$) if there exists an edge $x'y'$ of $G$ with one end $x'$ in $F(x)$ and another end $y'$ in $F(y)$.
If $F(x)$ and $F(y)$ are neighboring fibers of $H$, then denote by 
$\partial_yF(x)$ the set of all vertices $x'\in F(x)$ having a neighbor $y'$ in 
$F(y)$ and call $\partial_yF(x)$ the {\it boundary of $F(x)$ relative to 
$F(y)$}.

\begin{lemma} 
    \label{lem_gated_borders}
    Let $H$ be a gated subgraph of a median graph $G =(V,E)$. Two fibers $F(x)$ 
    and $F(y)$ of $H$ are neighboring if and only if $x\sim y$. If $F(x)\sim 
    F(y)$, then  the boundary $\partial_yF(x)$ of $F(x)$ relative to $F(y)$ 
    induces a gated subgraph of $G$ of dimension $\le \dim(G)-1$.
\end{lemma}

For a vertex $x$ of a gated subgraph $H$ of $G$ and its fiber $F(x)$, the union 
of all boundaries $\partial_y F(x)$ over all $F(y)\sim F(x), y\in V(H)$, is 
called the {\it total boundary} of $F(x)$ and is denoted by $\partial^* F(x)$. 
More precisely, the vertices and the edges of $\partial^* F(x)$ are the unions
of vertices and edges of all boundaries $\partial_y F(x)$ over all fibers 
$F(y)\sim F(x), y\in V(H)$.

\begin{lemma} 
    \label{lem_dim-1_Uborders}
    Let $H$ be a gated subgraph of a median graph $G$ of dimension $d$.
    Then the total boundary $\partial^* F(x)$ of any fiber $F(x)$ of $H$ does 
    not contain $d$-dimensional cubes.
\end{lemma}

\begin{proof} 
    Suppose by way of contradiction that $\partial^* F(x)$ contains a 
    $d$-dimensional cube $Q$.
    Since $Q$ is a gated subgraph of $G$, we can consider the gate $x'$ of $x$ 
    in $Q$ and  the vertex furthest from $x$ in $Q$ (this is the vertex of $Q$ 
    opposite to $x'$).
    Denote this furthest from $x$ vertex of $Q$ by $x''$. From the choice of 
    $x'$ as the gate of $x$ in $Q$ and $x''$ as the opposite to $x'$ vertex of 
    $Q$, we conclude that $Q\subseteq I(x,x'')$. Suppose that $x''\in 
    \partial_y F(x)\subset \partial^* F(x)$.
    Since $\partial_y F(x)$ is gated (Lemma \ref{lem_gated_borders}) and 
    $Q\subseteq I(x'',x)$,  $Q$ is included in the boundary $\partial_y F(x)$. 
    This contradicts Lemma \ref{lem_gated_borders} that $\partial_y F(x)$ has 
    dimension $\le d-1$.
\end{proof}

\begin{lemma} 
    \label{lem_isometric_Uborders} 
    Let $H$ be a gated subgraph of a median graph $G$. 
    Then the total boundary $\partial^* F(x)$ of any fiber $F(x)$ of $H$ 
    induces an isometric subgraph of $G$.
\end{lemma}

\begin{proof} 
    Pick $u,v\in \partial^* F(x)$, say $u\in \partial_y F(x)$ and $v\in 
    \partial_z F(x)$.
    Let $w$ be the median of the triplet $x,u,v$. Since $w\in I(u,x)\subseteq 
    \partial_y F(x)\subset \partial^* F(x)$ we deduce that $I(u,w)\subseteq 
    \partial^* F(x)$. Analogously, we can show that $I(v,w)\subseteq \partial^* 
    F(x)$. Since $w\in I(u,v)$ and $I(u,w)\cup I(w,v)\subseteq \partial^* 
    F(x)$, the vertices $u$ and $v$ can be connected in $\partial^* F(x)$ by a 
    shortest path passing via $w$.
\end{proof}

We conclude this section with an additional property of fibers of stars of  centroids of $G$.
Recall, that $c$ is a centroid of $G$ if $c$ minimizes the function 
$M(x)=\sum_{v\in V} d_G(x,v).$

\begin{lemma} 
    \label{lem_small_halfspaces} 
    Let $c$ be a centroid of a median graph $G$ with $n$ vertices. Then any 
    fiber $F(x)$ of the star $\St(c)$ of $c$ has at most $n/2$ vertices.
\end{lemma}

\begin{proof} 
    Suppose by way of contradiction that $|F(x)|>n/2$ for some vertex $x\in 
    \St(c)$. Let $u$ be a neighbor of $c$ in $I(x,c)$.
    If $v\in F(x)$, then  $x\in I(v,c)$ and  $u\in I(x,c)$, and we conclude 
    that $u\in I(v,c)$. Consequently, $F(x)\subseteq W(u,c)$, whence
    $|W(u,c)|>n/2$. Therefore $|W(c,u)|=n-|W(u,c)|<n/2$. But this contradicts 
    the fact that $c$ is a centroid of $G$. Indeed,
    since $u\sim c$, one can easily show that $M(u)-M(c)=|W(c,u)|-|W(u,c)|<0$.
\end{proof}

Unfortunately, the total boundary $\partial^* F(x)$ of a fiber does not always induce a  median subgraph. Therefore,
even if  $\partial^* F(x)$ is an isometric subgraph of $G$ of dimension $\le \dim(G)-1$, one cannot recursively apply the algorithm
to the subgraphs induced by the total boundaries $\partial^* F(x)$. However, if $G$ is 2-dimensional (i.e., $G$ is cube-free), then the total boundaries of fibers are isometric
subtrees of $G$ and one can use for them  distance and routing schemes for trees. Even in this case, we still need an additional
property of total boundaries, which we will establish in the next section.

\section{Fibers in cube-free median graphs}
\label{sect_cube-free_fibers}

In this section, we establish additional properties of fibers of stars and  of 
their total boundaries in cube-free median graphs $G=(V,E)$.
Using them we can show that for any pair $u,v$ of vertices of $G$, the 
following trichotomy holds: the distance $d_G(u,v)$ either can be computed as 
$d_G(u,c)+d_G(c,v)$, or as the sum of distances from $u,v$ to appropriate 
vertices $u',v'$  of $\partial^*F(x)$ plus the distance between $u',v'$ in 
$\partial^*F(x)$, or via a recursive call to the  fiber containing $u$ and $v$.

\subsection{Classification of fibers}
\label{classification-fibers}

From now on, let $G=(V,E)$ be a cube-free median graph. Then the star $\St(z)$ 
of any vertex $z$ of $G$ is the union of all squares and edges containing $z$. 
Specifying the bijection between stars of median graphs and simplex graphs of 
arbitrary graphs mentioned above, the stars of cube-free median graphs 
correspond to simplex graphs of triangle-free graphs.

Let $z$ be an arbitrary vertex of $G$ and let ${\mathcal F}_z=\{ F(x): x\in
\St(z)\}$ denote the partition of $V$ into the fibers of $\St(z)$.
We distinguish two types of fibers: the fiber $F(x)$ is called a {\it panel} if $x$ is adjacent to $z$ and $F(x)$ is called a
{\it cone} if $x$ has distance two to $z$. The interval $I(x,z)$ is the edge $xz$ if $F(x)$ is a panel and is a
square $Q_x:=(x,y',z,y'')$ if $F(x)$ is a cone. In the second case, since $y'$ and $y''$ are the only neighbors of $x$ in $\St(z)$,
by Lemma \ref{lem_gated_borders} we deduce that the cone $F(x)$ is adjacent to the panels $F(y')$ and $F(y'')$ and that $F(x)$ is not adjacent to any
other panel or cone. By the same lemma, any panel $F(y)$ is not adjacent to any other panel, but $F(y)$ is adjacent to all cones
$F(x)$ such that the square $Q_x$ contains the edge $yz$. For an illustration,
see Figure \ref{fig_fibers}.

\subsection{Total boundaries of fibers are quasigated}
\label{crucial-total-boundary}

For a set $A$, an \emph{imprint}  of a vertex $u\notin A$ on $A$ is a vertex $a\in A$ such that
$I(u,a)\cap A=\{ a \}$. Denote by $\eproj{u}{A}$ the set of all imprints of $u$ 
on $A$. 
The most important property of imprints is that for any vertex $z\in A$, there 
exists a shortest $(u,z)$-path passing via an imprint, i.e., that $I(u,z)\cap 
\eproj{u}{A}\ne \varnothing$.
Therefore, if the set $\eproj{u}{A}$ has constant size, one can store
in the label of $u$ the distances to the vertices of  $\eproj{u}{A}$. Using 
this, for any $z\in A$, one can compute $d_G(u,z)$ as $\min \{ 
d_G(u,a)+d_G(a,z): a\in \eproj{u}{A}\}$.
Note that a set $A$ is gated if and only if any vertex $u\notin A$ has a unique imprint on $A$.
Following this, we will say that a set $A$ is {\it $k$-gated} if for any vertex 
$u\notin A$, $|\eproj{u}{A}|\le k$. In particular, we will say that a set $A$ 
is {\it quasigated} if $|\eproj{u}{A}|\le 2$ for any vertex $u\notin A$.
The main goal of this subsection is to show that the total boundaries of fibers are quasigated.

\begin{figure}
    \centering
    \includegraphics[scale=0.4]{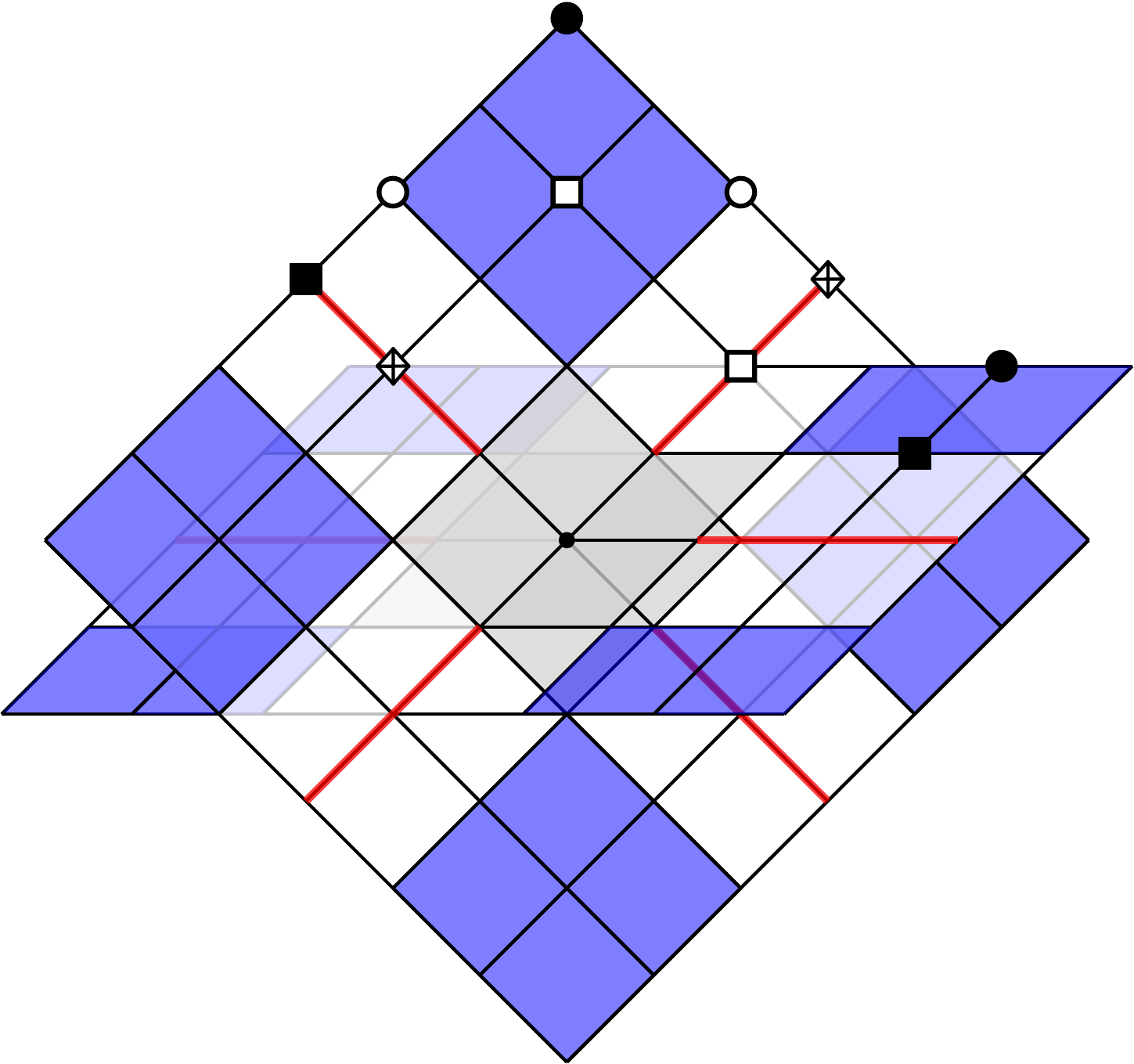}
    \caption{\label{fig_fibers}
        A star $\St(z)$ (in gray), its fibers (\cones in blue, and \panels in
        red), and an illustration of the classification of pairs of vertices:
        the pairs of empty circle points are \roommates, of disk points are
        \aNeighboring,of empty square points are \neighboring, and of full
        square and diamond points are \separated.
    }
\end{figure}

Let $T$ be a tree with a distinguished vertex $r$ in $G$.
The vertex $r$ is called the {\it root} of $T$ and $T$ is called a {\it rooted 
tree}. 
We will say that a rooted tree $T$ has {\it gated branches} if for any vertex 
$x$ of $T$ the unique path $P(x,r)$ of $T$ connecting $x$ to the root $r$ is a 
gated subgraph of $G$.

By adapting Lemmas \ref{lem_gated_borders}, \ref{lem_dim-1_Uborders}, and
\ref{lem_isometric_Uborders} to cube-free median graphs, we obtain the 
following result:

\begin{lemma} 
    \label{lem_Uborders_cube-free}
    For each \fiber $F(x)$ of a star $\St(z)$ of a cube-free median graph $G$,
    the total boundary $\partial^*F(x)$ is an isometric tree with gated branches.
\end{lemma}

\begin{proof}  
    By Lemma \ref{lem_isometric_Uborders}, $\partial^* F(x)$ is an isometric 
    subgraph of $G$. Now we show that  $\partial^* F(x)$  is a tree. By Lemma 
    \ref{lem_gated_borders}, each boundary $\partial_yF(x)$ is a gated tree  
    rooted at $x$, thus  $\partial^* F(x)$ is the union of the gated trees. 
    Suppose by way of contradiction that $\partial^* F(x)$ contains a cycle 
    $C$. Let $u$ be a furthest from $x$ vertex of $C$ and suppose that $u\in 
    \partial_yF(x)$.
    Let $v',v''$ be the neighbors of $u$ in $C$. Since $G$ is bipartite, from 
    the choice of $u$ we conclude that $v',v''$ are closer to $x$ than $u$. 
    Since $\partial_yF(x)$ is gated, and thus convex, this implies that 
    $v',v''\in I(u,x)\subseteq \partial_yF(x)$, contrary to the fact that 
    $\partial_yF(x)$ is a gated tree  rooted at $x$.
    This shows that $\partial^*F(x)$ is an isometric tree of $G$.

    For any vertex $v\in \partial^*F(x)$ there exists a fiber $F(y)\sim F(x)$ 
    of $\St(z)$ such that $v$ belongs to the boundary $\partial_y F(x)$ of 
    $F(x)$ relative to $F(y)$.
    Since $\partial_y F(x)$ is a gated subtree of $G$ and $v,x\in \partial_y 
    F(x)$, the unique path $P(v,x)$ connecting $v$ and $x$ in $\partial_y F(x)$ 
    is a convex subpath of $\partial_y F(x)$, and therefore a convex subpath of 
    the whole graph $G$. Since convex subgraphs are gated, $P(x,v)$ is a gated 
    path of $G$ belonging to $\partial^*F(x)$. This proves that 
    $\partial^*F(x)$ is a tree with gated branches.
\end{proof}

By Lemma \ref{lem_Uborders_cube-free},  each $\partial^*F(x)$ induces a tree of $G$. This tree is an isometric subgraph of $G$, which means that
any two vertices of $\partial^*F(x)$ can be connected in $\partial^*F(x)$ by a shortest path of $G$.  Moreover, since $\partial^*F(x)$ has gated branches,
the respective shortest path connecting any vertex $v$ of
$\partial^*F(x)$ to the root $x$ coincides with the entire interval $I(v,x)$. However the tree $\partial^*F(x)$
is not necessarily gated itself. Since a panel $F(x)$ may be adjacent to an arbitrary number of cones,
one can think that the imprint-set $\eproj{u}{\partial^*F(x)}$ of a vertex $u$ of $F(x)$
may have an arbitrarily large size. The following  lemma shows that this is not the case,
namely that $|\eproj{u}{\partial^*F(x)}|\le 2$. This is one of the key ingredients in
the design of the distance and routing labeling schemes
presented in Sections \ref{sect_dist_labeling} and \ref{sect_rout_labeling}. 

\begin{lemma} 
    \label{lem_2_gates_Uborders} 
    Any rooted tree $T$ with gated branches of $G=(V,E)$ is quasigated.
\end{lemma}

\begin{proof} 
    Let $r$ be the root of $T$. Pick any $u\in V\setminus V(T)$ and suppose by 
    way of contradiction
    that  $\eproj{u}{T}$ contains three distinct imprints $x_1$, $x_2$, and 
    $x_3$. Since $T$ has gated branches,
    none of the vertices $x_1,x_2,x_3$ belong to the path of $T$ between the 
    root $r$ and another vertex from this triplet.
    In particular,  $r$ is different from $x_1,x_2,x_3$.
    Suppose additionally that among all rooted trees $T'$ with gated branches 
    of $G$ and such that  $|\eproj{u}{T'}|\ge 3$,
    the tree $T$ has the minimal number of vertices. This minimality choice (and
    the fact that any subtree of $T$ containing $r$
    is also a  rooted tree with gated branches)  implies that $T$ is exactly 
    the union of the three gated
    paths $P(r,x_1),P(r,x_2),$ and $P(r,x_3)$ connecting the root $r$ with the 
    leaves $x_1,x_2,$ and $x_3$ of $T$.
    Notice that $P(r,x_1),P(r,x_2),$ and $P(r,x_3)$ not necessarily pairwise 
    intersect only in $r$.
    
    First, notice that $x_1,x_2,x_3\in I(u,r)$. Indeed, let  $z_i$ denote the 
    median of the triplet $x_i,u,r$. If $z_i\ne x_i$, since
    $z_i\in I(x_i,r)=P(x_i,r)\subset T$ and $z_i\in I(u,x_i)$, we obtain a 
    contradiction with the inclusion of $x_i$ in $\eproj{u}{T}$.
    Thus $z_i=x_i$, yielding $x_i\in I(u,z_i)$.
    
    Let $y_i$ be the neighbor of $x_i$ in the path $P(r,x_i)$, $i=1,2,3$. Since 
    $G$ is bipartite, either $x_i\in I(y_i,u)$ or $y_i\in I(x_i,u)$. Since 
    $x_i\in \eproj{u}{T}$, necessarily
    $x_i\in I(y_i,u)$. Let $T'_i$ be the subtree of $T$ obtained by removing 
    the leaf $x_i$. From the minimality choice of $T$, we cannot replace $T$ by 
    the subtree $T'_i$. This means that
    $|\eproj{u}{T'_i}|\le 2$. Since $x_j,x_k\in \eproj{u}{T'_i}$ for 
    $\{i,j,k\}=\{ 1,2,3\}$, necessarily $I(y_i,u)\cap \{ x_j,x_k\}\ne 
    \varnothing$ holds.
    
    Suppose without loss of generality that $\dist_G(r,x_3) = \max\{ 
    \dist_G(r,x_i) : i = 1,2,3 \}:=k$. Since $I(y_3,u)\cap \{ x_1,x_2\}\ne 
    \varnothing$ holds, we can
    suppose without loss of generality that $x_2\in I(y_3,u)$.
    Since $x_3\in I(y_3,u)$, from these two inclusions we obtain that 
    $\dist_G(x_3,u) + 1 = \dist_G(y_3,x_2) + \dist_G(x_2,u)$. Therefore, 
    $\dist_G(x_3,u) \geq \dist_G(x_2,u)$.
    Since $x_2,x_3\in I(u,r)$, we have 
    $\dist_G(u,x_2)+\dist_G(x_2,r)=\dist_G(u,x_3)+\dist_G(x_3,r)$. Since 
    $\dist_G(r,x_3)\ge \dist_G(r,x_2)$, all
    this is possible only if  $\dist_G(x_3,u) = \dist_G(x_2,u)$ and  
    $d_G(x_3,r)=d_G(x_2,r)$. Moreover,
    $\dist_G(y_3,x_2)=1$ holds, i.e., $y_3$ and $x_2$ are adjacent in $T$. 
    Since $x_2$ is a leaf of $T$, this is possible only if $y_2$ and $y_3$ 
    coincide. Let $y:=y_2=y_3$.
    We distinguish two cases:

    \smallskip\noindent
    {\bf Case 1.} $\dist_G(x_1,r)=k$.
    
    \smallskip\noindent
    Since all three vertices $x_1,x_2,x_3$ have the same distance $k$ to $r$, 
    we can apply to $x_1$ the same analysis as to $x_3$ and deduce that the 
    neighbor $y_1$ of
    $x_1$ in $T$ coincides with one of the vertices $y_2,y_3$.
    Since $y_2=y_3=y$, we conclude that the vertices $x_1,x_2,x_3$ have the 
    same neighbor $y$ in $T$. Since $y$ is closer to $r$ than each of the 
    vertices $x_1,x_2,x_3$ and since
    $x_1,x_2,x_3\in I(r,u)$, we conclude that $x_1,x_2,x_3\in I(y,u)$. By the 
    minimality of $T$, we conclude that $k=1$ and $y=r$, i.e., $T$ consists 
    only of $x_1,x_2,x_3,$ and $y=r$.
    Applying the quadrangle condition three times, we can find  three vertices 
    $x_{i,j}, i,j\in \{ 1,2,3\}, i\ne j,$ such that
    $x_{i,j}\sim x_i,x_j$ and $d_G(x_{i,j},u)=k-1$ (see Figure 
    \ref{fig_lemma_2gates}, left). If two of the vertices $x_{1,2},x_{2,3},$ 
    and $x_{3,1}$ coincide, then we will get a forbidden $K_{2,3}$: if, say 
    $x_{1,2}=x_{2,3},$
    then this copy of $K_{2,3}$ contains the vertices $y,x_1,x_2,x_3,$ and 
    $x_{1,2}=x_{2,3}$. Thus $x_{1,2},x_{2,3},$ and $x_{3,1}$
    are pairwise distinct. Since $G$ is bipartite, this implies that 
    $d_G(x_i,x_{j,k})=3$ for $\{ i,j,k\}=\{ 1,2,3\}$. Since $x_{1,2},x_{2,3}\in 
    I(x_2,u)$, by quadrangle condition
    there exists a vertex $w$ such that $w\sim x_{1,2},x_{2,3}$ and 
    $d_G(w,u)=k-2$. Since $G$ is bipartite, $d_G(w,x_{3,1})$ equals to 3 or to 
    1. If $d_G(w,x_{3,1})=3=d_G(y,w)$, then the triplet $y,w,x_{3,1}$ has two
    medians $x_1$ and $x_3$, which is impossible, because $G$ is  median. Thus  
    $d_G(w,x_{3,1})=1$, i.e., $w\sim x_{3,1}$. Then one can easily see that
    the vertices $y,x_1,x_2,x_3,x_{1,2},x_{2,3},x_{3,1},w$
    define an isometric  3-cube of $G$, contrary to the assumption that $G$ is 
    cube-free.
    This finishes the analysis of Case 1.

    \smallskip\noindent
    {\bf Case 2.} $\dist_G(x_1,r) < k$.

    \smallskip\noindent
    This implies that $d_G(r,x_1)\le k-1=d_G(r,y)$. Let $r'$ be the neighbor of $r$ in the $(r,y)$-path of $T$. Notice that $r'\notin I(r,x_1)=P(r,x_1)$. Indeed,
    otherwise, $r'\in P(r,x_1)\cap P(r,x_2)\cap P(r,x_3)$ and we can replace the tree $T$ by the subtree $T'$ rooted at $r'$ and
    consisting of the subpaths of $P(r,x_i)$ comprised between $r'$ and $x_i$, $i=1,2,3$. Clearly $T'$ is a rooted tree with gated branches and $x_1,x_2,x_3\in \eproj{u}{T'}$,
    contrary to the minimality choice of the counterexample $T$. Thus $r'\notin
    P(r,x_1)$.

    Let also
    $P(r,x_1)=(r,v_1,\ldots,v_{m-1},v_m=:x_1)$. Notice that $r$ may coincide with $y_1$ and $x_1$ may coincide with $v_1$. Since $v_1,r'\in I(r,u)$, applying the
    quadrangle condition we will find a vertex $v_1'\sim v_1,r'$ at distance
    $d_G(r,u)-2$ from $u$. Since $r'\notin I(r,x_1)$, $v_1'\ne v_2$. Since
    $v_2,v'_1\in I(v_1,u)$,
    by quadrangle condition we will find $v'_2\sim v_2,v'_1$ at distance
    $d_G(r,u)-3$ from $u$. Again, since $r'\notin I(r,x_1)$, $v'_2\ne v_3$.
    Continuing this way,
    we will find the vertices $r',v'_1,v'_2,\ldots,v'_{m-1},v'_m =: x'_1$
    forming
    an $(r',x'_1)$-path $P(r',x'_1)$ and such that $v'_i \sim v_i,v'_{i-1}$,
    $v'_i \ne v_{i+1}$, and $v'_i$ is
    one step closer to $u$ than $v_i$ and $v'_{i-1}$ (see Figure
    \ref{fig_lemma_2gates}, right). From its construction, the path
    $P(r',x'_1)$ is a shortest path. We assert that $P(r',x'_1)$ is gated.  If
    this is not the case,
    by Lemma \ref{thm_convex=gated} and since $P(r',x'_1)$ is shortest, we can find two vertices $v'_{i-1},v'_{i+1}$ having a common neighbor $z'$ different from $v'_i$. Let $z$ be the median
    of the triplet $z',v_{i-1},v_{i+1}$. Then $z$ is a common neighbor of $z',v_{i-1},v_{i+1}$ and $z$ is different from $v_i$ (otherwise, we obtain a forbidden $K_{2,3}$). But then one can easily check that
    the vertices $v_{i-1},v_i,v_{i+1}, v'_{i-1},v'_i,v'_{i+1},z,z'$ induce in $G$ an isometric 3-cube, contrary to the assumption that $G$ is cube-free.  Consequently, $P(r',x'_1)$ is a gated path of $G$.

    Let $T''$ be the tree rooted at $r'$ and consisting of the gated path $P(r',x'_1)$ and the gated subpaths of $P(r,x_2)$ and $P(r,x_3)$ between $r'$ and $x_2,x_3$, respectively. Clearly, $T''$ is a rooted tree with gated branches.
    Notice that $x'_1,x_2,x_3\in \eproj{u}{T''}$. Indeed, if $x_2$ or $x_3$
    belonged to $I(x'_1,u)$, then $x'_1$ would belong to $I(x_1,u)$ and we
    would
    conclude that
    $x_2$ or $x_3$ belongs to $I(x_1,u)$, which is impossible because $x_1\in
    \eproj{u}{T}$.
    On the other hand, $x'_1$ cannot belong to $I(x_2,u)$ or to $I(x_3,u)$ because $d_G(x'_1,u)=d_G(x_1,u)-1\le d_G(x_2,u)=d_G(x_3,u)$. Consequently, $|\eproj{u}{T''}|\ge 3$. Since $T''$ contains less vertices than $T$, we obtain a contradiction
    with the minimality choice of $T$. This concludes the analysis of Case 2, thus  $T$ is quasigated.
\end{proof}

\begin{figure}[ht]
    \begin{minipage}{0.4\textwidth}
        \centering
        \includegraphics[width=0.9\textwidth]{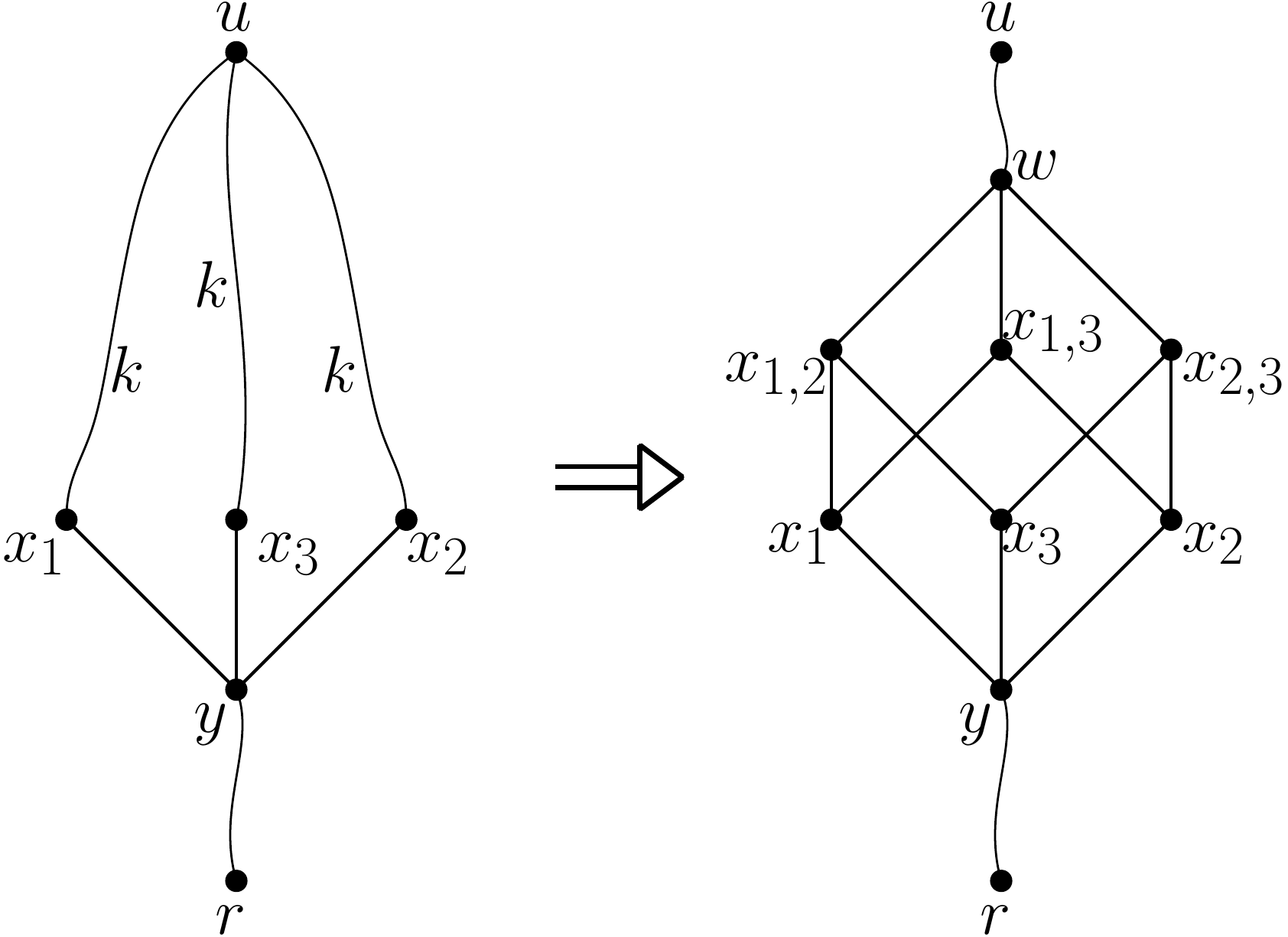}
    \end{minipage}
    \begin{minipage}{0.49\textwidth}
        \centering
        \includegraphics[width=0.8\textwidth]{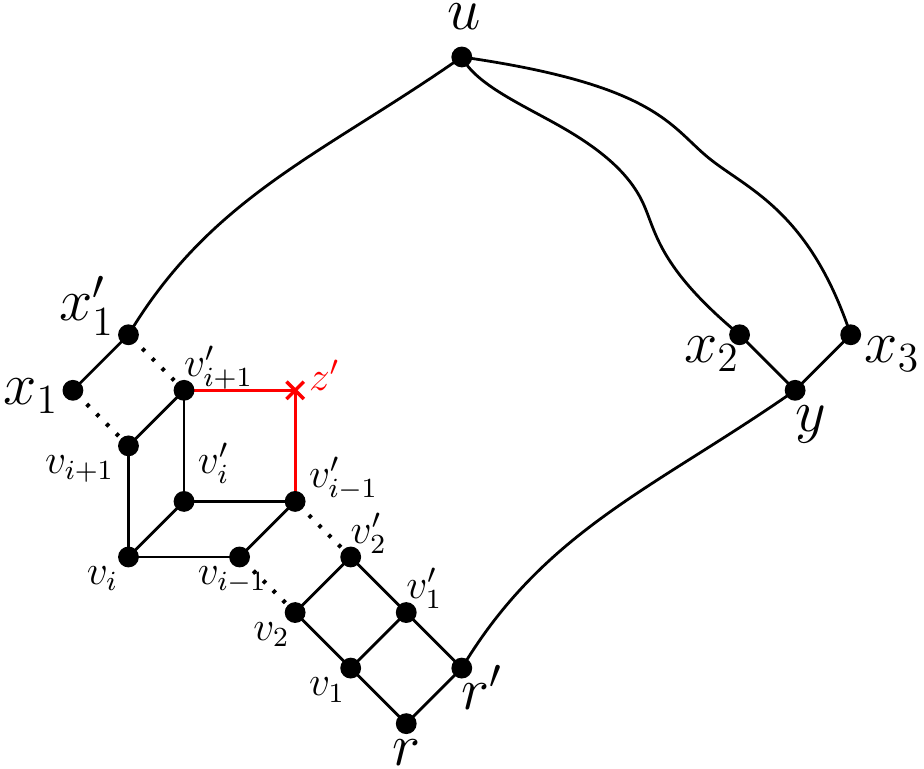}
    \end{minipage}
    \caption{\label{fig_lemma_2gates}
        Cases 1 and 2 of Lemma
        \ref{lem_2_gates_Uborders}.
    }
\end{figure}

Applying  Lemmas \ref{lem_Uborders_cube-free} and \ref{lem_2_gates_Uborders} to the cube-free
median subgraph of $G$ induced by the fiber $F(x)$, we immediately  obtain:

\begin{corollary} 
    \label{proj-total-boundary} 
    The total boundary $\partial^*F(x)$ of any fiber $F(x)$ is quasigated.
\end{corollary}

\subsection{Classification of pairs of vertices}
\label{classification-pairs}

In Section \ref{classification-fibers} we classified the fibers of
$\St(z)$ into panels and cones. Now, we use this classification
to provide a classification of pairs of vertices of $G$ with respect to the 
partition into fibers, which extends the one done in \cite{chepoi2006distance}
for planar median graphs.

Let $z$ be an arbitrary fixed vertex of a cube-free median graph $G=(V,E)$.
Let ${\mathcal F}_z=\{ F(x): x\in \St(z)\}$ denote the partition of $V$ into 
the fibers of $\St(z)$.

Let $u,v$ be two arbitrary vertices of $G$ and suppose that $u$ belongs to the fiber $F(x)$ and $v$ belongs
to the fiber $F(y)$ of ${\mathcal F}_z$. We say that $u$ and $v$ are {\it
\roommates} if they belong to the same fiber, i.e., $x=y$.
We  say that $u$ and $v$ are {\it \neighboring} if $F(x)$ and $F(y)$ are two
neighboring fibers (then one of them is a panel and another is a cone).
We say that $u$ and $v$ are {\it \aNeighboring} if $F(x)$ and $F(y)$ are
distinct cones neighboring  with a common panel, i.e., there exists a panel
$F(w)\sim F(x),F(y)$.
Finally, we say that $u$ and $v$ are {\it \separated} if the fibers $F(x)$ and
$F(y)$ are distinct, are not neighboring, and if both $F(x)$ and $F(y)$ are
cones,
then they are not \aNeighboring. For an illustration, see Figure
\ref{fig_fibers}.
From the definition it easily follows that any two vertices $u,v$ of $G$ are 
either \roommates, or  \separated, or \neighboring, or \aNeighboring.
Notice also the following transitivity property of this classification: if $u'$ 
belongs to the same fiber $F(x)$ as $u$ and $v'$ belongs to the same fiber 
$F(y)$ as $v$, then $u',v'$ are classified in the same category as $u,v$.

We continue with distance formulae for \separated, \aNeighboring, and
\neighboring vertices.
The illustration of each of the formulae is provided in Figure
\ref{fig_shortest_paths}.

\begin{lemma}
    \label{lem_shortest_path_separated} 
    For vertices $u$ and $v$ belonging  to the fibers $F(x)$ and $F(y)$ of 
    $\St(z)$, respectively, the following conditions are equivalent:
    \begin{enumerate}[(i)]
        \item $u$ and $v$ are \separated;
        \item $I(x,z)\cap I(y,z)=\{ z\}$;
        \item $d_G(u,v)=d_G(u,z)+d_G(z,v)$, i.e., $z\in I(u,v)$.
    \end{enumerate}
\end{lemma}

\begin{proof} 
    (i)$\Longleftrightarrow$(ii): Notice that $u$ and $v$ are \separated if and 
    only if $x\ne y$ and either $F(x)$ and $F(y)$ both are panels, or both are 
    cones not having a neighboring panel, or one is a cone and another is a 
    panel and the cone and the panel are not neighboring.
    If $F(x)$ and $F(y)$ are panels, then $I(x,z)=\{ x,z\}$ and $I(y,z)=\{ 
    y,z\}$, thus  $I(x,z)\cap I(y,z)=\{ z\}$.
    If $F(x)$ and $F(y)$ are cones, then $I(x,z)$ and $I(y,z)$ are two squares 
    $Q_x$ and $Q_y$.  By Lemma \ref{lem_gated_borders}, $Q_x$ and $Q_y$ 
    intersect in an edge $wz$ if and only if $F(w)$ is a panel neighboring 
    $F(x)$ and $F(y)$, i.e., if and only if $u$ and $v$ are not \separated. 
    Finally, if $F(x)$ is a cone and $F(y)$ is a panel, then $I(x,z)$ is the 
    square $Q_x$ and $I(y,z)$ is the edge $yz$. Then  $F(x)$ and $F(y)$ are not 
    neighboring if and only if  the edge $yz$ is not an edge of the square 
    $Q_x$, i.e., if and only if $I(x,z)\cap I(y,z)=\{ z\}$.

    (ii)$\Longleftrightarrow$(iii): First, suppose that $I(x,z)\cap I(y,z)=\{ 
    z\}$. To show that $z\in I(u,v)$ it suffices to prove that $z$ is the 
    median of $u,v,z$.
    Suppose by way of contradiction that the median of $u,v,z$ is the vertex 
    $w\ne z$. Let $s$ be a neighbor of $z$ in $I(z,w)$. Then obviously $s\in 
    \St(z)$.
    Since $I(x,z)\cap I(y,z)=\{ z\}$, $s$ does not belong to at least one of 
    the intervals $I(x,z)$ and $I(y,z)$, say $s\notin I(x,z)$. This implies 
    that $d_G(s,x)=d_G(z,x)+1$.
    Since $x$ is the gate of $u$ in $\St(z)$ and $s\in \St(z)$, necessarily 
    $x\in I(u,s)$. This implies that there is a shortest $(s,u)$-path passing 
    via $z$ and $x$, i.e., $d_G(s,u)=1+d_G(z,u)$.  On the other hand, since 
    $s\in I(z,w)\subset I(z,u)$, we conclude that $d_G(z,u)=1+d_G(s,u)$. 
    Comparing the two equalities, we obtain a contradiction.

    Conversely, suppose that $z\in I(u,v)$. This implies that $z$ is the median 
    of the triplet $u,v,z$ and that $I(u,z)\cap I(v,z)=\{ z\}$. Since $x$ is 
    the gate of $u$ and $y$ is the gate of $v$ in $\St(z)$, we conclude that
    $x\in I(u,z)$ and $y\in I(v,z)$. Consequently, $I(x,z)\subseteq I(u,z)$ and 
    $I(y,z)\subseteq I(v,z)$, proving that $I(x,z)\cap I(y,z)=\{ z\}$.  This 
    establishes (iii)$\Longrightarrow$(ii).
\end{proof}

\begin{remark} 
    The equivalence (ii)$\Longleftrightarrow$(iii) of Lemma  
    \ref{lem_shortest_path_separated} holds for all median graphs.
\end{remark}

\begin{lemma} 
    \label{lem_shortest_path_neighboring}
    Let $u$ and $v$ be two \neighboring vertices such that $u$ belongs to the
    \panel $F(x)$ and $v$ belongs to the \cone $F(y)$. Let $u_1$ and $u_2$ be 
    the two imprints  of $u$ on the total boundary $\partial^* F(x)$ (it may 
    happen that $u_1=u_2$) and let $v^+$ be the gate of $v$ in  $F(x)$.
    Then
    $$
        \dist_G(u,v) =
            \min\{
                \dist_G(u,u_1) + \dist_{\partial^* F(x)}(u_1,v^+),
                \dist_G(u,u_2) + \dist_{\partial^* F(x)}(u_2,v^+)
            \}
            + \dist_G(v^+,v).
    $$
\end{lemma}

\begin{proof}
    By Lemma \ref{fiber-gated1} $F(x)$ is  gated. Hence there must exist a
    shortest $(u,v)$-path passing via  $v^+$. 
    The vertices $u_1,u_2,$ and $v^+$ belong to the total boundary $\partial^* 
    F(x)$ of $F(x)$. Since, by Lemma \ref{lem_isometric_Uborders}, $\partial^* 
    F(x)$ is an isometric tree  and since, by Lemma \ref{lem_2_gates_Uborders},
    $u$ has at most two imprints $u_1$ and $u_2$ in $\partial^* F(x)$, we 
    conclude that $\dist_G(u,v^+) = \min\{ \dist_G(u,u_1) + \dist_{\partial^* 
    F(x)}(u_1,v^+), \dist_G(u,u_2) + \dist_{\partial^* F(x)}(u_2,v^+) \}$. 
    Consequently, there is a shortest $(u,v)$-path passing first via one of the 
    vertices $u_1,u_2$ and  then via $v^+$, establishing the asserted property.
\end{proof}

\begin{lemma} 
    \label{lem_shortest_path_aNeighboring}
    Let $u$ and $v$ be two \aNeighboring vertices belonging to the \cones 
    $F(x)$ and $F(y)$, respectively, and  let $F(w)$ be the \panel neighboring 
    $F(x)$ and $F(y)$. 
    Let $u^+$ and $v^+$ be the gates of $u$ and $v$ in $F(w)$.
    Then $\dist_G(u,v) =\dist_G(u,u^+)+\dist_{\partial^* F(w)}(u^+,v^+)+\dist_G(v^+,v).$
\end{lemma}

\begin{proof} 
    Since the halfspace $W(w,z)$ is convex and $u,v\in F(x)\cup F(w)\cup 
    F(y)\subset W(w,z)$, any shortest $(u,v)$-path $P(u,v)$ is contained in 
    $W(w,z)$. We assert that $P(u,v)\subset F(x)\cup F(w)\cup F(y)$.
    Indeed, since $u\in F(x)$, $v\in F(y)$ and the fibers $F(x), F(y)$ are not 
    neighboring, while moving from $u$ to $v$ along $P(u,v)$, we have to leave 
    $F(x)$ and enter a panel  neighboring $F(x)$. But the cone $F(x)$ has only 
    two neighboring panels: $F(w)$ and a panel $F(w')\subset W(z,w)$. Since 
    $P(u,v)\subset W(w,z)$, necessarily $P(u,v)$ must enter $F(w)$ (and not 
    $F(w')$). Analogously, one can show that while moving from $v$ to $u$ along 
    $P(u,v)$ when we leave $F(y)$ we must enter the same panel $F(w)$. 
    Consequently, since the fibers $F(x),F(w)$, and $F(y)$ are gated, the path 
    $P(u,v)$ must be included in their union.

    Next we show that $u^+$ and $v^+$  belong to a common shortest 
    $(u,v)$-path. Indeed, by what has been shown above, any shortest 
    $(u,v)$-path intersects $F(w)$, in particular, there  exists a vertex $s\in 
    I(u,v)\cap F(w)$. Since $u^+$ is the gate of $u$ in $F(w)$ and $v^+$ is the 
    gate of $v$ in $F(w)$, we  deduce that $u^+\in I(u,s)$ and $v^+\in I(v,s)$. 
    Since $s\in I(u,v)$, there exists a shortest path from $u$ to $v$ passing 
    via $u^+,s,$ and $v^+$. This shows that $\dist_G(u,v) = 
    \dist_G(u,u^+)+\dist_{G}(u^+,v^+)+\dist_G(v^+,v).$ Since $\partial^* F(w)$ 
    is an isometric tree, $d_G(u^+,v^+)=d_{\partial^* F(w)}(u^+,v^+)$, 
    establishing the required equality $\dist_G(u,v) = 
    \dist_G(u,u^+)+\dist_{\partial^* F(w)}(u^+,v^+)+\dist_G(v^+,v)$.
\end{proof}

\begin{figure}[ht]
    \begin{center}
        \begin{minipage}{0.3\linewidth}
            \centering
            \separated vertices
        \end{minipage}
        \begin{minipage}{0.3\linewidth}
            \centering
            \neighboring vertices
        \end{minipage}

        \begin{minipage}{0.3\linewidth}
            \centering
            \includegraphics[width=\linewidth]{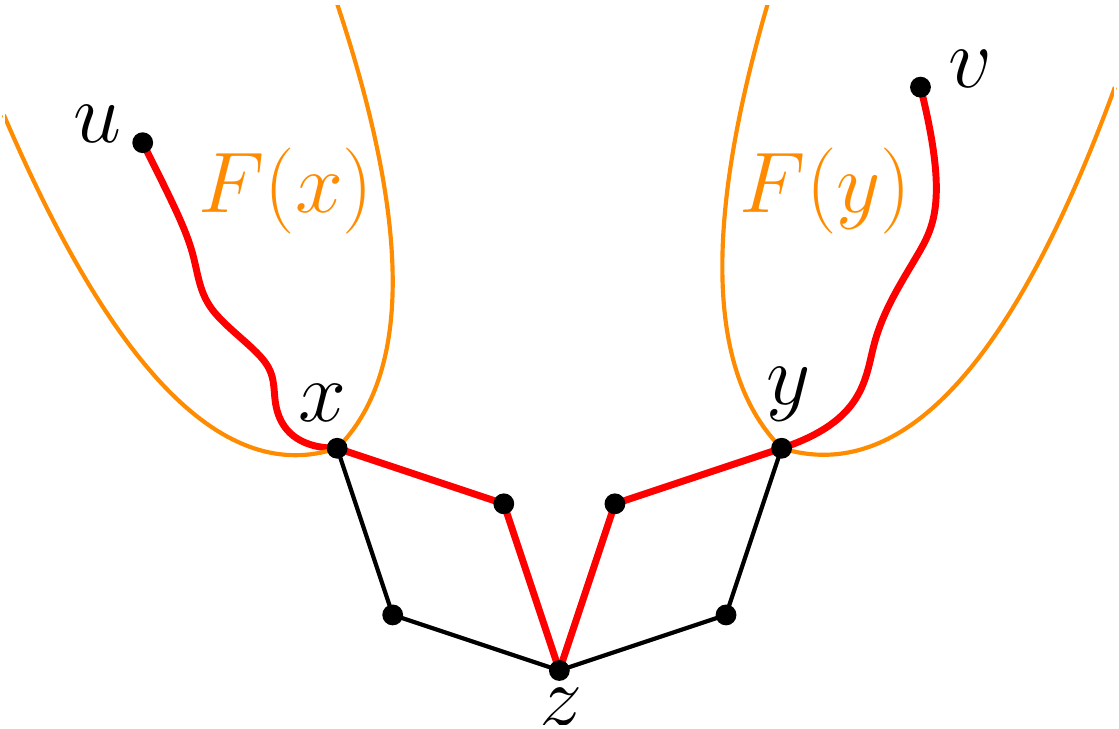}
        \end{minipage}
        \begin{minipage}{0.3\linewidth}
            \centering
            \includegraphics[width=\linewidth]{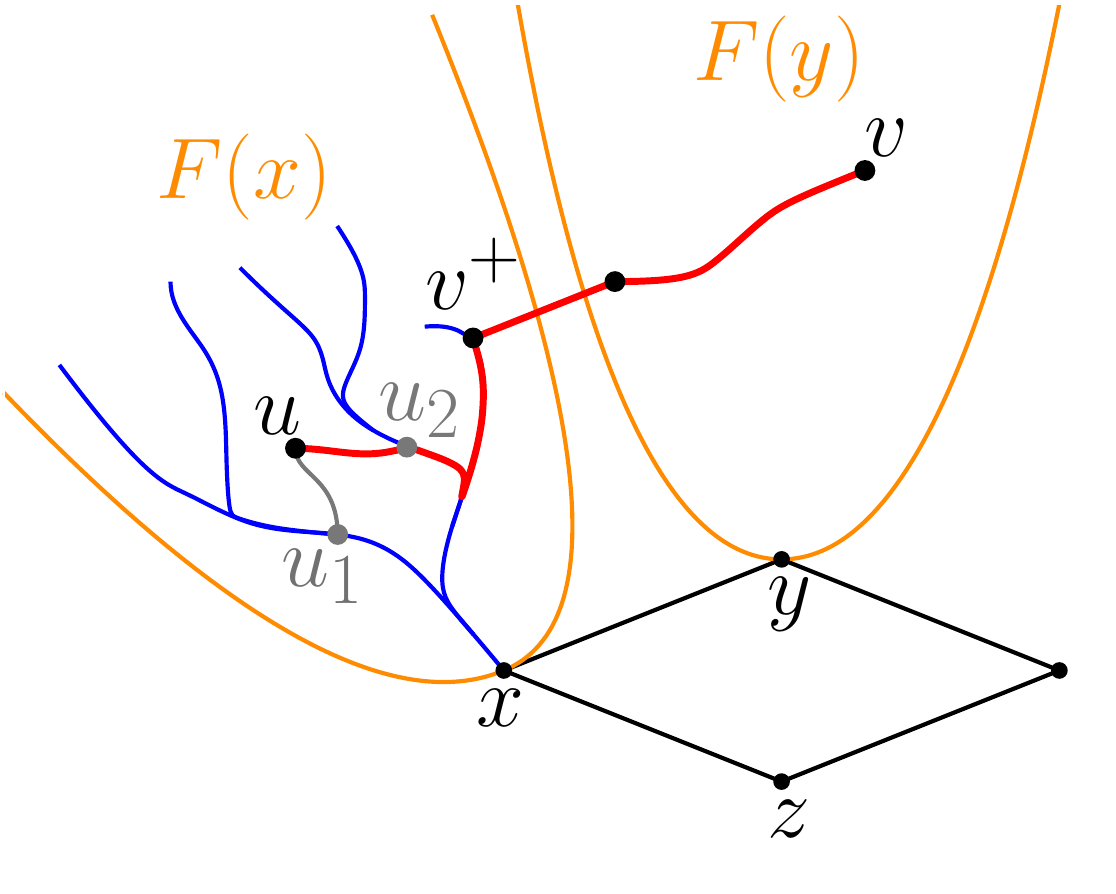}
        \end{minipage}
    \end{center}
    \centering
    \aNeighboring vertices

    \includegraphics[width=0.3\linewidth]{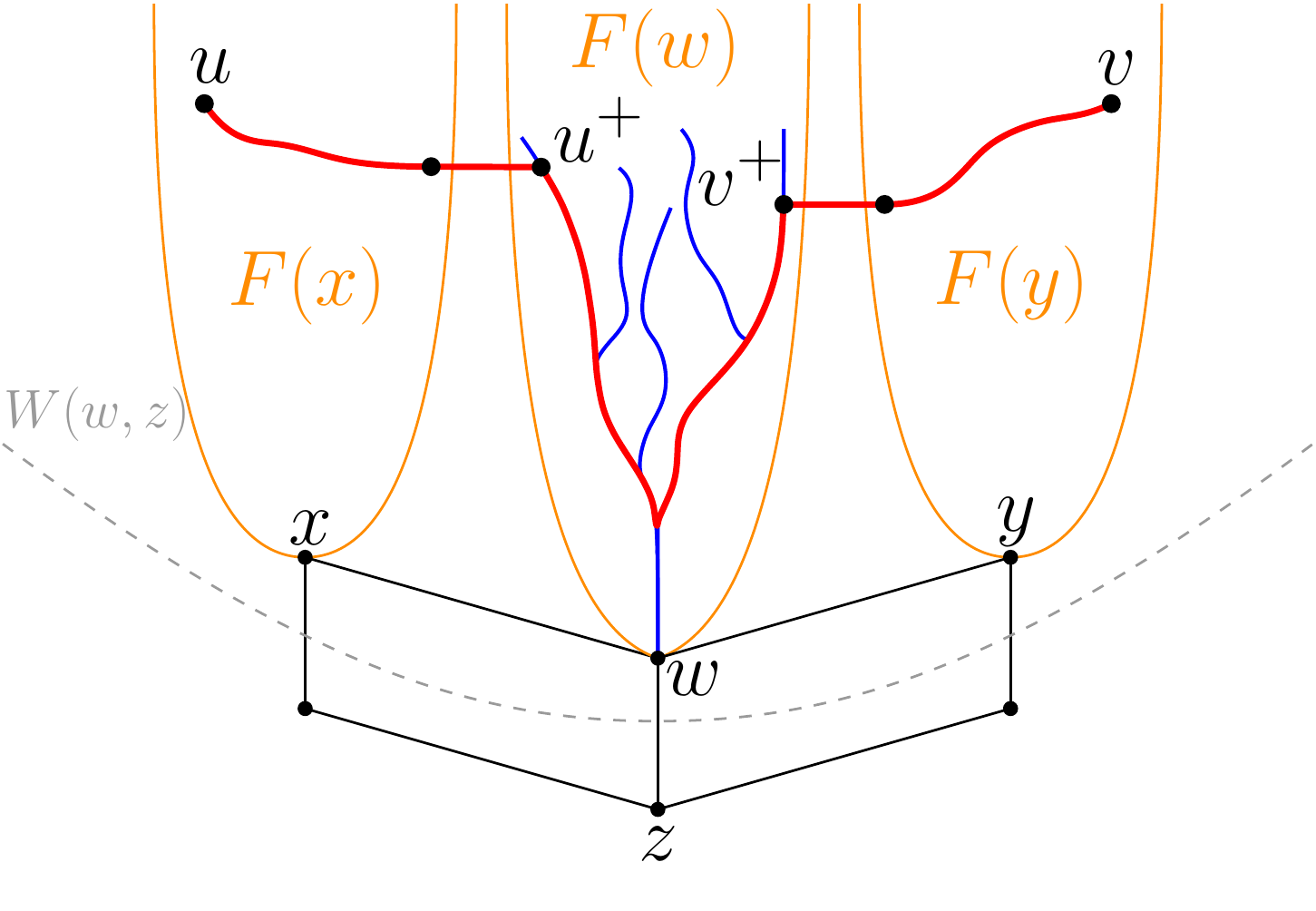}
    \caption{
        \label{fig_shortest_paths}
        An illustration of Lemmas \ref{lem_shortest_path_separated},
        \ref{lem_shortest_path_neighboring} and
        \ref{lem_shortest_path_aNeighboring}: examples of shortest paths (in
        red) between \separated, \neighboring, and \aNeighboring vertices $u$
        and $v$. The total boundaries of the \panels appear in blue.
    }
\end{figure}

\section{Distance labeling scheme for cube-free median graphs}
\label{sect_dist_labeling}

Let $G=(V,E)$ be a cube-free median graph with $n$ vertices and let $c$ be a centroid of $G$. Let $u,v$ be any pair of vertices of $G$ for which we have to compute the distance $d_G(u,v)$.
Applying Lemmas \ref{lem_shortest_path_separated}, \ref{lem_shortest_path_neighboring}, and \ref{lem_shortest_path_aNeighboring} of previous section with $m$ instead of $z$, the distance $d_G(u,v)$
can be computed once $u$ and $v$ are \separated, \neighboring, or \aNeighboring
and once $u$ and $v$ keep in their labels the distances to $c$, to the
respective gates $u^+$ and $v^+$, and to the
imprints $u_1$ and $u_2$ if $u$ belongs to a panel. It also requires keeping in the labels of $u$ and $v$ the information necessary to compute each of the distances $\dist_{\partial^* F(x)}(u_1,v^+),
\dist_{\partial^* F(x)}(u_2,v^+),\dist_{\partial^* F(w)}(u^+,v^+)$. Since the total boundaries are isometric trees, this can be done by keeping in the label of $u$ the labels of $u_1,u_2$, and $u^+$
in a distance labeling scheme for a tree (as well as keeping in the label of $v$ such a label of $v^+$). This shows that $d_G(u,v)$ can be computed in all cases except when $u$ and $v$
are \roommates, i.e., they belong to a common fiber $F(x)$ of $\St(c)$.  Since
$F(x)$ is gated and thus median,
we can apply the same recursive procedure to each fiber $F(x)$ instead of $G$. Therefore, $d_G(u,v)$ is computed in the first recursive call when $u$ and $v$ will no longer belong to the same
fiber of the current centroid. Since at each step the division into fibers is performed with respect to a centroid, $|F(x)|\le n/2$ by Lemma \ref{lem_small_halfspaces}, thus the tree of
recursive calls has logarithmic depth.

In this section, we present the formal description of the distance labeling scheme. The encoding scheme is described by the algorithm \distenc{} presented in Section
\ref{sect_dist_encoding}. Section \ref{sect_dist_queries} presents the algorithm
\distdec{} used for answering distance queries. In Section \ref{sect_labeling_trees_stars}
we formally present the distance labeling schemes for trees and stars.

\subsection{Distance and routing labelings for trees and stars}
\label{sect_labeling_trees_stars}

We present the distance labeling scheme (\distEncTree, \distDecTree)
for trees, which we briefly  described in Section \ref{sect_prelim_dist_rout}.
The procedure \distEncTree{} that gives a label $\LDt{v}{T}$ to every vertex
$v$ of a tree $T$ works as follows:
\begin{enumerate}[{\hskip1em(}1{)}]
    \item Give to every vertex $v$ a unique identifier $\id(v)$;
    \item Find a centroid $c$ of $T$;
    \item For every vertex $v$ of $T$, concatenate $(\id(c), d_T(v,c))$ to the
    current label of $v$;
    \item Repeat Step $2$ for each subtree with at least two vertices, created
    by the removal of $c$.
\end{enumerate}
Given two labels $\LDt{u}{T}$ and $\LDt{v}{T}$, the procedure
\distDecTree{} can find the last common separator $c$ of $u$ and $v$ (i.e.,
the common centroid stored latest in their labels) and return $d_T(u,c) +
d_T(v,c)$ as the distance $d_T(u,v)$.
The encoding \routEncTree{} for routing in trees is similar, just
 replace $(\id(c), d_T(u,c))$ at step (3) by $(\id(c),
\port(u,c), \port(c,u))$.
Then, the decoding function \routDecTree($u$,$v$), returns $\port(u,c)$
(stored in the label of $u$) if $c \neq u$, or $\port(c,v)$ (stored in the
label of $v$) otherwise (where $c$ is again the last common separator of $u$
and $v$).

We present the distance labeling scheme for stars $\St(z)$ of
any median graph $G$. It is based on the fact that median graphs are isometrically
embeddable into hypercubes and that $\St(z)$ is  gated, and thus is an isometric median
subgraph of $G$. So, we can suppose that $\St(z)$ is isometrically embedded into a
hypercube. Let $\varphi:\St(z)\rightarrow Q_d$ be such an isometric embedding so that $\varphi(z)=\varnothing$.
Consequently, for each vertex $x$ of $\St(z)$, $\varphi(x)$ is a set of cardinality equal to the dimension of the
cube $I(x,z)$, thus $\varphi(x)$ has size at most
$\log_2 n$, where $n=|\St(z)|$. For any two vertices, $x$ and $y$ of $\St(z)$,
$d_{\St(z)}(x,y)=|\varphi(x)\Delta \varphi(y)|:= |(\varphi (x) \cup \varphi 
(y)) \setminus (\varphi (x) \cap \varphi (y))|$.

Using the isometric embedding $\varphi$, we can describe a simple
encoding \encStar($\St(z)$) of the vertices of $\St(z)$ which can be
used to answer distance and routing queries.
For a vertex $x\in \St(z)$, let $\Lb{x}{\St(z)}=\varphi (x)$.
Then \encStar($\St(z)$) gives to $z$ the label $\varnothing$ and to every
neighbor of $z$ a unique label in $\{1,\ldots,\degree(z)\}$.
For any vertex $x$ at distance $k$ from $z$, $I(x,z)$
contains exactly $k$ neighbors of $z$ and the labels of these neighbors
completely define $\varphi(x)$ and  $\Lb{x}{\St(z)}$.

Giving unique labels to the neighbors of $z$ require
$\ceil{\log_2(\degree(z))}$ bits and thus, in the worst case, \encStar($\St(z)$)
gives labels of length $O(\degree(z)\log(\degree(z)))$. If the dimension of
$\St(z)$ is a fixed constant, then \encStar($\St(z)$) gives labels of 
logarithmic length.
For a vertex $x$ of $\St(z)$ labeled by the set $X:=\varphi(x)$, the vertex of 
$\St(z)$ labeled by the value $\min \{i : i \in X\}$ is called the 
\emph{$\lleft$ of $x$}, and the one labeled by $\max \{i : i \in X\}$ is called 
the \emph{$\rright$ of $x$}.

For simplicity, we assume that for a vertex $x$ labeled $X$ and a vertex $x'$
labeled $X' = X \setminus \{i\}$, $i \in X$, we have $\port(x,x') = \port(x',x)
= i$.
Since $\varphi$ is an isometric embedding, it is easy to see that for any two
vertices $x$ and $y$ encoded by the sets $X:=\Lb{x}{\St(z)}=\varphi(x)$ and
$Y:=\Lb{y}{\St(z)}=\varphi(y)$, the distance $d_{\St(z)}(x,y)$ between $x$ and
$y$ is $|X
\triangle Y| $.
This is exactly the value returned by \distDecStar($X$, $Y$).
Routing decisions follow from the same property. Assume that $|X| \leq |Y|$.
If $X \subseteq Y$, \routDecStar($X$,$Y$) returns the port to any vertex
labeled by $X \cup \{i\}$ with $i \in Y \setminus X$ (say the minimal $i$).
If $X \not\subset Y$, then \routDecStar($X$,$Y$) returns the port of any
vertex labeled by $X \setminus \{i\}$ for $i \in X \setminus Y$ (say the
minimal $i$).

\begin{figure}[ht]
    \centering
    \includegraphics[width=0.4\textwidth]{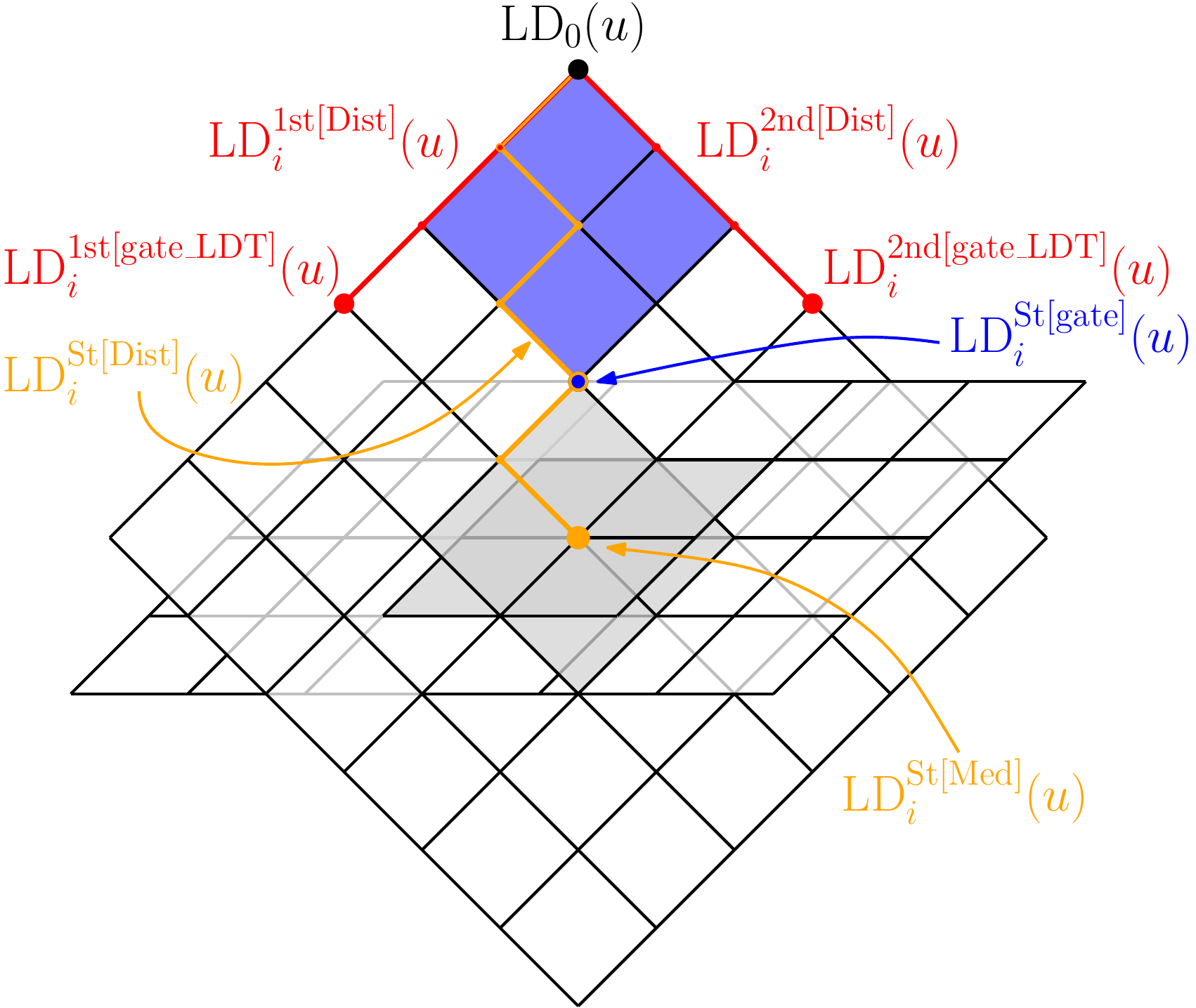}
    \caption{\label{fig_dist_labeling}
        Illustration of $\LD_0(u)$ and of the information  added to $\LD(u)$
        at step $i$.
    }
\end{figure}

\subsection{Encoding}
\label{sect_dist_encoding}

Let $G=(V,E)$ be a cube-free median graph with $n$ vertices. We describe now
how \distenc{} constructs for every vertex $u$ of
$G$ a distance label $\LD(u)$. This is done recursively and every depth of the recursion will be called a \emph{step}.
Initially, we suppose that every vertex $u$ of $G$ is given a unique identifier
$\id(u)$. We define this naming step as Step $0$ and we denote the
corresponding part of $\LD(u)$ by $\LD_0(u)$ (i.e., $\LD_0(u) := \id(u)$).
At Step $1$, \distenc{} computes a centroid $c$ of $G$, the star $\St(c)$ of $c$, and the partition ${\mathcal F}_c:=\{ F(x): x\in \St(c)\}$ of $V$
into fibers. Every vertex $u$ of $G$ ($c$ included) receives the
identifier $\id(c)$ of $c$ and its distance $\dist_G(u,c)$ to $c$.
After that, every vertex $x$ of $\St(c)$ receives a special identifier
$\Lb{x}{\St(c)}$ of size $O(\log |V|)$ consisting of a labeling for the star $\St(c)$,
as described in Section \ref{sect_labeling_trees_stars}.
Then, \distenc{} computes the gate $u^\downarrow$ in $\St(c)$ of every vertex
$u$ of $G$ and adds its identifier $\Lb{u^\downarrow}{\St(c)}$ to $\LD(u)$.
Note that the identifiers $\Lb{x}{\St(c)}$ of the vertices of $\St(c)$ can
also be used to distinguish the fibers of $\St(c)$.
This triple $(\id(c), \dist_G(u,c), \Lb{u^\downarrow}{\St(c)})$ contains the
necessary information relative to $\St(c)$ and is thus referred as the \emph{part
    ``star''} of the information $\LD_1(u)$ given to $u$ at Step $1$.
We denote this part by $\LD_1^{\st}(u)$. We also set
$\LD_1^{\st[\med]}(u) := \id(c)$, $\LD_1^{\st[\distance]}(u) :=
\dist_G(u,c)$ and $\LD_1^{\st[\rootB]}(u) := \Lb{u^\downarrow}{\St(c)}$
for the three components of the label $\LD_1^{\st}(u)$.

\medskip
\scalebox{0.9}{\begin{algorithm}[H]
    \caption{\label{alg_distenc}\distenc[]($G$, $\LD(V)$)}
    \Input{
        A cube-free median graph $G=(V,E)$ and a labeling $\LD(V)$,
        initially consisting of a unique identifier $\id(u)$ for every $u
        \in V$
    }

    \lIf{$V = \{v\}$}{ \Stop }

    \BlankLine
    Find a centroid $c$ of $G$ \;
    Compute the partition ${\mathcal F}_c$ of $G$ into fibers \;
    For each fiber $F(x)$ of  ${\mathcal F}_c$ compute its total boundary 
    $\partial^* F(x)$ \;
    $\Lb{\St(c)}{\St(c)}$ $\leftarrow$ \encStar($\St(c)$) \;
    \ForEach{\panel $F(x) \in {\mathcal F}_c$}{
        $\LDt{\partial^* F(x)}{\partial^* F(x)}$ $\leftarrow$ 
        \distEncTree($\partial^* F(x)$)\;

        \ForEach{$u \in F(x)$}{
        Find the gate $u^\downarrow$ of $u$	in $\St(c)$ \;
        Find the imprints $u_1$ and $u_2$ of $u$ on $\partial^* F(x)$;
        $(d, ~d_1, ~d_2) \leftarrow (\dist_G(u,c), ~\dist_G(u,u_1),
        ~\dist_G(u,u_2))$ \;
        $L_\st \leftarrow (\id(c), d,
        \Lb{u^\downarrow}{\St(c)})$ \;\label{alg-line_dist_enc_median1}
        $L_{\lleft} \leftarrow (\LDt{u_1}{\partial^* F(x)},d_1)$
        \;\label{alg-line_dist_enc_left1}
        $L_{\rright} \leftarrow (\LDt{u_2}{\partial^* F(x)},d_2)$
        \;\label{alg-line_dist_enc_right1}
        $\LD(u) \leftarrow \LD(u) \conc
        (L_\st,L_{\lleft},L_{\rright})$\;
        }
        \distenc[]($F(x)$, $\LD(V)$) \;
    }
    \ForEach{\cone $F(x) \in {\mathcal F}_c$}{
          \ForEach{$u \in F(x)$}{
            Find the gate $u^\downarrow$ of $u$	in $\St(c)$ \;
            Find the panels $F(w_1)$ and $F(w_2)$ neighboring $F(x)$\;
            Find the gates $u^+_1$ and $u^+_2$ of $u$ in $F(w_1)$ and $F(w_2)$ 
            \;
            $(d, ~d_1, ~d_2) \leftarrow (\dist_G(u,c), ~\dist_G(u,u^+_1),
            ~\dist_G(u,u^+_2))$ \;
            $L_\st \leftarrow (\id(c), d,
            \Lb{u^\downarrow}{\St(c)})$ \;\label{alg-line_dist_enc_median2}
            $L_{\lleft} \leftarrow (\LDt{u^+_1}{\partial^* F(w_1)},d_1)$
            \;\label{alg-line_dist_enc_left2}
            $L_{\rright} \leftarrow (\LDt{u^+_2}{\partial^* F(w_2)},d_2)$
            \;\label{alg-line_dist_enc_right2}
            $\LD(u) \leftarrow \LD(u) \conc
            (L_\st,L_{\lleft},L_{\rright})$ \;
        }
        \distenc[]($F(x)$, $\LD(V)$). \
    }
\end{algorithm}}

Afterwards, at Step 1, the algorithm considers each fiber $F(x)$ of ${\mathcal F}_c$.
If $F(x)$ is a panel, then the algorithm computes the total boundary $\partial^* F(x)$ of $F(x)$, which is an isometric quasigated
tree. The vertices $v$ of this tree $\partial^* F(x)$  are given special
identifiers $\LDt{v}{\partial^* F(x)}$ of size $O(\log^2 |V|)$ consisting of a distance
labeling scheme for trees described in Section \ref{sect_labeling_trees_stars}. For each vertex $u$
of the  panel $F(x)$, the algorithm computes the two imprints $u_1$ and $u_2$ of $u$ in $\partial^* F(x)$
(it may happen that $u_1=u_2$) and stores
$(\LDt{u_1}{\partial^* F(x)}, \dist_G(u,u_1))$ and $(\LDt{u_2}{\partial^* F(x)},
\dist_G(u,u_2))$ in $\LD_1^{\lleft}(u)$ and $\LD_1^{\rright}(u)$.

If $F(x)$ is a cone and $F(w_1),F(w_2)$ are the two panels neighboring
$F(x)$, then for each vertex $u$ of $F(x)$, the algorithm computes the gates
$u^+_1$ and $u^+_2$ of $u$ in $F(w_1)$ and $F(w_2)$.
Since $u^+_i\in \partial_x F(w_i)\subset \partial^* F(x), i=1,2,$ the labels $\LDt{u^+_1}{\partial^* F(w_1)}$ and $\LDt{u^+_2}{\partial^* F(w_2)}$
in the distance labelings of trees $\partial^* F(w_1)$ and $\partial^* F(w_2)$ are well-defined. Thus the algorithm stores $(\LDt{u^+_1}{\partial^* F(w_1)}, \dist_G(u,u^+_1))$ and $(\LDt{u^+_2}{\partial^* F(w_2)},
\dist_G(u,u^+_2))$ in $\LD_1^{\lleft}(u)$ and $\LD_1^{\rright}(u)$. This ends 
Step $1$.

Since ${\mathcal F}_c$ partitions $V$ into gated median subgraphs,  the label $\LD_2(u)$ added to $\LD(u)$ at Step $2$
is obtained as $\LD_1(u)$, where   $G$ is replaced by the fiber $F(u^\downarrow)$
containing $u$, and so on. Since each fiber contains
no more than half of the vertices of the current graph, at Step $\ceil{\log_2
|V|}$, each fiber  consists of a single vertex, and the
algorithm stops. Therefore,
for each pair of vertices $u$ and $v$ of $G$,    there exists a step of the recursion after which $u$ and $v$ are no longer \roommates.
For an illustration of the parts of $\LD_i(u)$, see Fig.\ref{fig_dist_labeling}. For a vector $L(v) := (t_1,\ldots,t_k)$ of vectors $t_1,\ldots,t_k$ and an
arbitrary vector $t$, we denote by $L(v) \conc t := (t_1,\ldots,t_k,t)$
the \emph{concatenation} of $L(v)$ and $t$.

\subsection{Distance queries}
\label{sect_dist_queries}

Let $u$ and $v$ be two vertices of a cube-free median graph $G=(V,E)$ and let
$\LD(u)$ and $\LD(v)$ be their labels returned by \distenc. Here we describe how
the algorithm \distdec{} can compute the information about the relative positions of $u$ and $v$
with respect to each other and how, using it, to compute the distance $d_G(u,v)$.

\label{sect_dist_algo}
We continue with the formal description of the algorithm \distdec{}.
The functions \DistanceNeighboring, \DistanceANeighboring, and 
\DistanceSeparated, used in this algorithm, are given below
(the function \distDecStar  ~is described in Section 
\ref{sect_labeling_trees_stars}).

Given the vertices $u$ and $v$, first the algorithm  detects if  $u$ and $v$ 
coincide. This is done in line 1 of  \distdec{}. If $u\ne v$, then  \distdec{} 
finds the largest integer $i$ such that $\LD_i^{\st[\med]}(u) = 
\LD_i^{\st[\med]}(v)$ (line 2).
This corresponds to the first time $u$ and $v$ belong to different fibers in a  
partition.  Let $c$ be a centroid vertex of the current graph.
In lines 3,4,5, the algorithm \distdec{} retrieves  the distances $d,d_u,$ and 
$d_v$ between the gates $u^\downarrow$ and $v^\downarrow$ of $u$ and $v$ in the 
star  $\St(c)$,  and the distances from $u^\downarrow$ and $v^\downarrow$ to  
$c$, respectively.
This is done by using the identifiers $\LD_i^{\st[\rootB]}(u)$ and 
$\LD_i^{\st[\rootB]}(v)$ and the distance decoder for distance labeling in 
stars.  With this information at hand, one  can easily decide for each of $u$ 
and $v$ if  it belongs to a cone or to a panel, and, moreover, to  decide if 
the vertices $u$ and $v$ are \neighboring, \aNeighboring, or \separated. In 
each of these cases, a call to an appropriate function is done in lines 6-9.

\medskip
\scalebox{0.91}{\begin{algorithm}[H]
        \caption{\label{alg_distdec}\distdec[]($\LD(u)$, $\LD(v)$)}
        \Input{The labels $\LD(u)$ and $\LD(v)$ of two vertices $u$ and $v$ of
            $G$}
        \Output{The distance between $u$ and $v$ in $G$}

        \BlankLine
        \lIf{$\LD_0(u) = \LD_0(v)$ \tcc*[h]{$u = v$}}{
            \Return $0$
        }

        \BlankLine
        Let $i$ be the largest integer such that $\LD_i^{\st[\med]}(u) = 
        \LD_i^{\st[\med]}(v)$ \;

        \BlankLine
        $d \leftarrow \distDecStar(\LD_i^{\st[\rootB]}(u),
        \LD_i^{\st[\rootB]}(v))$
        \tcp*{$\dist_G(u^\downarrow,v^\downarrow)$}
        $d_u \leftarrow \distDecStar(\LD_i^{\st[\rootB]}(u), 0)$
        \tcp*{$\dist_G(u^\downarrow,c)$}
        $d_v \leftarrow \distDecStar(\LD_i^{\st[\rootB]}(v), 0)$
        \tcp*{$\dist_G(v^\downarrow,c)$}

        \BlankLine
        \lIf{$d = 1$ and $d_u = 1$\hskip2em$~$}{
            \Return \DistanceNeighboring($\LD_i(u)$, $\LD_i(v)$)
        }
        \lIf{$d = 1$ and $d_v = 1$\hskip2em$~$}{
            \Return \DistanceNeighboring($\LD_i(v)$, $\LD_i(u)$)
        }
        \lIf{$d = 2$ and $d_u = d_v = 2$}{
            \Return \DistanceANeighboring($\LD_i(u)$, $\LD_i(v)$)
        }
        \Return \DistanceSeparated($\LD_i(u)$, $\LD_i(v)$). \
\end{algorithm}}

\label{sect_dist_functions}

\medskip

First suppose that the vertices $u$ and $v$ are \neighboring  ($d=1$ and one of
$d_u, d_v$ is 1 and another is 2), i.e., one of the vertices $u,v$
belongs to a cone and another one belongs to a panel,
and the cone and the panel are neighboring. The  function \DistanceNeighboring 
returns the distance $d_G(u,v)$ in the assumption that $u$ belongs to a panel 
and $v$ belongs to a cone (if $v$ belongs to a panel and $u$ to a cone,
it suffices to swap the names of the vertices $u$ and $v$ before using 
\DistanceNeighboring). The function finds the gate $v^+$ of $v$  in  the panel 
of $u$ by looking at $\LD_i^{\st[\rootB]}(v)$ (it also retrieves the distance 
$d_G(v,v^+)$).
It then retrieves  the imprint $u^*$ of $u$ (and the distance $d_G(u,u^*)$)
on the total boundary of the panel that minimizes the distance of $u$ to one of 
the two imprints plus the distance from this imprint to the gate $v^+$ using 
their tree distance labeling scheme.  Finally, \DistanceNeighboring returns 
$\dist_G(u,u^*)+d_G(u^*,v^+)+d_G(v^+,v)$ as $d_G(v,u)$.

\medskip
\scalebox{0.91}{\begin{myFunction}
        \newcommand{\dir}{\text{dir}}
        \Fn{\DistanceNeighboring{$\LD_i(u)$, $\LD_i(v)$}}{
            $\dir \leftarrow \lleft$ \tcp*[l]{If
                $\LD_i^{\st[\rootB]}(u) =
                \max\{i : i \in \LD_i^{\st[\rootB]}(v)\}$}
            \If{$\LD_i^{\st[\rootB]}(u) = \min\{i : i \in
                \LD_i^{\st[\rootB]}(v)\}$}{
                $\dir \leftarrow \rright$ \;	
            }
            $d_1 \leftarrow \distDecTree(\LD_i^{\dir[\gateD]}(v),
            \LD_i^{\lleft[\imprintD]}(u))$ \tcp*{The distance from the gate to the first imprint}
            $d_2 \leftarrow \distDecTree(\LD_i^{\dir[\gateD]}(v),
            \LD_i^{\rright[\imprintD]}(u))$ \tcp*{The distance from the gate to the second imprint}
            \Return $\LD_i^{\dir[\distance]}(v) + \min\left\{ d_1 +
            \LD_i^{\lleft[\distance]}(u),  d_2 +
            \LD_i^{\rright[\distance]}(u)\right\}$. \	
        }
\end{myFunction}}

\medskip\noindent
Now suppose that the vertices $u$ and $v$ are \aNeighboring (i.e.,
$d=d_u=d_v=2$). Then both $u$ and $v$ belong to cones.
By inspecting $\LD_i^{\st[\rootB]}(u)$ and
$\LD_i^{\st[\rootB]}(v)$, the function \DistanceANeighboring determines the \panel $F(w)$ sharing a border with the cones $F(u^\downarrow)$ and $F(v^\downarrow)$. Then the function
retrieves the gates $u^+$ and $v^+$ of $u$ and $v$ in this panel $F(w)$ and the distances $d_G(u,u^+)$ and $d_G(v,v^+)$.
The distance between the gates $u^+$ and $v^+$ is retrieved using the distance decoder
for trees. The algorithm returns $\dist_G(u,u^+)+\dist_G(u^+,v^+)+\dist_G(v^+,v)$ as $d_G(u,v)$.

\medskip
\scalebox{0.91}{\begin{myFunction}
        \newcommand{\dir}{\text{dir}}
        \Fn{\DistanceANeighboring{$\LD_i(u)$, $\LD_i(v)$}}{
            \ForEach{$x \in \{u,v\}$}{
                $\dir_x \leftarrow \lleft$ \tcp*{The common \panel is the
                    \lleft of the \cone of $x$}
                \If{$\LD_i^{\st[\rootB]}(u) \cap
                    \LD_i^{\st[\rootB]}(v) =
                    \min\{i : i \in
                    \LD_i^{\st[\rootB]}(x)\}$}{
                    $\dir_x \leftarrow \rright$ \tcp*{The common \panel is the
                        \rright of the \cone of $x$}
                }
            }
            $d \leftarrow \distDecTree(\LD_i^{\dir_u[\gateD]}(u),
            \LD_i^{\dir_v[\gateD]}(v))$ \;
            \Return $\LD_i^{\dir_u[\distance]}(u) +
            \LD_i^{\dir_v[\distance]}(v) +
            d$. \
        }
\end{myFunction}}

\medskip\noindent
In the remaining cases, the vertices $u$ and $v$ are \separated. By Lemma \ref{lem_shortest_path_separated},
there exists a shortest path between $u$ and $v$ passing
via $c$. Both $u$ and $v$ have stored the centroid $c$ and their distances to
$c$. Therefore, \DistanceSeparated simply returns the sum of those two
distances.

\medskip
\scalebox{0.91}{\begin{myFunction}
        \Fn{\DistanceSeparated{$\LD_i(u)$, $\LD_i(v)$}}{
            \Return $\LD_i^{\st[\distance]}(u) +
            \LD_i^{\st[\distance]}(v)$. \
        }
\end{myFunction}}

\subsection{The efficient implementation of  \distenc{}}
\label{sect_constr_dist_lab_scheme}

In this subsection we show how to implement a single run of the  algorithm  \distenc{} on an $n$-vertex cube-free median graph $G$ in $O(n)$ time.
Since  the algorithm is recursively called to
the fibers $F(x), x\in \St(c)$ and these fibers have size at most $\frac{n}{2}$ and are cube-free median graphs, the total complexity of  \distenc{} is $O(n\log_2 n)$. The main difficulty
with this is that we have to compute centroids, fibers, gates, and imprints 
without knowing the distance matrix of $G$ (with the distance matrix at hand, 
\distenc{}  can be easily implemented in $O(n^2\log n)$ time).

\subsubsection{Computation of a centroid $c$} 

We compute a centroid $c$ of $G$ using a recent linear-time algorithm of 
\cite{BeChChVa} for computing centroids/medians of arbitrary median graphs. For 
a median graph with $m$ edges this algorithm has complexity $O(m)$. By 
Corollary \ref{upper_edges}, $G$ contains at most $2n$ edges, thus a centroid 
$c$ of $G$ can be computed in $O(n)$ time.

\subsubsection{Partition of a median graph into fibers}  

We describe how to partition in $O(m)$ time any median graph $G$ with $m$ edges
into fibers with respect to any gated subgraph $H$ of $G$. For this, we adapt 
the classical Breadth-First-Search.

Recall that the \emph{Breadth-First Search (BFS)} \cite[Chapter 22]{Cormen} 
rooted at vertex $v_0$ uses a queue $Q$. For each vertex $v$ of $G$, two 
variables $d(v)$ and $f(v)$ are computed. 
Initially, $v_0$ is inserted  in $Q$, $d(v_0)$ is set to $0$, and $f(v_0)$ is set to null.
When a vertex $u$ arrives at the head of $Q$, it is removed from $Q$ and all 
the not yet discovered neighbors $v$ of $u$ are inserted in $Q$; $u$ becomes 
the \emph{parent} $f(v)$ of $v$ and $d(v)$ is set to $d(u)+1$. The edges 
$f(v)v$ define the \emph{BFS-tree} $T(v_0)$ of $G$. The main property of BFS is 
that $d(v)$ is the distance $\dist_G(v,v_0)$ and that $f(v)$ belongs to a 
shortest $(v,v_0)$-path. 

First, given any subgraph $H$ of any connected  graph $G$, we adapt the 
classical BFS to compute, for each
vertex $v$ of $G$, a closest to $v$ vertex of $H$, i.e., a vertex of $H$ 
realizing the distance $\dist_G(v,H) := \min \{ \dist_G(v,x): x\in V(H)\}$.  
For each vertex $v$ of $G$, the algorithm
computes the variables $\closest(v)$, $d(v)$, and $f(v)$.  The vertices of $H$ are first inserted in the
queue $Q$ and, for each $x\in V(H)$, we set
$\closest(x) := x, d(x) := 0,$ and $f(x)$ is set to null. When a vertex $u$ of 
$G$ arrives at the head of
$Q$, it is removed from $Q$ and all not yet discovered neighbors $v$
of $u$ are inserted in $Q$; $u$ becomes the \emph{parent} $f(v)$ of
$v$, $d(v)$ is set to $d(u)+1$, and $\closest(v)$ is set to $\closest(u)$.
We call this algorithm a \emph{BFS traversal of}  $G$ \emph{with respect to}  $H$.
Clearly, the algorithm has linear-time complexity $O(m)$.
Its correctness follows from  the following lemma:

\begin{lemma} 
    \label{BFS-dist-to-subgraph} 
    For any graph $G$, any subgraph $H$ of $G$, and any vertex $v$ of $G$, 
    $\closest(v)$ is a closest to $v$ vertex of $H$ and 
    $d(v)=\dist_G(v,\closest(v))=\dist_G(v,H)$.
\end{lemma}

\begin{proof} 
    The proof is inspired by the correctness proof of BFS; see, for example, 
    the proof of \cite[Theorem 22.5]{Cormen}.
    First, by induction on $d(v)$ one can easily show that $d(v)\ge 
    \dist_G(v,\closest(v))$. Indeed, let $u=f(v)$ and 
    $x=\closest(u)=\closest(v)$. By induction assumption,
    $d(u)\ge \dist_G(u,x)$. Since $d(v)=d(u)+1$ and $\dist_G(v,x)\le 
    \dist_G(u,x)+1$, we deduce that $d(v)=d(u)+1\ge \dist_G(u,x)+1\ge 
    \dist_G(v,x)$.
    Second, as in case of classical BFS, one can prove that at each execution  
    of the algorithm, if the queue $Q$ consists of the vertices 
    $v_1,v_2,\ldots,v_k$, then
    $d(v_1)\le d(v_2)\le \cdots\le d(v_r)$ and $d(v_r)\le d(v_1)+1$ hold. 
    Indeed, it can be easily seen that this invariant is preserved when
    a vertex is removed at the head of  $Q$ or is inserted at the end of $Q$.
    
    Using these two properties,  by induction on $k:=\dist_G(v,H)$ we will show 
    that  $\closest(v)$ is a closest to $v$ vertex of $H$ and
    that $d(v)=\dist_G(v,\closest(v))$. Suppose by way of contradiction that 
    $d(v)>k$, thus there exists a vertex $y\ne x$ of $H$ such that 
    $\dist_G(v,y)=k$.
    Let $u=f(v)$.  Let also $w$ be a neighbor of $v$ in $I(v,y)$. Since
    $\dist_G(w,H)=\dist_G(w,y)=k-1$, by induction hypothesis $\closest(w)=y$ 
    and  $d(w)=\dist_G(w,\closest(w))=\dist_G(w,y)\le k-1$.

    Consider the moment when the vertex $w$ is removed from  the queue $Q$. If 
    $v$ is not yet in $Q$, since $v\sim w$, the algorithm will pick $w$
    as the parent of $v$ and $y=\closest(w)$ as $\closest(v)$. This contradicts 
    the assumption that $u=f(v)$ and $\closest(v)=x\ne y$. On the other hand, 
    if $v$ is present
    in $Q$ or if $v$ has been already removed from $Q$, then necessarily 
    $u=f(v)$ was inserted in $Q$ before $w$ and  from the invariant of the 
    queue $Q$, we conclude that $d(u)\le d(w)\le k-1$. Therefore
    $d(v)=d(u)+1=k$, contrary to the assumption  $d(v)>k$. Hence 
    $d(v)=k=\dist_G(v,H)$. Since $d(v)$ is the length of a path between 
    $\closest(v)$ and $v$,
    we conclude that $d(v)=\dist_G(v,\closest(v))=\dist_G(v,H)$.
\end{proof}

We apply the previous algorithm to a gated subgraph $H$ of a median graph $G$. For a vertex $v$ of $G$ we use the same variables $d(v)$ and $f(v)$, but instead of $\closest(v)$ we use
the variable $\fib(v)$ which is updated in the same way as $\closest(v)$. To compute the fibers $\{ F(x): x\in V(H)\}$, for each vertex $x\in V(H)$,
we construct a BFS-tree $T(x)$ rooted at $x$ and consisting of all vertices $v$ such that $\fib(v)=x$ and of the edges of the form $vf(v)$.

\begin{corollary} 
    \label{BFS-dist-to-gated} 
    For any median graph $G$, any gated subgraph $H$ of $G$, and any vertex $v$ 
    of $G$, $\fib(v)$ is the gate of $v$ in $H$ and $d(v)=\dist_G(v,\fib(v))$. 
    For each vertex $x\in V(H)$, the vertex-sets of the tree $T(x)$ and of the 
    fiber $F(x)$ coincide.
\end{corollary}

\begin{proof} 
    The first assertion follows from Lemma \ref{BFS-dist-to-subgraph}. The 
    equality of vertex-sets of $T(x)$ and $F(x)$ immediately follows from the 
    definition of fibers and the first assertion.
\end{proof}

By Corollary \ref{BFS-dist-to-gated}, each $T(x)$ is a spanning tree of $F(x)$, whence the vertex-set of each fiber $F(x)$ is computed. To compute the edge-set of each $F(x)$, we traverse the edges of $G$
and each edge $uv$ such that $d(u)<d(v)$ and $\fib(u)=\fib(v)=:x$ is inserted in the fiber $F(x)$. If we traverse the edges in a BFS-order with respect to $H$, then we will get the adjacency lists
of the vertices from each fiber. The edges of $G$ not included in the fibers are the edges running between neighboring fibers. Therefore, if $uv$ is an edge such that $x:=\fib(u)\ne \fib(v)=:y$, then $u$ is inserted
in  $\partial_y F(x)$ (the boundary of $F(x)$ relative to $F(y)$) and $v$ is inserted in $\partial_x F(y)$ (the boundary of $F(y)$ relative to $F(x)$). For the computation of gates, it will be useful
that for vertex $u$ in $\partial_y F(x)$ we keep its neighbor $v\in \partial_x F(y)$ and for $v\in \partial_x F(y)$ we keep its neighbor $u\in \partial_y F(x)$; for this we set $\neighbor_{x,y}(u)=v$
and $\neighbor_{y,x}(v)=u$.  Additionally, $u$ is inserted in the total boundary  $\partial^* F(x)$ and $v$ is inserted in the total boundary $\partial^* F(y)$. This way, we constructed the
vertex-sets of relative  boundaries and of total boundaries
of all fibers. Consequently, for each vertex $x$ of $H$ and each vertex $v$ of $F(x)$ (i.e., such that $\fib(v)=x$) we can set $\tbound(v)=1$ if $v\in \partial^* F(x)$ and $\tbound(v)=0$
otherwise. Notice that each vertex $v\in \partial^* F(x)$ may be included in several relative boundaries $\partial_y F(x)$.
Since each such inclusion  corresponds to an edge incident to $v$, the total size of  all relative boundaries  is at most twice the number of edges of $G$, i.e., $O(m)$. The same conclusion
holds about the total size of the lists $\neighbor_{x,y}$ over all $x,y\in V(H), x\sim y$. Finally, to compute the edge-sets of all total boundaries $\partial^* F(x), x\in V(H)$, we traverse again  all edges of $G$ and
we insert a current edge $uv$ in the total boundary $\partial^* F(x)$ if and only if $x=\fib(u)=\fib(v)$ and $\tbound(u)=\tbound(v)=1$. Consequently,
we obtain the following result:

\begin{lemma} 
    \label{compute-fibers} 
    Given a median graph $G$ with $m$ edges and a gated subgraph $H$ of $G$, 
    the following entities can be computed in total linear time $O(m)$ ($O(n)$ 
    time if $G$ is cube-free):

    \begin{itemize}
        \item the vertex-sets and the edge-sets of all fibers $F(x), x\in V(H)$;
        \item the vertex-sets and the-edge sets of all total boundaries 
        $\partial^* F(x), x\in V(H)$;
        \item the vertex-sets of all relative boundaries $\partial_y F(x)$ and 
        the lists $\neighbor_{x,y}$ for all $x,y\in V(H)$ with $x\sim y$.
    \end{itemize}
\end{lemma}

\subsubsection{Computation of gates to neighboring panels}

Computation of gates is used in lines 9, 18, and 20 of  \distenc{}. In lines 9 
and 18, the gates of all vertices $v$ in the star $\St(c)$ are computed. This 
can be done by running a BFS traversal of $G$ with respect
to the gated set $\St(c)$. By Corollary \ref{BFS-dist-to-gated}, $\fib(v)$ is the gate of any vertex $v$ of $G$ in $\St(v)$. Therefore, the lines 9 and 18 can be executed in $O(n)$ time.

In line 20, for each vertex $u$ belonging to a cone $F(x)$ we have to compute the gates $u^+_1$ and $u^+_2$ of $u$ in the two neighboring panels $F(w_1)$ and $F(w_2)$. Notice that $u^+_1$ belongs to
the relative boundary $\partial_{x} F(w_1)$ and $u_2^+$ belongs to the
relative boundary $\partial_{x} F(w_2)$.  Consider the relative boundaries
$\partial_{w_1} F(x)$ and  $\partial_{w_2} F(x)$. They are
gated subgraphs of $F(x)$. We run two BFS traversals of $F(x)$, one with
respect to $\partial_{w_1} F(x)$  and the second one with respect to
$\partial_{w_2} F(x)$. For a vertex $u\in F(x)$, let $\gate_1(u)$ and
$\gate_2(u)$ be the
two gates of $u$ in $\partial_{w_1} F(x)$  and in  $\partial_{w_2} F(x)$,
respectively, returned by the algorithms (in view of Lemma
\ref{BFS-dist-to-gated}). Then we can set $u^+_1$ to be the vertex
$\neighbor_{w_1,x}(\gate_1(u))$ (which is  the neighbor of $\gate_1(u)$ in
$\partial_x F(w_1)$) and
$u^+_2$ to be the vertex $\neighbor_{w_2,x}(\gate_2(u))$ (which is the neighbor
of $\gate_2(u)$ in $\partial_x F(w_2)$). Since $\gate_1(u)$ and $\gate_2(u)$ are
the gates of $u$ in $\partial_{w_1} F(x)$  and in  $\partial_{w_2} F(x)$,
$u^+_1$ and $u^+_2$ are the gates of $u$ in $\partial_x F(w_1)$ and
$\partial_x F(w_2)$, respectively. If the cone $F(x)$ has $n_i$ vertices, then
the computation of the gates $u^+_1$ and $u^+_2$ of all the vertices $u\in 
F(x)$ will take $O(n_i)$ time.
Since the cones of $G$ are pairwise disjoint, this computation takes total $O(n)$ time. Consequently, we obtain:

\begin{lemma} 
    \label{comput-gates} 
    Given a cube-free median graph $G$ with $n$ vertices, the following 
    entities can be computed in total linear time $O(n)$:
    \begin{itemize}
        \item the gates of all vertices $v$ of $G$ in the star $\St(c)$;
        \item the gates of all vertices $v$ belonging to the cones of $G$ in 
        the two neighboring panels.
    \end{itemize}
\end{lemma}

\begin{remark} 
    For a median graph $G$ with $m$ edges and a gated subgraph $H$ of $G$, the 
    same algorithm computes in  $O(m)$ time the gates of the vertices of $G$ in 
    all the boundaries of their fibers.
\end{remark}

\subsubsection{Computation of imprints}

In line 10 of  \distenc{}, for each vertex $u$ of a panel $F(x)$ we have to compute the two imprints $u_1$ and $u_2$ of $u$ on the total boundary $\partial^* F(x)$. This computation is
based on the following properties of imprints.

\begin{lemma} 
    \label{imprint1} 
    The imprints $u_1,u_2$ of $u\in F(x)$ on $\partial^* F(x)$ satisfy  the 
    following properties:
    \begin{enumerate}
    \item[(a)] there exist $y_1,y_2\in \St(c), y_1\ne y_2$,  and $x\sim 
    y_1,y_2$ such that $u_1$ is a gate of $u$ in $\partial_{y_1} F(x)$ and 
    $u_2$ is the gate of $u$ in  $\partial_{y_2} F(x)$;
    \item[(b)] $u_1,u_2\in I(u,x)$;
    \item[(c)] if $w$ is a neighbor of $u$ in $I(u,u_i)$, $i\in \{ 1,2\}$, then 
    $u_i$ is an imprint of $w$.
    \end{enumerate}
\end{lemma}

\begin{proof} 
    By Lemma \ref{lem_Uborders_cube-free}, $\partial^*F(x)$ is an isometric 
    tree with gated branches. Recall also that $\partial^* F(x)$ is the union 
    of all relative boundaries $\partial_y F(x), y\sim x$, which are all gated 
    trees.
    Therefore, by the definition of imprints, $u_1$ and $u_2$ are the gates of 
    $u$ in all the relative boundaries to which they belong.  This implies that 
    $u_1$ and $u_2$ cannot belong to a common relative boundary, establishing 
    (a).
    Since $x$ belongs to every relative boundary, this also implies that  
    $u_1,u_2\in I(u,x)$, establishing (b). 
    It remains to show property (c). If we suppose that $u_1$ is not an imprint 
    of $w$, then from imprint's definition there exists $z\in \partial^* F(x), 
    z\ne u_1$ such that $z\in I(w,u_1)$. Since $w\in I(u,u_1)$, we deduce that 
    $z\in I(w,u_1)\subset I(u,u_1)$ contrary to the assumption that $u_1$ is an 
    imprint of $u$.
\end{proof}

\begin{lemma} 
    \label{BFSimprint2} 
    If $uv$ is an edge of $F(x)$, then either their imprints coincide, i.e., 
    $\{ u_1,u_2\}=\{v_1,v_2\}$, or one of the vertices of the pair $\{ 
    u_1,u_2\}$ coincides with one of the vertices of the pair $\{ v_1,v_2\}$ 
    and the two other vertices from each pair are adjacent. 
\end{lemma}

\begin{proof} 
    We will use the following general assertion:

    \begin{claim} 
        \label{quadruplet} 
        If $a,b,a',b'$ are vertices of a bipartite graph $G$ such that $a\sim 
        b$, $\dist_G(a,a')\le \dist_G(b,b')$, and $b'\in I(b,a')$, then either 
        $a'=b'$ and $\dist_G(b,b')=\dist_G(a,a')+1$ or $a'\sim b'$ and 
        $\dist_G(b,b')=\dist_G(a,a')$.
    \end{claim}
    
    \begin{proof} 
        From the choice of the quadruplet $a,b,a',b'$ we deduce that 
        $\dist_G(b,b')+\dist_G(b',a')=\dist_G(b,a')\le \dist_G(a,a')+1\le 
        \dist_G(b,b')+1$. This implies that $\dist_G(b',a')\le 1$. If $a'\sim 
        b'$, from previous inequalities we deduce that 
        $\dist_G(a,a')=\dist_G(b,b')$. If $a'=b'$, since $G$ is bipartite and  
        $\dist_G(a,a')\le \dist_G(b,b')$, we must have 
        $\dist_G(b,b')=\dist_G(a,a')+1$.
    \end{proof}

    To prove the lemma, suppose that $\dist_G(v,v_1)$ is the smallest distance 
    among $\{ \dist_G(u,u_1),\dist_G(u,u_2),\dist_G(v,v_1),\dist_G(v,v_2)\}$. 
    From imprint's definition it follows that one of the vertices $u_1,u_2$,
    say $u_1$, belongs to $I(u,v_1)$. By Claim \ref{quadruplet} we conclude 
    that either $u_1\sim v_1$ and $\dist_G(u,u_1)=\dist_G(v,v_1)$ or  $u_1=v_1$ 
    and $\dist_G(u,u_1)=\dist_G(v,v_1)+1$.
    
    \smallskip\noindent
    {\bf Case 1.} $u_1=v_1$ and $\dist_G(u,u_1)=\dist_G(v,v_1)+1$.
    
    \smallskip\noindent
    Suppose without loss of generality that $\dist_G(v,v_2)\le \dist_G(u,u_2)$ 
    (the other case is similar). From imprint's definition it follows that one 
    of the vertices $u_1,u_2$
    belongs to $I(u,v_2)$. Since $v_1=u_1$ is an imprint of $v$, this cannot be 
    $u_1$. Thus $u_2\in I(u,v_2)$. By Claim \ref{quadruplet}, either $u_2=v_2$ 
    or $u_2\sim v_2$ and $\dist_G(u,u_2)=\dist_G(v,v_2)$, and we are done.
    
    \smallskip\noindent
    {\bf Case 2.} $u_1\sim v_1$ and  $\dist_G(u,u_1)=\dist_G(v,v_1)$.
    
    \smallskip\noindent
    Suppose without loss of generality that $\dist_G(v,v_2)\le 
    \dist_G(u,u_2)$.  From imprint's definition it follows that one of
    the vertices $u_1,u_2$, belongs to $I(u,v_2)$. First suppose that $u_1\in 
    I(u,v_2)$. We assert that in this case $v_1$ belongs to $I(v,v_2)$ contrary
    to the assumption that $v_2$ is an imprint of $v$. Indeed, since $G$ is 
    bipartite and $u\sim v$, either $u\in I(v,v_2)$ or $v\in I(u,v_2)$.
    If $u\in I(v,v_2)$, then $u_1\in I(u,v_2)$ and $v_1\in I(v,u_1)$ imply that 
    $v_1\in I(v,v_2)$. Otherwise, if $v\in I(u,v_2)$,
    since $v,u_1\in I(u,v_2)$ and $v_1\in I(v,u_1)$, the convexity of intervals 
    (Lemma \ref{convex-interval}) implies that $v_1\in I(u,v_2)$. Since $v\in 
    I(u,v_1)$,
    we conclude that $v_1\in I(v,v_2)$. This shows that  $u_1\in I(u,v_2)$ is 
    impossible. Suppose now that $u_2\in I(u,v_2)$. By Claim \ref{quadruplet},
    either $u_2=v_2$ or  $u_2\sim v_2$ and $\dist_G(u,u_2)=\dist_G(v,v_2)$.
    It remains to show that the second possibility is impossible.
    
    Suppose by way of contradiction that $u_1\sim v_1$ and $u_2\sim v_2$ hold, 
    and 
    among all such edges of $F(x)$ suppose that the edge $uv$ is closest to 
    $\partial^* F(x)$. This implies that, 
    in the tree $\partial^* F(x)$, one of the vertices $u_i,v_i$ ($i = 1, 2$) 
    is 
    the parent of another one. Suppose without loss of generality
    that $\dist_G(u,x)<\dist_G(v,x)$. Since $v_i\in I(u_i,v)$ and, by Lemma 
    \ref{imprint1}(b), $v_1,v_2\in I(v,x)$ and $u_1,u_2\in I(u,x)$, we conclude 
    that, for each $i=1,2$,  $u_i$ is the parent of $v_i$.
    Let $z$ be a neighbor of $u$ in $I(u,u_1)$. Since $v,z\in I(u,v_1)$, by 
    quadrangle
    condition there exists a common neighbor $w$ of $v$ and $z$ one step closer 
    to $v_1$. Analogously, let $z'$ be a neighbor of $u$ in $I(u,u_2)$ and $w'$ 
    be a common neighbor
    of $z'$ and $v$ in $I(v,v_2)$. Since $u,w,w'$ are neighbors of $v$ in 
    $I(v,x)$ and the median graph $G$ is cube-free, Corollary 
    \ref{neighbors_interval} implies that two of the vertices $u,w,w'$ must 
    coincide.
    By definition  of $w,w'$ we have  $w\ne u$ and $w'\ne u$, thus necessarily 
    $w=w'$. This yields $z=z'$. Consequently, $w=w'\in I(v,v_1)\cap I(v,v_2)$ 
    and $z=z'\in I(u,u_1)\cap I(u,u_2)$. By Lemma \ref{imprint1}(c),
    the vertices $v_1,v_2$ are the imprints of $w=w'$ and the vertices 
    $u_1,u_2$ are the imprints of $z=z'$ and we obtain a counterexample (the 
    edge $wz$) closer to $\partial^* F(x)$ than $uv$, contrary to the choice of 
    the edge 
    $uv$.
\end{proof}

From Lemma \ref{BFSimprint2} and its proof we obtain the following corollary:

\begin{corollary} 
    \label{BFSimprint3} 
    If $uv$ is an edge of $F(x)$ with $\dist_G(x,v)<\dist_G(x,u)$ and $u_1=v_1$ 
    and $u_2\ne v_2$, then $v_2$ is the parent of $u_2$ in $\partial^* F(x)$.
\end{corollary}

\begin{proof} 
    Since $u_2$ and $v_2$ are adjacent in the tree $\partial^* F(x)$, one must 
    be the parent of other. Suppose by way of contradiction that $u_2$ is the 
    parent of $v_2$. Since $v_2\in I(v,x), u_2\in I(v_2,v),$ and $u\in 
    I(u_2,v)$, we conclude that $u\in I(v,x)$ contrary to the assumption that  
    $\dist_G(x,v)<\dist_G(x,u)$.
\end{proof}

\begin{lemma} 
    \label{BFSimprint4} 
    If $v,w$ are the neighbors of $u$ in $I(x,u)$, then the imprints $u_1,u_2$ 
    of $u$ belong to the set $\{ v_1,v_2,w_1,w_2\}$ of imprints of $v$ and $w$.
\end{lemma}

\begin{proof} 
    By Lemma \ref{imprint1}(b), $u_1,u_2\in I(u,x)$. Let $z'$ be a neighbor of 
    $u$ in $I(u,u_1)$ and $z''$ be a neighbor of $u$ in $I(u,u_2)$. By Lemma 
    \ref{imprint1}(c), $u_1$ is an imprint of $z'$ and $u_2$ is an imprint of 
    $z''$. Since $z'\in I(u,u_1)\subseteq I(u,x)$ and $z''\in I(u,u_2)\subseteq 
    I(u,x)$,  $z',z''$ must be neighbors of $u$ in $I(u,x)$, i.e., $z',z''\in 
    \{ v,w\}$. If $z'=z''=v$, then $u_1$ and $u_2$ are imprints of $v$ and we 
    are done.
    Otherwise, if $z'=v$ and $z''=w$, then $u_1$ is an imprint of $v$ and $u_2$ 
    is an imprint of $w$, and we are done again.
\end{proof}

We continue with an algorithm for computing the imprints of vertices $u$ of
$F(x)$. It consists in running three BFS traversals of $F(x)$. The first BFS
(with respect to $x$) computes the distances $d(u)=\dist_G(u,x)$ from each
vertex $u\in F(x)$ to $x$ and the (at most two) neighbors of $u$ in the interval
$I(u,x)$. The second BFS (with respect to the total boundary $\partial^* F(x)$)
computes the first imprint $u_1$ of each $u\in F(x)$ as a closest to $u$ vertex
of $\partial^* F(x)$. It also computes the distance $d_1(u)$ from $u$ to this
imprint $u_1$ and the parent $f(u)$ of $u$ belonging to the interval
$I(u,u_1)$. By Lemma \ref{imprint1}, $u_1$ is also an imprint of $f(u)$ (this
explain why in line 7, $v=f(u)$ implies $v_1=u_1$). Finally,
the vertices of $F(x)$ are traversed for the third time according to the order computed by the first BFS traversal. For each vertex $u\in F(x)$, in the assumption that the two imprints of its predecessors $v,w\in I(u,x)$ have been already computed,
the algorithm returns one of these four vertices as the second imprint $\imp_2(u)$ of $u$. This choice is justified by Lemma \ref{BFSimprint2}. In the pseudocode of  Algorithm \ref{alg_comput_imprints} and in the proof of Lemma \ref{BFSimprint5} we use the  convention that whenever $i$ denotes an element of the pair $\{1,2\}$, $j$ is the remaining element of $\{1,2\}$, i.e., an element such that $\{ i,j\}=\{ 1,2\}$. Clearly,  Algorithm \ref{alg_comput_imprints} is linear in the number of vertices of $F(x)$.
Its correctness follows from the following lemma:

\begin{lemma} 
    \label{BFSimprint5} 
    For each vertex $u$ of $F(x)$, $u_2:=\imp_2(u)$ is an imprint of $u$.
\end{lemma}

\begin{proof} 
    Since we assume that $v$ is the parent $f(u)$ of $u$ in the BFS with 
    respect to $\partial^* F(x)$, we have $v_1=u_1=:z$. In the proof we will 
    use Lemma \ref{BFSimprint2} and the fact  that if the two imprints of a 
    vertex of $F(x)$ are different, then they cannot be adjacent (this easily 
    follows from imprint's definition).

    \smallskip\noindent
    {\bf Case 1.} There exists $i\in \{ 1,2\}$ such that $w_i=z$.
    
    \smallskip\noindent
    By Lemma \ref{BFSimprint2},
    the second imprint $u_2$ of $u$ coincides with one of the imprints $v_2,w_j$
    and coincides or is adjacent with the second imprint. This implies that if 
    one of the vertices $v_2$ or $w_j$ coincides with $z$,
    then the second vertex also coincides with $z$. Indeed, suppose by way of 
    contradiction that $v_2=z$ but $w_j\ne z$ (the case  $w_j=z$ and $v_2\ne z$ 
    is similar). If $u_2=z$, then $w_j$ must be adjacent to $u_2=z=w_i$, which 
    is impossible.
    Similarly, if $u_2=w_j$ then by Lemma \ref{BFSimprint2}, $u_2$ must 
    coincide or be adjacent to $v_2=z=w_i$, a contradiction. This concludes the 
    case when one of the vertices $v_2$ or $w_j$ coincides with $z$. Now 
    suppose that
    both $v_2$ or $w_j$  are different from $z$. In this case they are both not
    adjacent to $z=v_1=w_i$. If $v_2=w_j$, then clearly $u_2$ coincides with 
    this vertex. Otherwise, since $v,w\in I(u,x)$ and $u_2$ is one of 
    $v_2,w_j$, by Lemma \ref{BFSimprint2}
    and Corollary \ref{BFSimprint3} we conclude that another vertex from this 
    pair must be the parent of $u_2$ in $\partial^* F(x)$. Since in line 10, 
    the algorithm selects as $\imp_2(v)$ the vertex of the pair $v_2,w_j$
    furthest from $x$ (in $\partial^* F(x)$ or in $G$), this shows that 
    $\imp_2(u)$ is indeed the second imprint $u_2$ of $u$.
    
    \smallskip\noindent
    {\bf Case 2.} Both vertices $w_1,w_2$ are distinct from $z$.
    
    \smallskip\noindent
    Since $z=u_1$ and $w\in I(u,x)$, by Corollary \ref{BFSimprint3} one of the 
    imprints $w_1,w_2$, say $w_i$ must be adjacent to $z$ and 
    $\dist_G(u,z)=\dist_G(w,w_i)$. Since $u_1=z$, $u_2$ is different from $w_i$.
    Note also that since $w_j$ is not adjacent to $w_i$, we have $w_j\ne z$. By 
    Lemma \ref{BFSimprint2}, $u_2$ is adjacent or coincides with $w_j$. Since 
    this is impossible if $u_2=z$, we conclude that $u_2\ne z$. Consequently,
    $u_2$ is one of the vertices $v_2$ or $w_j$.  If $v_2=w_j$, then clearly, 
    $u_2$ is that vertex. Otherwise, since $v,w\in I(u,v)$, by Corollary 
    \ref{BFSimprint3}, the  vertex of the pair $v_2,w_j$ which is different 
    from $u_2$ is the parent of $u_2$ in $\partial^* F(x)$. Since in line 10, 
    the algorithm selects as $\imp_2(v)$ the vertex of the pair $v_2,w_j$
    furthest from $x$ (in $\partial^* F(x)$ or in $G$), this shows that 
    $\imp_2(u)$ is indeed the second imprint $u_2$ of $u$. Consequently, in all 
    cases we have $u_2=\imp_2(u)$, concluding the proof.
\end{proof}

\begin{algorithm}[H]
    \caption{\label{alg_comput_imprints}ComputeImprints($F(x)$, $\partial^*
    F(x)$)}
    \Input{A panel $F(x)$ and its total boundary $\partial^* F(x)$}
    \Output{For each vertex $u\in F(x)$, its two imprints $\imp_1(u)$ and
    $\imp_2(u)$}

    \BlankLine

    Run a first BFS on $F(x)$ with respect to $x$ in order to compute, for each 
    vertex $u$, its distance $d(u)$ to $x$ and its predecessors $\pred_1(u)$ and
    $\pred_2(u)$ in $I(u,x)$ \;

    \BlankLine

    Run a second BFS on $F(x)$ with respect to $\partial^* F(x)$ in order to 
    compute, for each vertex $u$, a closest to $u$ vertex $\imp_1(u)$ in 
    $\partial^* F(x)$, and the parent $f(u)$ of $u$ \;

    \BlankLine

    Assume the vertices of $F(x)$ ordered in the BFS-order computed by the
    first BFS traversal \;
    \ForEach{$u \in F(x)$}{
        $u_1 \leftarrow \imp_1(u)$ \;

        $(v, w) \leftarrow (\pred_1(u), \pred_2(u))$
        \tcp*{$\pred_1(u) = f(u)$}

        $(v_1, v_2) \leftarrow (\imp_1(v), \imp_2(v))$
        \tcp*{$v = f(u)$ implies that $u_1 = v_1$}

        $(w_1, w_2) \leftarrow (\imp_1(w), \imp_2(w))$ \;

        \If{$\exists i \in \{1,2\}$ s.t. $w_i = u_1 = v_1$}{
            Set as $\imp_2(u)$  the vertex of $\{v_2, w_j\}$ furthest from
            $x$ \;
        }
        \Else{
            Let $i \in \{1,2\}$ be such that $w_i$ is the parent of $u_1 =
            v_1$ in $\partial^* F(x)$ \;
            \lIf{$w_j = v_2$}{
                $\imp_2(u) \leftarrow v_2$
            }
            \hspace{4.9em}%
            \lElse{
                \hspace{1ex}%
                Set  as $\imp_2(u)$ the vertex of $\{v_2, w_j\}$ furthest
                from $x$
            }
        }
    }
\end{algorithm}

\subsubsection{The size of labels} 

Consider now the length of labels given by \distenc{}.
To analyze the distance decoder, we consider a \emph{RAM model} in which 
standard arithmetical operations on words of size $O(\log n)$ (additions, 
comparisons, etc.) are supposed to take constant time.

\begin{lemma}
    \label{lem_distenc_length}
    \distenc{} gives to every vertex of an $n$-vertex cube-free median graph $G = (V,E)$ a label of length $O(\log^3n)$ bits.
\end{lemma}
\begin{proof}
    Since at each division step we select a centroid, by Lemma
    \ref{lem_small_halfspaces}  every vertex $v \in V$ will appear in at most
    $\ceil{\log_2 |V|}$ different
    \fibers. For each of these \fibers, $\LD(v)$ will receive $O(\log^2n)$ new bits.
    Indeed, the information stored correspond to Lines
    \ref{alg-line_dist_enc_median1}, \ref{alg-line_dist_enc_left1} and
    \ref{alg-line_dist_enc_right1} (or \ref{alg-line_dist_enc_median2},
    \ref{alg-line_dist_enc_left2} and \ref{alg-line_dist_enc_right2}) of
    Algorithm \ref{alg_distenc}.
    $L_\st$ clearly has size $O(\log n)$ because so does
    $\Lb{v^\downarrow}{\St(c)}$ as seen in Section \ref{sect_labeling_trees_stars} for stars, and $L_\lleft$ and $L_\rright$ both have size
    $O(\log^2n)$ because the tree labeling they contain has size $O(\log^2n)$
    as seen in Section \ref{sect_labeling_trees_stars} for trees.
\end{proof}

\subsubsection{Complexity and correctness}

Suppose that the partition of $G$ into fibers contains $k$ fibers 
$F_1,\ldots,F_k$ with $n_1,\ldots,n_k$ vertices, respectively, where 
$\sum_{i=1}^k n_i=n$. Each of the fibers $F_i$ is a gated subgraph of $G$ and 
thus is a cube-free median graph. Therefore, in $F_i$ we can compute a centroid 
$c_i$  by the algorithm of \cite{BeChChVa} and partition $F_i$ into fibers with 
respect to $\St(c_i)$,  and compute their boundaries and the imprints. All this 
can be done in $O(n_i)$ time, leading to a total of $\sum_{i=1}^k O(n_i) = 
O(n)$. Continuing this way, we conclude that at the partition iteration $j$, 
the total time to compute the centroid vertices, to partition the current 
fibers into smaller fibers, to compute their total boundaries and the imprints 
on them will take $O(n)$ time. Since we have $\log_2n$ partition steps, we 
conclude that the total complexity of the partition algorithm is $O(n\log n)$.

The correctness of the algorithm \distenc{} results from Lemmas \ref{compute-fibers},\ref{comput-gates},\ref{BFSimprint5} and the following
properties of cube-free median graphs: stars and fibers are gated (Lemmas
\ref{prop_St(m)_convex} and \ref{fiber-gated1}); total boundaries of fibers are
quasigated (Corollary \ref{proj-total-boundary}) isometric trees with gated
branches (Lemma \ref{lem_Uborders_cube-free}); and from the formulae for
computing the distance between \separated, \neighboring, and \aNeighboring
vertices (Lemmas \ref{lem_shortest_path_separated},
\ref{lem_shortest_path_neighboring}, and \ref{lem_shortest_path_aNeighboring}).

Given two labels $\LD(u)$ and $\LD(v)$, \distdec{} can find the last common
centroid of $u$ and $v$ by reading their label once. This can be done in time 
$O(\log^2 n)$ assuming the \emph{word-RAM model}. This complexity can be 
improved to constant time by adding the following appropriate $O(\log^2 n)$ 
bits information concatenated to $\LD(u)$ and $\LD(v)$. For that,
consider the tree $T$ (of recursive calls) in which vertices at depth $i$ are the centroids chosen at step $i$ and in which the children
of a vertex $x$ are the centroids  chosen at step $i + 1$ in the fibers 
generated by $x$ at step $i$. We can observe that every vertex of $G$ appears 
in this tree, that the last common centroid $c$ of any two vertices $u$ and $v$ 
of $G$ is their nearest common ancestor in the tree $T$, and that its depth $j$ 
in this tree corresponds to its position in $\LD(u)$ and $\LD(v)$, i.e.,
$\LD_{j}^{\St[\med]}(u) = \LD_{j}^{\St[\med]}(v) = \id(c)$.
As noticed in \cite{Peleg05}, any distance labeling for trees $T$ can be
modified to support \emph{nearest common ancestor's depth (NCAD)} queries by
adding the depth $\depth(u)$ of $u$ in $T$ to the label $L(u)$ given to
each vertex $u \in V(T)$ by the distance labeling. Given two vertices $u$
and $v$ of $T$, the NCAD decoder then returns $\frac{1}{2}(\depth(u) +
\depth(v) - \dist_T(u,v))$.
So, during the execution of \distenc, we can also construct the tree $T$ of
recursive calls and then give an NCAD label $L'(u)$ in $T$ to every vertex of
$G$. Now, the first step of \distdec{} will consist in decoding $L'(u)$ and
$L'(v)$ in order to find the last common median of $u$ and $v$.
Once this step is done, \distdec{} has to call $\distDecStar$
on labels of size $O(\log n)$ which requires a constant number of steps.
After that, either the information necessary to compute $\dist_G(u,v)$ is
directly encoded in $\LD(u)$ or $\LD(v)$, or \distdec{} needs to
to decode distance labels for trees, which can be done in constant time \cite{FrGaNiWe_trees}.
Consequently \distdec{} has a constant time complexity.
The fact that \distdec($\LD(u)$, $\LD(v)$) returns $\dist_G(u,v)$ follows from
Lemmas \ref{lem_shortest_path_separated}, \ref{lem_shortest_path_neighboring}
and \ref{lem_shortest_path_aNeighboring}. This completes the proof of Theorem 
\ref{thm_dist_labeling}.

\section{Routing labeling schemes for cube-free median graphs}
\label{sect_rout_labeling}

In this section we briefly describe the routing scheme; since there exists an
important resemblance with the distance labeling scheme, a formal
description of the routing scheme is given in the appendix.
The idea of encoding is the same as the one for the distance labeling schemes:
the graph $G$ is partitioned recursively into
fibers with respect to centroids. At every step, the labels of the
vertices are given a vector of three parts, named ``\st'', ``\lleft'', and
``\rright'' as before. However, the information stored in these parts is not
completely the same as for distances. This is due to the fact that we
need to keep the information specific for routing and also because, at the
difference of distance queries, the routing queries are not commutative.
For instance, for computing the distance between \neighboring vertices $u$ and
$v$, we assumed that $u$ belongs to a \panel and $v$ to a \cone. The case when
$u$ belongs to a \cone and $v$ belongs to a \panel is reduced to the first case
by calling the same corresponding function but commuting the arguments.
This is no longer possible in the routing queries: routing from a \panel to a
\cone is different from the routing from a \cone to a \panel.

As for distances, the routing decision is taken the first time
$u$ and $v$ belong to different fibers of the current partition. Let $c$ be a
centroid vertex of the current graph under partition and let $F(x)$ and $F(y)$ be
the two fibers containing $u$ and $v$, respectively. If $u$ and $v$ are
\separated, then $d_G(u,v)=d_G(u,c)+d_G(c,v)$, thus routing from $u$ to $v$ can
be done by routing from $u$ to $c$ (unless $u=c$). Therefore, the encoding
scheme must keep in the label of $u$ the identifier of some neighbor of $u$ in
$I(u,c)$. If $u=c$, then it suffices to route from $u=c$ to the gate $y$ of $v$
in $\St(c)$. This is done by using the routing scheme for stars.

If $u$ and $v$ are \aNeighboring, then $F(x)$ and $F(y)$ are \cones with a
common neighboring \panel $F(w)$. Similarly to distance scheme, the algorithm
finds $F(w)$. Since the gates $u^+$ of $u$ and $v^+$ of $v$ in $F(w)$
belong to a common shortest $(u,v)$-path, it suffices to route the message from
$u$ to $u^+$. Therefore the encoding must keep in the label of $u$ the
identifier of a neighbor of $u$ in $I(u,u^+)$.
The same information is required when $u$ and $v$ are \neighboring and $F(x)$
is a \cone and $F(y)$ is a \panel. Indeed, in this case there is a shortest
$(u,v)$-path passing via the gate $u^+$ of $u$ in $F(y)$ and one of the
imprints of $v$ in $\partial^* F(y)$.
Therefore, to route from $u$ to $v$ it suffices to route from $u$ to $u^+$.

Finally, let $u$ and $v$ be \neighboring, however now $F(x)$ is a
\panel and $F(y)$ is a \cone. Recall that in this case there exists a shortest
$(u,v)$-path passing via one of the imprints $u_1$ or $u_2$ of $u$ on
$\partial^* F(x)$ and the gate $v^+$ of $v$ in $F(x)$. Therefore, if $u$ is
different from $v^+$ then it suffices to route the message from $u$ to a
neighbor of $u$ in $I(u,u_1)$ or $I(u,u_2)$ (depending of the position of $v$).
Therefore, in the label of $u$ we have to keep the identifiers of those two
neighbors of $u$. To decide to which of them we have to route the message from
$u$, we need to compare $d_G(u,u_1)+d_G(u_1,v^+)$ and $d_G(u,u_1)+d_G(u_1,v^+)$.
Therefore, at the difference of the routing scheme in trees, our routing scheme
for cube-free median graphs must incorporate the distance scheme. On the other
hand, if $u$ coincides with $v^+$, then necessarily we have to route the
message to a neighbor of $u$ in $I(u,v)$, which necessarily belong to the \cone
$F(y)$ and not to $F(x)$ (because $v^+$ is the gate of $v$ in $F(x)$). There
exists a unique vertex $\twin(v^+)$ of $F(y)$ adjacent to $v^+$. We cannot keep
the identifier of $\twin(v^+)$ in the label of $u=v^+$ because a vertex in a
\panel may have arbitrarily many neighbors in the neighboring \cones. Instead,
we can keep the identifier  of $\twin(v^+)$ in the label of $v$ (recall that a
\cone has only two neighboring \panels). Consequently, we obtain a routing
scheme for cube-free median graphs with labels of vertices of size $O(\log^3
n)$.

\section{Conclusion}

In this paper we presented distance and routing labeling schemes for cube-free
median graphs $G$ with labels of size $O(\log^3 n)$.
For that, we partitioned  $G$ into fibers (of size $\leq
n/2$) of the star $\St(c)$ of a centroid $c$ of $G$.
Each fiber is further recursively partitioned using the same algorithm.
We classified the fibers into panels and cones and the pairs of vertices $u, v$
of $G$ into \roommates,  \separated, \neighboring, and \aNeighboring pairs. If
$u$ and $v$ are \roommates, then $d_G(u,v)$ is computed  at a later step of the
recursion. Otherwise, we showed how to retrieve $\dist_G(u,v)$ by keeping in the labels of
$u$  and $v$ some distances from those two vertices to their gates/imprints in a constant number of fibers.
Our main technical ingredient is the fact that the total boundaries of fibers of
cube-free median graphs are isometric quasigated trees.

This last property of total boundaries  is a major obstacle in generalizing our approach to all
median graphs, or even to median graphs of dimension $3$. The main problem is
that in this case the total boundaries are no longer median graphs.

\begin{example} 
    In Fig. \ref{non-median_boundary} we present a median graph $G$ of 
    dimension 3, in which the total boundary of a fiber of a vertex $c$  is not 
    a median graph. The graph $G$ is just the cubic grid $3\times 3\times 3$ 
    and the vertex $c$ is one of the corners  of this grid. The star $\St(c)$ 
    of $c$ is a single 3-cube $Q$. Let $x$ be the vertex of $Q$ opposite to 
    $c$, and $y,z,w$ are the three vertices of $Q$ at distance 2 from $c$. Then 
    the fiber $F(x)$ is neighboring to $F(y), F(z),$ and $F(w)$. The total 
    boundary $\partial^* F(x)$ of $F(x)$ consists of three squares $\partial_y 
    F(x), \partial_z F(x),$ and $\partial_w F(x)$, pairwise intersecting in 
    edges incident to $x$ and all three intersecting in $x$.
    Consequently, $\partial^* F(x)$ is not a median graph.
    
    In this example, $c$ is not a centroid of $G$. To repair this, we can 
    consider the $5\times 5\times 5$ cubic grid $G'$ in which $G$ is embedded 
    in such a way that $c$ becomes the unique centroid of $G'$. The star of $c$ 
    in $G'$ is a $3\times 3\times 3$ grid. The vertex $x$ belongs to this 
    extended star, but the boundary of $F(x)$ will still consists of the same 
    three squares, thus $\partial^* F(x)$ is not median.
\end{example}

\begin{figure}[ht]
    \centering
    \includegraphics[width=0.4\textwidth]{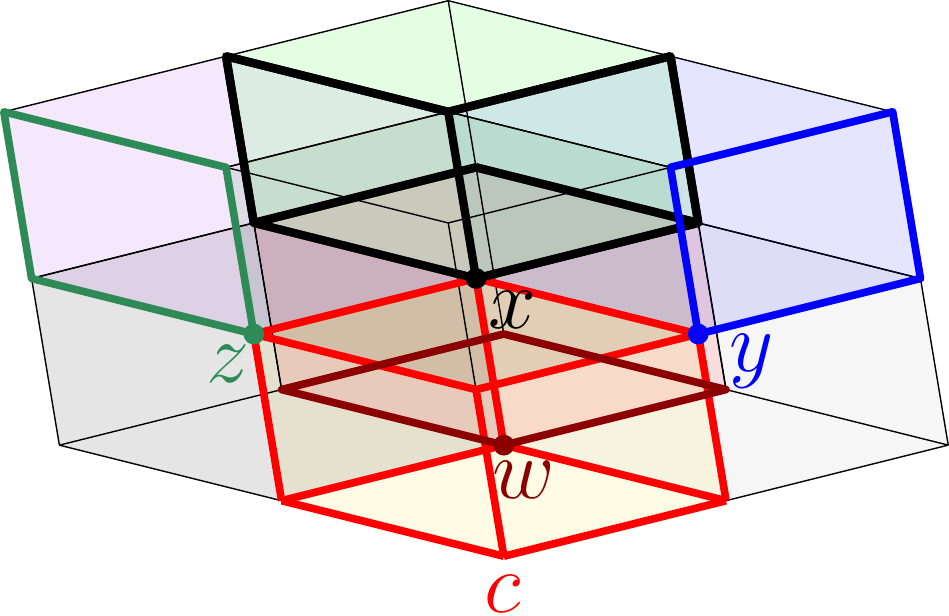}
    \caption{\label{non-median_boundary}
       A median graph of dimension 3 in which the total boundary of a fiber is not median.
    }
\end{figure}

Therefore, we can no longer recursively  apply to total boundaries the
distance and routing labeling
schemes for median graphs of smaller dimension (as we did in case of cube-free median graphs by applying such schemes for trees).
Nevertheless, a more brute-force approach works for arbitrary median graphs $G$
of constant maximum degree $\Delta$. In this case, all cubes of $G$ have
constant size. Thus, the star $\St(c)$ cannot have more than $O(2^\Delta)$
vertices, i.e., $\St(c)$ has a constant number of fibers.  Since every fiber is
gated, at every step of the encoding algorithm, every vertex $v$ can store in
its label the distance from $v$ to its gates in all fibers  of $\St(c)$.
Consequently,  this leads to simple distance and routing labeling  schemes with 
labels of (polylogarithmic) length for all median graph with maximum degree at 
most $\Delta$:

\begin{proposition} 
    Any median graph $G$ with maximum degree $\Delta$ admits distance and 
    routing labeling  schemes  with labels of length $O(2^{\Delta}\log^3 n)$ 
    bits.
\end{proposition}

Nevertheless, the maximum degree seems to be not the right complexity parameter of general median graphs. Similarly to high-dimensional computational geometry (where the dimension of the space
is often constant), the largest dimension $d$ of a cube of a median graph $G$ might be considered as such a parameter. Consequently, we would like to formulate the following open question
(which seems interesting and nontrivial already in dimension 3 and if we allow a constant stretch factor):

\begin{question}
    \label{question1}
    Does there exist a polylogarithmic distance labeling scheme (exact or approximate) for general
    median graphs or for median graphs of constant dimension?
\end{question}

In Section \ref{sect_median_graphs} we presented several problems from 
different research areas, which have a negative answer for all median graphs 
(viewed as CAT(0) cube complexes or domains of event structures)
but can be positively solved for cube-free median graphs (alias 2-dimensional CAT(0) cube complexes). Our Question \ref{question1} can be cast in this type of problems (even if it is formulated for finite
median graphs and is an algorithmic problem). On the other hand, the paper \cite{CC-JCSS} establishes that a conjecture from concurrency theory is  false already for cube-free median graphs but (using a
breakthrough result by Agol \cite{Agol}) is true for  hyperbolic median graphs. Gromov \emph{hyperbolicity} is a parameter of a median graph, stronger than  (cube-)dimension and in a sense is similar to
the treewidth: the hyperbolicity of a median graph $G$ is the size $h$ of the largest $h\times h$ square grid isometrically embedded in $G$. Constant hyperbolicity implies constant dimension but not the vice-versa
(already in case of cube--free median graphs).
Gavoille and Ly \cite{GaLy_hyperbolic} established that general graphs of
bounded hyperbolicity do not admit poly-logarithmic distance labeling schemes
unless we allow a multiplicative error of order $\Omega(\log \log n)$, at least.
We would like to finish this paper with a seemingly easier (but  still open for us)
version of Question \ref{question1}:

\begin{question}
    \label{question2}
    Does there exist a polylogarithmic distance labeling scheme (exact or 
    approximate) for median graphs of constant hyperbolicity?
\end{question}

\section{Glossary}
\label{sect_glossary}

{
    \renewcommand{\arraystretch}{1.2}
    \begin{tabular}{p{0.33\linewidth}p{0.67\linewidth}}
        \hline
        \textbf{Notions and notations}
            & \textbf{Definitions}
            \\
            \hline
        Boundary $\partial_{y} F(x)$
            & $\{ u \in F(x) : \exists v \in F(y), uv \in E(G) \}$.
            \\
        Centroid $c$
            & Vertex minimizing $u \mapsto \sum_{v \in V(G)} \dist_G(u,v)$.
            \\
        Cone $F(x)$ w.r.t. to $\St(z)$
            & Fiber $F(x)$ w.r.t. $\St(z)$ with $\dist_G(x,z) = 2$.
            \\
        Convex subgraph $H \subseteq G$
            & $\forall u, v \in V(H)$, $I(u,v) \subseteq V(H)$.
            \\
        Distance $\dist_G(u,v)$
            & Number of edges on a shortest $(u,v)$-path of $G$.
            \\
        Fiber $F(x)$ w.r.t. $H \subseteq G$
            &$\{u \in V(G) : x \text{ is the gate of $u$ in } H \}$.
            \\
        Gate of $u$ in $H \subseteq G$
            & $u' \in V(H)$ s.t. $\forall v \in V(H)$, $\dist_G(u,v) =
            \dist_G(u,u') + \dist_G(u',v)$.
            \\
        Gated subgraph $H \subseteq G$
            & Every vertex $u \in V(G)$ admits a gate in $H$.
            \\
        Halfspace $W(u,v)$
            &$\{w \in V(G) : \dist_G(u,w) < \dist_G(v,w) \}$.
            \\
        Imprints set $\eproj{u}{H}$
            & $\{ w \in V(H) : I(u,w) \cap V(H) = \{w\} \}$.
            \\
        Interval $I(u,v)$
            & $\{ w \in V(G) : \dist_G(u,v) = \dist_G(u,w) + \dist_G(w,v) \}$.
            \\
        Isometric subgraph $H \subseteq G$
            & $\forall u, v \in V(H), \dist_H(u,v) = \dist_G(u,v)$.
            \\
        Locally convex subgraph $H \subseteq G$
            & $\forall u, v \in V(H)$, with $\dist_G(u,v) \le 2$, $I(u,v)
            \subseteq V(H)$.
            \\
        Panel $F(x)$ w.r.t. to $\St(z)$
            & Fiber $F(x)$ w.r.t. $\St(z)$ with $\dist_G(x,z) = 1$.
            \\
        Quasigated $H \subseteq G$
            & $\forall u \in V(G), |\eproj{u}{H}| \le 2$.
            \\
        Star $\St(z)$
            & Union of the all hypercubes of $G$ containing $z$.
            \\
        Total boundary $\partial^* F(x)$
            & $\bigcup_{y \sim x} \partial_{y} F(x)$.
            \\
        \hline
    \end{tabular}
}

\subsection*{Acknowledgements} 

We would like to acknowledge the referees of this paper for careful reading
of the manuscript and numerous suggestions that helped us to improve the readability of the paper. The work on this paper  was supported by ANR project DISTANCIA (ANR-17-CE40-0015).

\newpage

\bibliographystyle{plainurl}
\bibliography{labeling_median_graphs}

\begin{thebibliography}{10}

\bibitem{AbChGaPe}
I.~Abraham, S.~Chechik, C.~Gavoille, and D.~Peleg.
\newblock Forbidden-set distance labels for graphs of bounded doubling
  dimension.
\newblock {\em ACM Trans. Algorithms}, 12:22:1--22:17, 2016.

\bibitem{AbGaGoMa}
I.~Abraham, C.~Gavoille, A.V. Goldberg, and D.~Malkhi.
\newblock Routing in networks with low doubling dimension.
\newblock In {\em 26th {IEEE} International Conference on Distributed Computing
  Systems, {ICDCS}}, page~75. {IEEE} Computer Society, 2006.

\bibitem{AbGh}
A.~Abrams and R.~Ghrist.
\newblock State complexes for metamorphic robots.
\newblock {\em Intl. J. Robotics Res.}, 23:811--826, 2004.

\bibitem{Agol}
I.~Agol.
\newblock The virtual {H}aken conjecture.
\newblock {\em Doc. Math.}, 18:1045--1087, 2013.
\newblock With an appendix by Agol, Daniel Groves, and Jason Manning.
\newblock URL: \url{http://www.emis.de/journals/DMJDMV/vol-18/33.html}.

\bibitem{AlstrupGHP16a}
S.~Alstrup, C.~Gavoille, E.B. Halvorsen, and H.~Petersen.
\newblock Simpler, faster and shorter labels for distances in graphs.
\newblock In {\em SODA}, pages 338--350, 2016.

\bibitem{AlstrupGHP16b}
S.~Alstrup, I.L. G{\o}rtz, E.B. Halvorsen, and E.~Porat.
\newblock Distance labeling schemes for trees.
\newblock In {\em ICALP}, pages 132:1--132:16, 2016.

\bibitem{Bandelt_retracts}
H.-J. Bandelt.
\newblock Retracts of hypercubes.
\newblock {\em J. Graph Theory}, 8:501--510, 1984.

\bibitem{BaCh_survey}
H.-J. Bandelt and V.~Chepoi.
\newblock Metric graph theory and geometry: a survey.
\newblock {\em Contemporary Mathematics}, 453:49--86, 2008.

\bibitem{BaChEp}
H.-J. Bandelt, V.~Chepoi, and D.~Eppstein.
\newblock Ramified rectilinear polygons: coordinatization by dendrons.
\newblock {\em Discr. Comput. Geom.}, 54:771--797, 2015.

\bibitem{BaHe}
H.{-}J. Bandelt and J.~Hedl{\'\i}kov{\'a}.
\newblock Median algebras.
\newblock {\em Discr. Math.}, 45:1--30, 1983.

\bibitem{BaVdV}
H.-J. Bandelt and M.~van~de Vel.
\newblock Embedding topological median algebras in products of dendrons.
\newblock {\em Proc. London Math. Soc.}, s3-58:439 -- 453, 1989.

\bibitem{bazzaro2009localized}
F.~Bazzaro and C.~Gavoille.
\newblock Localized and compact data-structure for comparability graphs.
\newblock {\em Discr. Math.}, 309:3465--3484, 2009.

\bibitem{BeChChVa}
L.~B{\'e}n{\'e}teau, J.~Chalopin, V.~Chepoi, and Y.~Vax{\`e}s.
\newblock {Medians in median graphs and their cube complexes in linear time}.
\newblock In {\em 47th International Colloquium on Automata, Languages, and
  Programming, {ICALP}}, volume 168 of {\em Leibniz International Proceedings
  in Informatics (LIPIcs)}, pages 10:1--10:17, Dagstuhl, Germany, 2020. Schloss
  Dagstuhl--Leibniz-Zentrum f{\"u}r Informatik.
\newblock URL: \url{https://drops.dagstuhl.de/opus/volltexte/2020/12417}, \href
  {http://dx.doi.org/10.4230/LIPIcs.ICALP.2020.10}
  {\path{doi:10.4230/LIPIcs.ICALP.2020.10}}.

\bibitem{BiHoVo}
L.~J. Billera, S.P. Holmes, and K.~Vogtmann.
\newblock Geometry of the space of phylogenetic trees.
\newblock {\em Adv. Appl. Math.}, 27:733--767, 2001.

\bibitem{breuer1967unexpected}
M.~A. Breuer and J.~Folkman.
\newblock An unexpected result in coding the vertices of a graph.
\newblock {\em J. Math. Anal. Appl.}, 20:583--600, 1967.

\bibitem{BrKlSk}
B.~Bre\v{s}ar, S.~Klav\v{z}ar, and R.~\v{S}krekovski.
\newblock On cube--free median graphs.
\newblock {\em Discr. Math.}, 307:345--351, 2007.

\bibitem{BrHa}
M.R. Bridson and A.~Haefliger.
\newblock {\em Metric {S}paces of {N}on-{P}ositive {C}urvature}, volume 319 of
  {\em Grundlehren der mathematischen Wissenschaften}.
\newblock Springer-Verlag, Berlin, 1999.

\bibitem{CC-JCSS}
J.~Chalopin and V.~Chepoi.
\newblock A counterexample to {T}hiagarajan's conjecture on regular event
  structures.
\newblock {\em J. Comput. Syst. Sci.}, 113:76--100, 2020.

\bibitem{ChChHiOs}
J.~Chalopin, V.~Chepoi, H.~Hirai, and D.~Osajda.
\newblock Weakly modular graphs and nonpositive curvature.
\newblock {\em Memoirs of AMS}, (to appear).

\bibitem{Chast}
M.~Chastand.
\newblock Fiber-complemented graphs. {I}. {S}tructure and invariant subgraphs.
\newblock {\em Discr. Math.}, 226:107--141, 2001.

\bibitem{Ch_metric}
V.~Chepoi.
\newblock Classification of graphs by means of metric triangles.
\newblock {\em Metody Diskret. Analiz.}, 49:75--93, 96, 1989.

\bibitem{Ch_CAT}
V.~Chepoi.
\newblock Graphs of some {CAT(0)} complexes.
\newblock {\em Adv. Appl. Math.}, 24:125--179, 2000.

\bibitem{Ch_nice}
V.~Chepoi.
\newblock Nice labeling problem for event structures: a counterexample.
\newblock {\em SIAM J. Comput.}, 41:715--727, 2012.

\bibitem{chepoi2006distance}
V.~Chepoi, F.~F. Dragan, and Y.~Vax{\`e}s.
\newblock Distance and routing labeling schemes for non-positively curved plane
  graphs.
\newblock {\em J. Algorithms}, 61:60--88, 2006.

\bibitem{ChHa}
V.~Chepoi and M.~F. Hagen.
\newblock On embeddings of {CAT(0)} cube complexes into products of trees via
  colouring their hyperplanes.
\newblock {\em J. Comb. Theory, Ser. {B}}, 103:428--467, 2013.

\bibitem{ChLaRa}
V.~Chepoi, A.~Labourel, and S.~Ratel.
\newblock On density of subgraphs of {C}artesian products.
\newblock {\em J. Graph Theory}, 93:64--87, 2020.

\bibitem{ChMa}
V.~Chepoi and D.~Maftuleac.
\newblock Shortest path problem in rectangular complexes of global nonpositive
  curvature.
\newblock {\em Comput. Geom.}, 46:51--64, 2013.

\bibitem{Cormen}
T.~H. Cormen, C.~E. Leiserson, R.~L. Rivest, and C.~Stein.
\newblock {\em Introduction to Algorithms}.
\newblock {MIT} Press, 3rd edition, 2009.

\bibitem{courcelle2003query}
B.~Courcelle and R.~Vanicat.
\newblock Query efficient implementation of graphs of bounded clique-width.
\newblock {\em Discr. Appl. Math.}, 131:129--150, 2003.

\bibitem{dragan2010collective}
F.F. Dragan and C.~Yan.
\newblock Collective tree spanners in graphs with bounded parameters.
\newblock {\em Algorithmica}, 57:22--43, 2010.

\bibitem{fraigniaud2001routing}
P.~Fraigniaud and C.~Gavoille.
\newblock Routing in trees.
\newblock In {\em ICALP}, pages 757--772. Springer, 2001.

\bibitem{FrGaNiWe_trees}
O.~Freedman, P.~Gawrychowski, P.~K. Nicholson, and O.~Weimann.
\newblock Optimal distance labeling schemes for trees.
\newblock In {\em Proceedings of the ACM Symposium on Principles of Distributed
  Computing}, pages 185--194. ACM, 2017.

\bibitem{GaLy_hyperbolic}
C.~Gavoille and O.~Ly.
\newblock Distance labeling in hyperbolic graphs.
\newblock In {\em International Symposium on Algorithms and Computation}, pages
  1071--1079. Springer, 2005.

\bibitem{gavoille2003distance}
C.~Gavoille and C.~Paul.
\newblock Distance labeling scheme and split decomposition.
\newblock {\em Discr. Math.}, 273:115--130, 2003.

\bibitem{gavoille2008optimal}
C.~Gavoille and C.~Paul.
\newblock Optimal distance labeling for interval graphs and related graph
  families.
\newblock {\em SIAM J. Discr. Math.}, 22:1239--1258, 2008.

\bibitem{GavoillePPR04}
C.~Gavoille, D.~Peleg, S.~P{\'{e}}rennes, and R.~Raz.
\newblock Distance labeling in graphs.
\newblock {\em J. Algorithms}, 53:85--112, 2004.

\bibitem{gavoille1996memory}
C.~Gavoille and S.~P{\'e}renn{\`e}s.
\newblock Memory requirement for routing in distributed networks.
\newblock In {\em PODC}, pages 125--133. ACM, 1996.

\bibitem{GawrychowskiU16}
P.~Gawrychowski and P.~Uznanski.
\newblock A note on distance labeling in planar graphs.
\newblock {\em CoRR}, abs/1611.06529, 2016.

\bibitem{GhPe}
R.~Ghirst and Peterson V.
\newblock The geometry and topology of reconfiguration.
\newblock {\em Adv. Appl. Math.}, 38:302--323, 2007.

\bibitem{Gr}
M.~Gromov.
\newblock Hyperbolic groups.
\newblock In S.~M. Gersten, editor, {\em Essays in group theory}, volume~8 of
  {\em Math. Sci. Res. Inst. Publ.}, pages 75--263. Springer, New York, 1987.

\bibitem{Hayashi}
K.~Hayashi.
\newblock A polynomial time algorithm to compute geodesics in {CAT(0)} cubical
  complexes.
\newblock In {\em ICALP}, pages 78:1--78:14, 2018.

\bibitem{kannan1992implicat}
S.~Kannan, M.~Naor, and S.~Rudich.
\newblock Implicit representation of graphs.
\newblock {\em SIAM J. Discr. Math.}, 5:596--603, 1992.

\bibitem{KlMu}
S.~Klavžar and H.M. Mulder.
\newblock Median graphs: characterizations, location theory and related
  structures.
\newblock {\em J. Combin. Math. Combin. Comput.}, 30:103 -- 127, 1999.

\bibitem{knuth2008}
D.E. Knuth.
\newblock {\em The Art of Computer Programming : Vol. 4. Fascicle 0,
  Introduction to combinatorial algorithms and Boolean functions}.
\newblock Boston, Mass. ; London : Addison-Wesley, 2008.

\bibitem{KoRiXi}
G.~Konjevod, A.W. Richa, and D.~Xia.
\newblock Optimal scale-free compact routing schemes in networks of low
  doubling dimension.
\newblock In {\em Proceedings of the Eighteenth Annual {ACM-SIAM} Symposium on
  Discrete Algorithms, {SODA}}, pages 939--948. {SIAM}, 2007.
\newblock URL: \url{http://dl.acm.org/citation.cfm?id=1283383.1283484}.

\bibitem{Mu}
H.M. Mulder.
\newblock {\em The {I}nterval {F}unction of a {G}raph}, volume 132 of {\em
  Mathematical Centre Tracts}.
\newblock Mathematisch Centrum, Amsterdam, 1980.

\bibitem{MulderSch}
H.M. Mulder and A.~Schrijver.
\newblock Median graphs and {H}elly hypergraphs.
\newblock {\em Discr. Math.}, 25:41--50, 1979.

\bibitem{OwPr}
M.~Owen and J.S. Provan.
\newblock A fast algorithm for computing geodesic distances in tree space.
\newblock {\em {IEEE/ACM} Trans. Comput. Biol. Bioinform.}, 8:2--13, 2011.

\bibitem{Peleg00book}
D.~Peleg.
\newblock {\em {D}istributed {C}omputing: {A} {L}ocality-{S}ensitive
  {A}pproach}.
\newblock SIAM, 2000.

\bibitem{Peleg00}
D.~Peleg.
\newblock Proximity-preserving labeling schemes.
\newblock {\em J. Graph Theory}, 33:167--176, 2000.

\bibitem{Peleg05}
D.~Peleg.
\newblock Informative labeling schemes for graphs.
\newblock {\em Theor. Comput. Sci.}, 340:577--593, 2005.

\bibitem{Ro}
M.~Roller.
\newblock Poc sets, median algebras and group actions.
\newblock Technical report, Univ. of Southampton, 1998.

\bibitem{Sa_survey}
M.~Sageev.
\newblock {CAT(0)} cube complexes and groups.
\newblock In M.~Bestvina, M.~Sageev, and K.~Vogtmann, editors, {\em Geometric
  Group Theory}, volume~21 of {\em IAS/Park City Mathematics Series}, pages
  6--53. AMS, IAS, 2012.

\bibitem{Scha}
T.J. Schaefer.
\newblock The complexity of satisfiability problems.
\newblock In {\em STOC}, pages 216--226, 1978.

\bibitem{Talwar}
K.~Talwar.
\newblock Bypassing the embedding: algorithms for low dimensional metrics.
\newblock In {\em Proceedings of the 36th Annual {ACM} Symposium on Theory of
  Computing, {STOC}}, pages 281--290. {ACM}, 2004.
\newblock URL: \url{https://doi.org/10.1145/1007352.1007399}, \href
  {http://dx.doi.org/10.1145/1007352.1007399}
  {\path{doi:10.1145/1007352.1007399}}.

\bibitem{thorup2001compact}
M.~Thorup and U.~Zwick.
\newblock Compact routing schemes.
\newblock In {\em SPAA}, pages 1--10. ACM, 2001.

\bibitem{Winkler1983}
P.~M. Winkler.
\newblock Proof of the squashed cube conjecture.
\newblock {\em Combinatorica}, 3:135--139, 1983.

\bibitem{WiNi}
G.~Winskel and M.~Nielsen.
\newblock Models for concurrency.
\newblock In S.~Abramsky, Dov~M. Gabbay, and T.~S.~E. Maibaum, editors, {\em
  Handbook of Logic in Computer Science (Vol. 4)}, pages 1--148. Oxford
  University Press, 1995.

\end{thebibliography}

\newpage
\section*{Appendices}
\addcontentsline{toc}{section}{Appendices}
\setcounter{tocdepth}{0}
\appendix

\section{Fibers in median graphs}

In this section we give the proofs of the well-known properties of median graphs stated in Section~\ref{sect_fibers}.

\begin{proof}[Proof of Lemma~\ref{thm_quadrangle}]
    Let $x$ be the median of the triplet $u,v,w$. Then $x$ must be
    adjacent to $v$ and $w$. Since $x\in I(u,v)\cap I(u,w)$, necessarily
    $d_G(u,x)=k-1$. Since any vertex $x'$ adjacent to $v,w$ and having distance 
    $k-1$ to $u$ is a median of $u,v,w$, we conclude that $x'=x$, concluding 
    the proof.
\end{proof}

\begin{proof}[Proof of Lemma~\ref{thm_convex=gated}] 
    Obviously, any gated set is convex and any convex set is connected and 
    locally-convex.  Assume that $A$ is convex but not gated. Then there exists 
    a vertex $u \in V \setminus A$ which does not have a gate in $A$. Let $x$ 
    be a closest to $u$ vertex of $A$. Since $x$ is not the  gate of
    $u$, there exists a vertex $y\in A$ such that $x\notin I(u,y)$. Let $m$ be 
    the median of the triplet $u,x,y$.
    Since $x\notin I(u,y)$, $m\ne x$. Since $m\in I(x,y)$ and $H$ is convex,
    $m$ belongs to $A$.
    Since $m\in I(x,u)$ and $m\ne x$, $d_G(u,m)<d_G(u,x)$, contrary to the
    choice of $x$.

    Finally suppose that $A$ is connected and locally-convex and we will show that $A$ is convex.
    Let $u$ and $v$ be any two vertices of $A$. We show
    that $I(u,v)\subseteq A$
    by induction on the distance $d_H(u,v)$ between $u$ and $v$ in $H$. If
    $d_H(u,v)=2$, then the property holds by
    local convexity of $A$. Let $d_H(u,v)=k\ge 3$ and suppose that
    $I(u',v')\subseteq A$ for any two
    vertices $u',v'\in A$ such that $d_H(u',v')\le k-1$. Pick any vertex $x\in
    I(u,v)$.
    Let $u'$ be the neighbor of $u$ on a shortest $(u,v)$-path of $G$ passing
    via $x$. Let also $u''$ be the neighbor of $u$ on a
    shortest $(u,v)$-path of $H$. Since $d_H(u'',v)=k-1$, by induction
    hypothesis, $I(u'',v)\subset A$. Since $G$ is bipartite and $u\sim u''$,
    $|d_G(u,v)-d_G(u'',v)|=1$. If $d_G(u'',v)>d_G(u,v)$, then $x\in
    I(u,v)\subset I(u'',v)\subset A$ and we are done.
    Now, let $d_G(u'',v)=d_G(u,v)-1=d_G(u',v)$ and $d_H(u,v)=d_G(u,v)$. By
    quadrangle condition there exists a vertex $z\sim u',u''$
    at distance $k-2$ from $v$. Since $z\in I(u'',v)\subset A$ and $u'\sim
    u,z$, by local convexity of $A$ we deduce that $u'$ belongs to $A$. Since
    $d_H(u',v)=d_G(u'',v)=k-1$, by induction hypothesis, $I(u',v)\subset A$.
    Since $x\in I(u',v)$, $x$ belongs to $A$ and we are done.
\end{proof}

\begin{proof}[Proof of Lemma~\ref{convex-interval}] 
    In view of Lemma \ref{thm_convex=gated} it suffices to show that  each 
    interval $I(u,v)$ of $G$ is locally-convex. Let $x,y\in I(u,v)$ with 
    $\dist_G(x,y)=2$ and let $z$ be a common neighbor of $x$ and $y$. Suppose 
    without loss of generality that $\dist_G(u,x)\le \dist_G(u,v)$. 
    Since $G$ is bipartite, either $\dist_G(u,x)=\dist_G(u,y)$ or 
    $\dist_G(u,x)=\dist_G(u,y)+2$. In the second case, obviously $z\in 
    I(x,y)\subseteq I(u,v)$.  In the first case, let $z',z''$ be the medians of 
    the triplets $u,x,y$ and $v,x,y$, respectively. 
    Then $z'\in I(u,x)\subseteq I(u,v)$ and $z''\in I(x,v)\subseteq I(u,v)$. 
    Therefore, if  $z$ coincide with one of the vertices $z',z''$, then we are 
    done. 
    Otherwise, the vertices $x,y,z,z',z''$ induce a forbidden $K_{2,3}$. 
\end{proof}

\begin{proof}[Proof of Lemma~\ref{prop_St(m)_convex}]  
    We will only sketch the proof (for a complete proof, see Theorem 6.17 of 
    \cite{ChChHiOs} and its proof for a more general class of graphs). By
    Lemma \ref{thm_convex=gated} it suffices to show that $\St(z)$ is locally 
    convex. Let $x,y\in \St(z)$ be two vertices at distance two and let $v\sim 
    x,y$.
    Then $Q_x=I(x,z)$ and $Q_y=I(y,z)$ are two cubes of $\St(z)$.
    We can suppose without loss of generality that $v\notin Q_x\cup Q_y$. This
    implies that
    $x,y\in I(z,v)$, i.e., we can suppose that $d_G(x,z)=d_G(y,z)=k$ and
    $d_G(z,v)=k+1$. By quadrangle condition,
    there exists $u$ such that $d_G(u,z)=k-1$ and $u\sim x,y$. Necessarily
    $I(u,z)$ is a $(k-1)$-cube $Q_u$ included in the $k$-cubes $Q_x$ and $Q_y$.
    Therefore $z$ has a neighbor
    $x'$  such that $I(x,x')$ is a $(k-1)$-cube disjoint from $Q_u$ and which
    together with $Q_u$ gives $Q_x$.
    Analogously, $z$ has a neighbor $y'$ such that $I(y,y')$ is a $(k-1)$-cube
    disjoint from $Q_u$
    and which together with $Q_u$ gives $Q_y$. By quadrangle condition there
    exists $v'\sim x',y'$ at
    distance $k-1$ to $v$. Then one can show that $I(v,v')$ induces a
    $(k-1)$-cube, which together with the
    $k$-cubes $Q_x$ and $Q_y$ define the $(k+1)$-cube $Q_v=I(v,z)$. This
    establishes that $v$ belongs to $\St(z)$.
\end{proof}

\begin{proof}[Proof of Lemma~\ref{fiber-gated1}] 
    Each fiber $F(x)$ induces a connected subgraph of $G$, thus it suffices to 
    show that $F(x)$ is locally convex. Pick $u,v\in F(x)$ with $d_G(u,v)=2$ 
    and let $z$ be any common neighbor of $u$ and $v$. Suppose by way of 
    contradiction that $z\in F(y)$ for $y\in V(H), y\ne x$.
    Then $x\in I(u,y)\cap I(v,y)$ and $y\in I(z,x)$. This implies in particular
    that $x\sim y$, $u,v\in I(z,x)$,
    and $z\in I(u,y)\cap I(v,y)$. By quadrangle condition, there exists $x'\sim
    u,v$, one step closer to $x$
    than $u$ and $v$. Then $z,x'\in I(u,y)$ and by quadrangle condition there
    exists a vertex $y'\sim x',z$
    one step closer to $y$ than $x'$ and $z$. But then the vertices
    $u,v,z,x',y'$ induce a $K_{2,3}$, which
    is a forbidden subgraph of median graphs.
\end{proof}

\begin{proof}[Proof of Lemma~\ref{lem_gated_borders}] 
    If $x\sim y$, then $F(x)\sim F(y)$. Conversely, let $F(x)\sim F(y)$, i.e., 
    there exists an edge $x'y'$ of $G$ such that $x'\in F(x)$ and $y'\in F(y)$. 
    Since $F(x)$ and $F(y)$ are convex and $G$ is bipartite, necessarily $x'\in 
    I(y',x)$ and $y'\in I(x',y)$.
    Since $x'\in F(x), y'\in F(y)$ and $H$ is gated, we deduce that $x\in
    I(x',y)$ and $y\in I(y',x)$. From this we conclude that 
    $d_G(x',x)=d_G(y',y)$
    and that $d_G(x,y)=d_G(x',y')=1$. This establishes the first assertion.

    To prove the second assertion, let $F(x)\sim F(y)$ and we have to prove
    that $\partial_yF(x)$ is gated. By induction on $k=d_G(x,x')$, we can
    show that $I(x',x)\subseteq \partial_yF(x)$ for any vertex $x'$ of
    $\partial_yF(x)$. For we show that any neighbor $x''$ of
    $x'$ in $I(x',x)$ belongs to $\partial_yF(x)$.
    Let $y'$ be the neighbor of $x'$ in $\partial_xF(y)$. Then $x'',y'\in
    I(x',y)$, $x'',y'\sim x'$, and $d_G(x',y)=k+1$, thus by quadrangle
    condition there
    exists a vertex $y''\sim y',x''$ at distance $k-1$ from $y$. Since $y''\in
    I(y',y)\subset F(y)$, we conclude that $x''\in \partial_yF(x)$. Thus
    $I(x',x)\subseteq \partial_y F(x)$, yielding that the subgraph induced by
    $\partial_y F(x)$ is connected.

    By Lemma \ref{thm_convex=gated} it remains to show that $\partial_y F(x)$
    is locally-convex. Pick $x',x''\in \partial_y F(x)$ at distance two and let
    $u\sim x',x''$.
    Since $F(x)$ is convex, $u\in F(x)$. Let $y'$ and $y''$ be the neighbors of
    $x'$ and $x''$, respectively, in $F(y)$. Let $v$ be the gate of $u$ in
    $F(y)$.
    Since $d_G(u,y')=d_G(u,y'')=2$ (because $G$ is bipartite) and $v\in
    I(u,y')\cap
    I(u,y'')$, we conclude that $v$ is adjacent to $u,y',$ and $y''$.
    Hence $v\in I(y',y'')\subset F(y)$, yielding $u\in \partial_y F(x)$. This
    finishes the proof that $\partial_y F(x)$ is gated.

    If $\dim(\partial_y F(x))=\dim(G)=d$, then $\partial_y F(x)$ contains a 
    $d$-dimensional cube $Q'$. Any vertex $x'$ of $\partial_y F(x)$ is adjacent 
    to a vertex $y'$ of $F(y)$. Clearly, $y'$ must belong to $\partial_x F(y)$.
    Since $\partial_x F(y)$ is gated, $y'$ is the unique neighbor of $x'$ in $\partial_x F(y)$. Let $Q''$ denote the subgraph of $\partial_x F(y)$ induced by the neighbors $y'$
    of vertices  $x'$ of $Q'$. We assert that $Q''$ is a $d$-cube. Indeed, if $x'x''$ is an edge of $Q'$, then since $\partial_x F(y)$ is gated we conclude that the neighbors
    $y'$ and $y''$ of $x'$ and $x''$ must be adjacent, i.e., $Q''$ is a $d$-cube.  Then  $Q'$ and  $Q''$
    induce a $(d+1)$-cube of $G$.
    Thus $\dim(\partial_y F(x))\le d-1$.
\end{proof}

\section{Detailed description of the routing scheme}
\label{sect_rout_encoding}

We present a detailed description of the labeling routing scheme.
Let $G=(V,E)$ be a cube-free median graph and let $u$ be any vertex of $G$. Let $i$
be any step of the algorithm \routenc{} applied to $G$ and let $c$ be a centroid
of the current median subgraph containing $u$ at step $i$. The ``\st'' part $\LR_i^{\st}(u)$ 
of the label of $u$ is composed of the identifier of $c$, a port from $u$ to $c$,
a port from $c$ to $u$, and the identifier of gate $x:=u^\downarrow$ of $u$ to $\St(c)$ (i.e., of the fiber containing $u$).
Note that $c$ cannot store the ports to other vertices in order to answer
routing queries from $c$. This is why the label of $u$
contains the port $\LR_i^{\st[\fromMed]}(u)$ from $c$ to $u$. Here are the components of $\LR_i^{\st}(u)$:
\begin{enumerate}[{\hskip1em(}1{)}]
    \item $\LR_i^{\st[\med]}(u) := \id(c)$ is the unique identifier of $c$;
    \item $\LR_i^{\st[\toMed]}(u)$ consists of a port to take from $u$ in
    order to reach $c$;
    \item $\LR_i^{\st[\fromMed]}(u)$ consists of a port to take from
    $c$ in order to reach $u$;
    \item $\LR_i^{\st[\rootB]}(u)$ contains the identifier of the \fiber
    containing $u$ (i.e., the star labeling of $u^\downarrow$).
\end{enumerate}

The \lleft and \rright parts of the label of $u$ contain similar information but
they depend of whether $u$ belongs to a panel or to a cone.  If $u$ belongs to a \panel $F(x)$ (recall that $x=u^\downarrow$),
then $\LR_i^{\lleft}(u)$ is composed of the following four components:
\begin{enumerate}[{\hskip1em(}1{)}]
    \item $\LR_i^{\lleft[\imprintD]}(u)$ is the tree distance labeling of the first imprint $u_1$ of $u$ on  $\partial^* F(x)$;
    \item $\LR_i^{\lleft[\imprintR]}(u)$ is the tree routing labeling of $u_1$ in the tree
    $\partial^* F(x)$;
    \item $\LR_i^{\lleft[\toImprint]}(u)$ is $\port(u,u_1)$;
    \item $\LR_i^{\lleft[\distance]}(u)$ is the distance $\dist_G(u,u_1)$.
\end{enumerate}
The \rright part $\LR_i^{\rright}(u)$ of the label of $u$ is defined in a
similar way with respect to the second imprint $u_2$ of $u$ on $\partial^*
F(x)$.

If $u$ belongs to a \cone $F(x)$, then $F(x)$ has two neighboring panels
$F(w_1)$ and $F(w_2)$. 
The components $\LR_i^{\lleft}(u)$ and $\LR_i^{\rright}(u)$ of the \lleft and 
\rright parts of the label of $u$, each consists of four components. For 
example, $\LR_i^{\lleft}(u)$ is composed of the following data:
\begin{enumerate}[{\hskip1em(}1{)}]
    \item $\LR_i^{\lleft[\gateD]}(u)$ consists of a tree distance labeling of
    the gate $u^+_1$ of $u$ in the panel $F(w_1)$; 
    \item $\LR_i^{\lleft[\gateR]}(u)$ is a tree routing labeling of $u^+_1$ in the tree $\partial^* F(w_1)$;
    \item $\LR_i^{\lleft[\toGate]}(u)$ contains the $\port(u,u^+_1)$;
    \item $\LR_i^{\lleft[\fromGate]}(u)$ is the port
    $\port(u^+_1,\twin(u^+_1))$
    from $u^+_1$ to $\twin(u^+_1)$.
\end{enumerate}
The \rright part $\LR_i^{\rright}(u)$ of the label of $u$ is defined in a similar way with respect to the gate  $u^+_2$ of $u$ in the panel $F(w_2)$.
We assume that no port is given the number $0$. If \routdec[] returns $0$ or if
a label stores a port equal to $0$, it means that there is no need to move. 
Here is the encoding algorithm:

\medskip
\scalebox{0.91}{\begin{algorithm}[H]
        \caption{\label{alg_routenc}\routenc[]($G$, $\LR(V)$)}
        \Input{
            A cube-free median graph $G=(V,E)$ and a labeling $\LR(V)$
            initially consisting on a unique identifier $\id(v)$ for every $v
            \in V$
        }

        \lIf{$V = \{v\}$}{ \Stop }

        \BlankLine
        Find a centroid $c$ of $G$ \;
        $\Lb{\St(c)}{\St(c)}$ $\leftarrow$ \encStar($\St(c)$) \;
        ${\mathcal F}_c \leftarrow \{ F(x): x \in \St(c) \}$\;

        \ForEach{\panel $F(x) \in {\mathcal F}_c$}{
            $\LDt{\partial^* F(x)}{\partial^* F(x)}$ $\leftarrow$ \distEncTree($\partial^* F(x)$) \;
            $\LRt{\partial^* F(x)}{\partial^* F(x)}$ $\leftarrow$ \routEncTree($\partial^* F(x)$) \;
            \ForEach{$u \in F(x)$}{
                Find the gate $u^\downarrow$ of $u$	in $\St(c)$ \;
                Find the two imprints  $u_1$ and $u_2$ of $u$ on $\partial^* F(x)$ \;
                $(d_1, ~d_2) \leftarrow (\dist_G(u,u_1), ~\dist_G(u,u_2))$ \;

                $L_\st \leftarrow (\id(c), \port(u,c), \port(c,u),
                \Lb{u^\downarrow}{\St(c)})$ \;
                $L_{\lleft} \leftarrow (\LDt{u}{\partial^* F(x)}, \LRt{u}{\partial^* F(x)},
                \port(u,u_1),
                \dist_G(u,u_1))$ \;
                $L_{\rright} \leftarrow (\LDt{u}{\partial^* F(x)}, \LRt{u}{\partial^* F(x)},
                \port(u,u_2),
                \dist_G(u,u_2))$ \;

                $\LR(u) \leftarrow \LR(u) \conc
                (L_\st,L_{\lleft},L_{\rright})$ \;
            }
            \routenc[]($F(x)$, $\LR(V)$) \;
        }
        \ForEach{\cone $F(x) \in {\mathcal F}_c$}{
            Let $F(w_1)$ be the \lleft \panel neighboring  $F(x)$
            \;
            Let $F(w_2)$ be the \rright \panel neighboring $F(x)$ \;

            $\LDt{\partial^* F(w_1)}{\partial^* F(w_1)}, \LDt{\partial^* F(w_2)}{\partial^* F(w_2)}$ $\leftarrow$
            \distEncTree($\partial^* F(w_1)$), \distEncTree($\partial^*
            F(w_2)$)
            \;
            $\LRt{\partial^* F(w_1)}{\partial^* F(w_1)}, \LRt{\partial^* F(w_2)}{\partial^* F(w_2)}$ $\leftarrow$
            \routEncTree($\partial^* F(w_1)$), \routEncTree($\partial^* F(w_2)$) \;
            \ForEach{$u \in F(x)$}{
                Find the gate $u^\downarrow$ of $u$	in $\St(c)$ \;
                Find the gate $u^+_1$ of $u$ in $F(w_1)$ and let $\twin (u^+_1)$ be the twin
                of $u_1^+$ in $F(x)$\;
                Find the gate $u^+_2$ of $u$ in $F(w_2)$ and let $\twin(u^+_2)$ be the twin
                of $u^+_2$ in $F(x)$ \;

                $L_\st \leftarrow (\id(c), \port(u,c), \port(c,u),
                \Lb{u^\downarrow}{\St(c)})$ \;
                $L_{\lleft} \leftarrow (
                    \LDt{u}{\partial^* F(w_1)},
                    \LRt{u}{\partial^* F(w_1)},
                    \port(u,u^+_1),
                    \port(u^+_1,\twin(u^+_1)))$ \;
                $L_{\rright} \leftarrow (
                    \LDt{u}{\partial^* F(w_2)},
                    \LRt{u}{\partial^* F(w_2)},
                    \port(u,u^+_2),
                    \port(u^+_2,\twin(u^+_2)))$ \;
                 $\LR(u) \leftarrow \LR(u) \conc
                (L_\st,L_{\lleft},L_{\rright})$ \;
            }
            \routenc[]($F(x)$, $\LR(V)$) \;
        }
\end{algorithm}}

\subsection{Routing queries}
\label{sect_rout_queries}

Let $u$ and $v$ be two arbitrary vertices of a cube-free median graph $G$ and let $\LR(u)$ and $\LR(v)$ be their labels returned by the encoding
algorithm \routenc{}. We describe how  the routing algorithm \routdec{} can decide by which port to send the message from $u$ to $v$
to a neighbor of $u$ closer to $v$ than $u$.

\subsubsection{The algorithm}
\label{sect_rout_algo}

We continue with the formal description of the routing algorithm \routdec{}. 
The specific functions ensuring routing from panel to cone, from cone to panel, 
from cone to cone, or between \separated vertices
will be described in the next subsection.
$~$

\medskip
\scalebox{0.91}{\begin{algorithm}[H]
        \caption{\label{alg_routdec}\routdec[]($\LR(u)$, $\LR(v)$)}
        \Input{
            The labels $\LR(u)$ and $\LR(v)$ of two vertices $u$ and $v$ of $G$,
            where $u$ is the source and $v$ the target}
        \Output{$\port(u,v)$}

        \BlankLine
        \lIf{$\LR_0(u) = \LR_0(v)$ \tcc*[h]{$u = v$}}{
            \Return $0$
        }

        \BlankLine
        Let $i$ be the highest integer such that $\LR_i^{\st[\med]}(u) = 
        \LR_i^{\st[\med]}(v)$\;

        \BlankLine
        $d \leftarrow \distDecStar(\LR_i^{\st[\rootB]}(u),
        \LR_i^{\st[\rootB]}(v))$
        \tcp*{$\dist_G(u^\downarrow,v^\downarrow)$}
        $d_u \leftarrow \distDecStar(\LR_i^{\st[\rootB]}(u), 0)$
        \tcp*{$\dist_G(u^\downarrow,c)$}
        $d_v \leftarrow \distDecStar(\LR_i^{\st[\rootB]}(v), 0)$
        \tcp*{$\dist_G(v^\downarrow,c)$}

        \BlankLine
        \lIf{$d = 1$ and $d_u = 1$\hskip2em$~$}{
            \Return \RoutingNeighboringPtoC($\LR_i(u)$, $\LR_i(v)$)
        }
        \lIf{$d = 1$ and $d_v = 1$\hskip2em$~$}{
            \Return \RoutingNeighboringCtoP($\LR_i(u)$, $\LR_i(v)$)
        }
        \lIf{$d = 2$ and $d_u = d_v = 2$}{
            \Return \RoutingConetoCone($\LR_i(u)$, $\LR_i(v)$)
        }

        \Return \RoutingSeparated($\LR_i(u)$, $\LR_i(v)$, $\LR_0(u)$). \
\end{algorithm}}

\subsubsection{Description and functions}
\label{sect_rout_functions}

As for distance queries, the first thing to do in order to answer a routing
query from $u$ to $v$ is to detect if $u$ and $v$ are
\neighboring, \aNeighboring or \separated, and the type (\cone or \panel) of
the fibers containing them. This is done in the same way as explained in Section
\ref{sect_dist_queries}. Again, we assume that $i$ is the first step such that
$u$ and $v$ are no longer  \roommates. Denote by $c$ the median
vertex used at this step. We denote by $F(x)$ the \fiber containing $u$ and by
$F(y)$ the fiber  containing $v$ (recall that $x$ is the gate of $u$ in $\St(c)$ and $y$ is the gate of $v$ in $\St(c)$).

\medskip\noindent
If $u$ and $v$ are \neighboring, the answer is computed differently when the
source $u$ is in a \cone and when $u$ is in a \panel. If $u$ is in a \cone $F(x)$ (and thus $v$ is in a panel $F(y)$),
we use the function \RoutingNeighboringCtoP. This function determines which part ($\LR_i^\lleft(v)$ or
$\LR_i^\rright(v)$) of $\LR_i(v)$ contains the information about the gate $u^+$
of $u$ on $F(y)$. Then the function returns the
port $\port(u,u^+)$ to the gate $u^+$ of $u$ in $F(y)$, stored as
$\LR_i^{\lleft[\toGate]}(u)$ or $\LR_i^{\rright[\toGate]}(u)$.

\medskip
\scalebox{0.91}{\begin{myFunction}
        \newcommand{\dir}{\text{dir}}
        \Fn{\RoutingNeighboringCtoP{$\LR_i(u)$, $\LR_i(v)$}}{

            $\dir \leftarrow \rright$ \;
            \If{$\LR_i^{\st[\rootB]}(v) = \min\{i : i \in
                \LR_i^{\st[\rootB]}(u)\}$}{
                $\dir \leftarrow \lleft$ \;	
            }
            \Return $\LR_i^{\dir[\toGate]}(u)$. \
        }
\end{myFunction}}

\medskip\noindent
If $F(x)$ is a panel and $F(y)$ is a cone, then $u$ stored the distances to its two imprints $u_1$ and $u_2$ on the total boundary
$\partial^* F(x)$ and $v$ stored the distance to its gate $v^+$ in $F(x)$ ($v^+$ also belongs to $\partial^* F(x)$) and its twin $\twin(v^+)$ in $F(y)$.
When $u$ is different from $v^+$, the function \RoutingNeighboringPtoC finds the tree distance labeling
of $v^+$, computes $\min\{\dist_G(u,u_1) + \dist_T(u_1,v^+), \dist_G(u,u_2) +
\dist_T(u_2,v^+)\}$, and returns the port to the imprint  of $u$ minimizing the two distance sums.
If $u$ belong to the total boundary $\partial^* F(x)$, then we distinguish two cases. If $u=v^+$, then
using the label $\LR_i(v)$ of $v$ the algorithm returns the port from 
$\twin(v^+)$ to $v^+=u$.  If $u$ belongs to $\partial^* F(x)$ but $u\ne v^+$,
since $\LR_i(u)$ and $\LR_i(v)$ contain a labeling for routing in trees of $u$ 
and $v^+$, \RoutingNeighboringPtoC computes $\port(u,v^+)$ using the
routing decoder for trees and returns it.

\medskip
\scalebox{0.91}{\begin{myFunction}
        \newcommand{\dir}{\text{dir}}
        \Fn{\RoutingNeighboringPtoC{$\LR_i(u)$, $\LR_i(v)$}}{

            $\dir_v \leftarrow \rright$ \;
            \If{$\LR_i^{\st[\rootB]}(u) = \min\{i : i \in
                \LR_i^{\st[\rootB]}(v)\}$}{
                $\dir_v \leftarrow \lleft$ \;	
            }

            \BlankLine
            \If{$\LR_i^{\lleft[\toGate]}(u) = 0$ or
                $\LR_i^{\rright[\toGate]}(u) =
                0$ \tcc*[h]{$u$ is on the
                    border}}{
                Let $\dir_u \in \{\lleft,\rright\}$ be such that
                $\LR_i^{\dir_u[\toGate]}(u) = 0$ \;
                \If{$\LR_i^{\dir_u[\gateD]}(u) =
                    \LR_i^{\dir_v[\gateD]}(v)$}{
                    \Return $\LR_i^{\dir_v[\fromGate]}(v)$ \tcp*[l]{$u$ is the
                    gate of $v$ on the \panel $F(x)$}
                }
                \Return \routDecTree($\LR_i^{\dir_u[\gateR]}(u)$,
                $\LR_i^{\dir_v[\gateR]}(v)$) \;
            }

            \BlankLine
            $d_\lleft \leftarrow \LR_i^{\lleft[\distance]}(u) +
            \distDecTree(\LR_i^{\lleft[\gateD]}(u),
            \LR_i^{\dir_v[\gateD]}(v))$ \;
            $d_\rright \leftarrow \LR_i^{\rright[\distance]}(u) +
            \distDecTree(\LR_i^{\rright[\gateD]}(u),
            \LR_i^{\dir_v[\gateD]}(v))$ \;

            \BlankLine
            $\dir_u \leftarrow \lleft$ \;
            \If{$d_\rright = \min\{d_\lleft, d_\rright\}$}{
                $\dir_u \leftarrow \rright$ \tcp*[l]{$u_\rright$ is the gate of
                    $u$
                    on a shortest path to $v$}
            }
            \Return $\LR_i^{\dir_u[\toGate]}(u)$. \
        }
\end{myFunction}}

\medskip\noindent
If $u$ and $v$ are \aNeighboring, then $F(x)$ and $F(y)$ are cones and the
function  \RoutingConetoCone is similar to the function
\RoutingNeighboringCtoP. The common \panel $F(w)$ neighboring $F(x)$ and $F(y)$
can be found by inspecting
$\LR_i^{\st[\rootB]}(u)$ and $\LR_i^{\st[\rootB]}(v)$. As in the case of \RoutingNeighboringCtoP, the function \RoutingANeighboring returns the port
$\port(u,u^+)$ from $u$ to its gate $u^+$ in $F(w)$. 

\medskip
\scalebox{0.91}{\begin{myFunction}
        \newcommand{\dir}{\text{dir}}
        \Fn{\RoutingConetoCone{$\LR_i(u)$, $\LR_i(v)$}}{
            $\dir \leftarrow \rright$ \;
            \If{$\LR_i^{\st[\rootB]}(u) \cap
                \LR_i^{\st[\rootB]}(v) =
                \min\{i : i \in
                \LR_i^{\st[\rootB]}(u)\}$}{
                $\dir \leftarrow \lleft$ \;
            }

            \Return $\LR_i^{\dir[\toGate]}(u)$. \
        }
\end{myFunction}}

\medskip\noindent
Finally, if $u$ and $v$ are \separated, two cases have to be considered
depending of whether $u$ is the centroid $c$ or not. If $u$ is not the centroid, then $u$
stored $\port(u,c)$. Since a shortest path from $u$ to $v$ passes via $c$,
\RoutingSeparated returns $\port(u,c)$. If $u$ coincides with $c$, the port
$\port(c,v)$ is not stored in $\LR_i(u)$ but in
$\LR_i^{\st[\fromMed]}(v)$, and \RoutingSeparated returns it.

\medskip
\scalebox{0.91}{\begin{myFunction}
        \Fn{\RoutingSeparated{$\LR_i(u)$, $\LR_i(v)$, $\id(u)$}}{
            \If{$\LR_i^{\st[\med]}(v) = \id(u)$}{
                \Return $\LR_i^{\st[\fromMed]}(v)$ \;
            }
            \Return $\LR_i^{\st[\toMed]}(u)$. \
        }
\end{myFunction}}

\end{document}